\documentclass[a4paper,onecolumn,10pt,accepted=2021-08-01]{quantumarticle}
\pdfoutput=1
\usepackage{amssymb}
\usepackage[utf8]{inputenc}
\usepackage[english]{babel}
\usepackage[T1]{fontenc}
\usepackage{amsfonts}
\usepackage{amsmath}
\usepackage{amsthm}
\usepackage{stmaryrd}
\usepackage{physics}
\usepackage{braket}
\usepackage{dsfont}
\usepackage{bbold}
\usepackage{graphicx}
\usepackage{relsize}
\usepackage{enumerate}
\usepackage[margin=1in]{geometry}
\usepackage{enumitem}
\usepackage{mathtools}
\usepackage{bbm}
\usepackage{graphbox}
\usepackage{comment}
\usepackage{float}
\usepackage[font=small]{caption}
\usepackage{booktabs}
\usepackage[table,xcdraw]{xcolor}
\usepackage{hyperref}
\hypersetup{colorlinks = true, citecolor = red}
\renewcommand{\epsilon}{\varepsilon}
\renewcommand{\phi}{\varphi}
\newcommand{\overbar}[1]{\mkern 1.5mu\overline{\mkern-1.5mu#1\mkern-1.5mu}\mkern 1.5mu}
\newcommand{\CLDUI}{\mathsf{CLDUI}}
\newcommand{\LDUI}{\mathsf{LDUI}}
\newcommand{\LDOI}{\mathsf{LDOI}}
\newcommand{\CDUC}{\mathsf{CDUC}}
\newcommand{\DUC}{\mathsf{DUC}}
\newcommand{\DOC}{\mathsf{DOC}}
\newcommand{\EWP}{\mathsf{EWP}}
\newcommand{\PSD}{\mathsf{PSD}}
\newcommand{\PCP}{\mathsf{PCP}}
\newcommand{\TCP}{\mathsf{TCP}}
\newcommand{\SEP}{\mathsf{SEP}}
\newcommand{\MLDUI}[1]{\mathcal{M}_{#1}(\mathbb{C})^{\times 2}_{\mathbb{C}^{#1}}}
\newcommand{\MLDOI}[1]{\mathcal{M}_{#1}(\mathbb{C})^{\times 3}_{\mathbb{C}^{#1}}}
\newcommand{\M}[1]{\mathcal{M}_{#1}(\mathbb{C})}
\newcommand{\C}[1]{\mathbb{C}^{#1}}
\newcommand{\Mreal}[1]{\mathcal{M}_{#1}(\mathbb{R})}
\newcommand{\T}[1]{\mathcal{T}_{#1}(\mathbb{C})}
\newcommand{\Msa}[1]{\mathcal{M}^{sa}_{#1}(\mathbb{C})}

\newtheorem{theorem}{Theorem}[section]
\newtheorem{definition}[theorem]{Definition}
\newtheorem*{definition*}{Definition}
\newtheorem{proposition}[theorem]{Proposition}
\newtheorem{remark}[theorem]{Remark}
\newtheorem{corollary}[theorem]{Corollary}
\newtheorem{lemma}[theorem]{Lemma}

\newtheorem*{conjecture*}{Conjecture}

\theoremstyle{definition}
\newtheorem{example}{Example}[section]

\definecolor{darkgreen}{rgb}{0,0.392,0}

\begin{document}

\title{Diagonal unitary and orthogonal symmetries {\color{white}{blablabl}} in quantum theory}

\author{Satvik Singh}
\email{  satviksingh2@gmail.com}
\affiliation{Department of Physical Sciences, Indian Institute of Science Education and Research (IISER) Mohali, Punjab, India.}

\author{Ion Nechita}
\email{nechita@irsamc.ups-tlse.fr}
\affiliation{Laboratoire de Physique Th\'eorique, Universit\'e de Toulouse, CNRS, UPS, France.}

\begin{abstract}
   We analyze bipartite matrices and linear maps between matrix algebras, which are respectively, invariant and covariant, under the diagonal unitary and orthogonal groups' actions. By presenting an expansive list of examples from the literature, which includes notable entries like the Diagonal Symmetric states and the Choi-type maps, we show that this class of matrices (and maps) encompasses a wide variety of scenarios, thereby unifying their study. We examine their linear algebraic structure and investigate different notions of positivity through their convex conic manifestations. In particular, we generalize the well-known cone of completely positive matrices to that of triplewise completely positive matrices and connect it to the separability of the relevant invariant states (or the entanglement breaking property of the corresponding quantum channels). For linear maps, we provide explicit characterizations of the stated covariance in terms of their Kraus, Stinespring, and Choi representations, and systematically analyze the usual properties of positivity, decomposability, complete positivity, and the like. We also describe the invariant subspaces of these maps and use their structure to provide necessary and sufficient conditions for separability of the associated invariant states.
\end{abstract}

\maketitle

\tableofcontents

\section{Introduction}

In Quantum Theory, the set of entangled bipartite states has a very special role: the presence of genuinely quantum correlations between the two parties allow for significant advantage in various information processing and computing tasks. However, this paradigm suffers from a significant theoretical hurdle because of the computational hardness of deciding whether a bipartite (or multipartite) quantum state is separable or entangled \cite{gurvits2003classical}. To overcome the NP-hardness of this decision problem in the most general case, a plethora of entanglement (resp.~separability) criteria have been discovered and studied: there are computationally efficient methods for certifying that a given quantum state is entangled (resp.~separable). The most useful entanglement criterion is that of the \emph{positive partial transposition} (PPT) \cite{peres1996separability,horodecki1996separability}, which turns out to be exact for qubit-qubit and qubit-qutrit systems. 

Another way to tackle the separability problem is to restrict the decision problem to special classes of bipartite quantum states. Quantum states satisfying some sort of symmetry are the obvious candidates, since they are physically relevant and mathematically tractable, because of their added structure and reduced number of parameters. One canonical situation is that of bipartite quantum states invariant under the tensor product of the standard representation of the unitary group; two cases need to be considered here: 
$$\forall U \in \mathcal U(d),\qquad \rho = (U \otimes U)\rho (U \otimes U)^* \quad \text{ or } \quad \rho = (U \otimes \bar U)\rho (U \otimes \bar U)^*.$$
These families of quantum states are respectively known as the \emph{Werner} \cite{Werner1989} and the \emph{Isotropic} \cite{Horodecki1999iso} states, and correspond to mixtures of the projections on the symmetric and the anti-symmetric subspaces, respectively to mixtures of the maximally mixed state and the maximally entangled state \cite[Example 6.10]{watrous2018theory}. For these one-parameter families (two real parameters if one ignores the unit trace normalization), separability, the PPT property, as well as other relevant properties have been characterized, thanks to their simple structure. Importantly, these states correspond, through the Choi-Jamio{\l}kowski isomorphism, to covariant quantum channels, that is
$$\forall U \in \mathcal U(d),\qquad \Phi(UXU^*) = \overbar U \Phi(X) U^\top \quad \text{ or } \quad \Phi(UXU^*) = U \Phi(X) U^*.$$

In this work, we analyze the properties of quantum states and channels which are symmetric with respect to the \emph{diagonal} unitary and orthogonal groups. These classes of states are described by roughly $d^2$ real parameters, and are of intermediate complexity between the full-unitary invariant states described above and the full set of bipartite quantum states (requiring order $d^4$ parameters). We shall consider three situations: the first two correspond to the equations above, with the unitary $U$ restricted to the class of diagonal unitary matrices (diagonal matrices with arbitrary complex phases), while the third one corresponds to $U$ being restricted to diagonal orthogonal matrices (diagonal matrices with arbitrary signs). These classes of states, called respectively LDUI, CLDUI, and LDOI,  have been introduced in \cite{chruscinski2006class, johnston2019pairwise,nechita2021graphical}. We provide a detailed analysis of these matrices, from various points of views: linear algebra, convexity, positivity, separability, etc. Further specializations of these matrices have appeared on numerous occasions in the literature; we thus give a unified treatment of these classes of matrices under one umbrella, in an effort to streamline the different proof techniques used previously. 

A focal point of our efforts is the separability problem. It turns out that the classes of states we investigate are still rich enough for the separability problem to be intractable \cite{yu2016separability,tura2018separability,johnston2019pairwise}, but the situation is simpler, since the search space for separable decompositions is smaller. We describe the separability properties of these invariant states in terms of two cones of pairs and triples of matrices, called the \emph{pairwise completely positive} \cite{johnston2019pairwise} and the \emph{triplewise completely positive} \cite{nechita2021graphical} cone, respectively. Both these notions can be understood as extensions of the classical case of completely positive matrices \cite{bremner1997complexity}, which is a key notion in combinatorics and optimization \cite{berman2003completely}. We borrow and extend techniques from these fields to provide several separability and entanglement criteria for our classes of invariant states. Building upon the present work, a novel approach to detect entanglement using critical graph-theoretic techniques has been developed in \cite{Singh2020entanglement}.

In the latter half of this paper, we switch perspectives and discuss everything from the point of view of linear maps between matrix algebras, using the Choi-Jamio{\l}kowski isomorphism. Quantum channels which are covariant with respect to the action of diagonal unitary matrices have appeared in the literature under the name of ``mean unitary conjugation channels'' (MUCC) \cite{liu2015unitary,lopes2015generic}, we dub them here DUC, CDUC, DOC, in parallel with the case of bipartite states. The action of these maps on the space of $d\times d$ complex matrices is parameterized by three $d\times d$ complex matrices: $A,B$ and $C$ (with $\operatorname{diag}B=\operatorname{diag}C=0$), and has the following form: 
$$  \Phi (Z) = \operatorname{diag}(A\ket{\operatorname{diag}Z}) + B\odot Z + C\odot Z^\top. $$
We provide several equivalent characterizations of these classes of maps, based on their Choi, Kraus and Stinespring representations. We also focus on the composition properties of these maps, thus revealing intimate connections with the corresponding class of LDOI matrices. As is done for bipartite matrices, we present a survey of some important examples of maps which lie in our class, the most notable among these being the Choi map and all its proposed generalizations. Remarkably, the results obtained in this work are used in \cite{singh2020ppt2} to show that the PPT square conjecture holds for the diagonal unitary covariant maps, which contain all Choi-type maps as special cases.  

Although the presentation is self-contained, we refer the reader for the background material and the proofs of some results to our previous work \cite{nechita2021graphical} as well as to the paper of Johnston and MacLean \cite{johnston2019pairwise}, where the notions of CLDUI and PCP matrices were introduced. The paper has two main parts, one focusing on bipartite matrices, the other one on linear maps. The sections are structured as follows. In Section \ref{sec:LDUI-CLDUI-LDOI} we review the main background material on invariant states from \cite{nechita2021graphical}, recalling several key results from this paper. Section \ref{sec:LDOIexamples} contains a list of salient examples, already present in the literature, or new. Sections \ref{sec:linear-structure} and \ref{sec:convex-structure} deal with the linear, resp.~convex structure of the sets of invariant states and the different cones associated with them. In Section \ref{sec:DOC} we change gears, and focus on linear maps between matrix algebras; completely positive maps, and quantum channels in particular are discussed here. Section \ref{sec:DOCexamples} contains a list of examples of quantum channels which are of particular interest. In the next two Sections (\ref{sec:kraus} and \ref{sec:tcp}) we discuss, respectively, the special form of Kraus and Stinespring representations of covariant maps, and the structure of their invariant subspaces, which we use to present necessary and sufficient conditions for separability of the corresponding invariant states. We conclude our work with an overview of our results and some directions for future work. At several instances in this paper, we employ the diagrammatic language of boxes and strings to represent tensor equations, in order to ease proofs and make the presentation more visually intuitive. For readers who are unfamiliar with this language, Section~3 of our previous work \cite{nechita2021graphical} should suffice for a quick introduction.

\section{Local diagonal unitary and orthogonal invariant matrices} \label{sec:LDUI-CLDUI-LDOI}

In this section, we recall the basic definitions and properties of the families of local diagonal unitary/orthogonal invariant matrices, which were first introduced in \cite{chruscinski2006class}, and later studied in \cite{johnston2019pairwise, nechita2021graphical}. For more details and proofs of the results stated here, the reader should refer to our previous work \cite[Sections 6,7 and Appendix B]{nechita2021graphical}. We start by fixing some basic notation. 

We use Dirac's \emph{bra-ket} notation to denote column vectors $v \in \mathbb{C}^d$ as kets $\ket{v}$ and their dual row vectors (conjugate transposes) $v^* \in (\mathbb{C}^d)^*$ as bras $\bra{v}$; note that we do not require the kets or the bras to have unit norm. In this notation, the standard \emph{inner product} $v^*w$ on $\mathbb{C}^d$ is denoted by $\langle v | w \rangle$ and the rank one matrix $vw^*$ is denoted by the \emph{outer product} $\ketbra{v}{w}$. The standard basis in $\C{d}$ is denoted by $\{ \ket{i} : i\in [d] \}$, where $[d]\coloneqq \{1,2,\ldots ,d \}$. $\Mreal{d,d'}$ and $\M{d,d'}$ denote the sets of all $d\times d'$ real and complex matrices, respectively, with the standard basis $\{ \ketbra{i}{j} : i\in [d], j\in [d'] \}$. If $d=d'$, we write $\Mreal{d} \coloneqq \Mreal{d,d}$ and $\M{d}\coloneqq \M{d,d}$. 
 For a matrix $V\in \M{d}$, $V^*$ and $\overbar{V}$ denote its adjoint and entrywise complex conjugate, respectively. We use the Hilbert-Schmidt inner product on $\M{d}$: $\langle A,B\rangle = \operatorname{Tr}(A^*B)$. \emph{Hadamard} (or entrywise) products of vectors $\ket{v}, \ket{w} \in \mathbb{C}^d$ and matrices $V,W \in \M{d}$ are denoted, respectively, by $\ket{v \odot w}$ and $V\odot W$. \emph{Transposition} on $\M{d}$ with respect to the standard basis in $\mathbb{C}^d$ is denoted by $\top$. $\Msa{d}$ denotes the real vector space of all \emph{self-adjoint} matrices in $\M{d}$. The convex cones of \emph{entrywise non-negative} and \emph{positive semi-definite} matrices in $\M{d}$ are denoted by $\EWP_d$ and $\PSD_d$ respectively.

For a vector $\ket{v}\in \C{d}$, $\operatorname{diag}\ket{v}\in \M{d}$ is the diagonal matrix with entries equal to that of $\ket{v}$. For a matrix $V\in \M{d}$, we define two kinds of diagonal operations. The first one extracts out the diagonal part of $V$ and puts it back in matrix form: $\operatorname{diag}V\in \M{d}$. The second one does the same thing but the result is a vector: $\ket{\operatorname{diag}V}\in \C{d}$. The above definitions are collected below:
\begin{equation}
    \operatorname{diag}\ket{v} = \sum_{i=1}^d v_i \ketbra{i}{i}, \qquad \operatorname{diag}V = \sum_{i=1}^d V_{ii}\ketbra{i}{i}, \qquad \ket{\operatorname{diag}V} = \sum_{i=1}^d V_{ii}\ket{i}.
\end{equation}

Vectors $\ket{v}\otimes \ket{w}$ in the tensor product space $\mathbb{C}^d \otimes \mathbb{C}^{d'}$ are denoted as $\ket{vw}$. For bipartite matrices $X\in \M{d}\otimes \M{d'}$, the \emph{partial transposition} with respect to the first and second subsystem is denoted by $X^\text{\reflectbox{$\Gamma$}}\coloneqq (\top\otimes \operatorname{id})(X)$ and $X^{\Gamma}\coloneqq (\operatorname{id}\otimes \top) (X)$ respectively. A matrix $X\in \M{d}\otimes \M{d'}$ is said to be PPT if both $X$ and $X^\Gamma$ are positive semi-definite. By utilizing the isomorphism $\M{d}\otimes \M{d'} \simeq \M{dd'}$, we denote the sets of all self-adjoint, entrywise non-negative and positive semi-definite matrices in $\M{d}\otimes \M{d'}$ by $\Msa{dd'}$, $\EWP_{dd'}$ and $\PSD_{dd'}$ respectively.  

The groups of \emph{diagonal unitary} and \emph{diagonal orthogonal} matrices in $\mathcal{M}_d(\mathbb{C})$ will play a central role in our paper. These are denoted by $\mathcal{DU}_d$ and $\mathcal{DO}_d$ respectively. A \emph{random diagonal unitary} matrix is a random variable $U\in \mathcal{DU}_d$, having independent and identically distributed (i.i.d.) complex phases on its diagonal: $U_{kk} = e^{i\theta_k}$, with $\theta_1, \theta_2, \ldots ,\theta_d$ i.i.d uniformly in $[0,2\pi]$. Similarly, a \emph{random diagonal orthogonal} matrix is a random variable $O\in \mathcal{DO}_d$, having independent and identically distributed real signs on its diagonal: $O_{kk} = s_k$, with $s_1, \ldots ,s_d$ i.i.d uniformly in $\{ \pm 1 \}$.

With all the notation in place, we begin with the most important definition of this section.

\begin{definition} \label{def:LDUI-CLDUI-LDOI}
    A bipartite matrix $X \in \mathcal{M}_d(\mathbb{C}) \otimes \mathcal{M}_d(\mathbb{C})$ is said to be
    \begin{itemize}
        \item \emph{local diagonal unitary invariant (LDUI)} if
        $$\forall U \in \mathcal{DU}_d, \qquad (U \otimes  U) X (U^* \otimes  U^*) = X,$$
        \item \emph{conjugate local diagonal unitary invariant (CLDUI)}  if
        $$\forall U \in \mathcal{DU}_d, \qquad (U \otimes U^*) X (U^* \otimes U) = X,$$
        \item \emph{local diagonal orthogonal invariant (LDOI)} if
        $$\forall O \in \mathcal{DO}_d, \qquad (O \otimes O) X (O \otimes O) = X.$$
    \end{itemize}
\end{definition}

The vector subspaces of LDUI, CLDUI and LDOI matrices in $\mathcal{M}_d(\mathbb{C}) \otimes \mathcal{M}_d(\mathbb{C})$ are denoted, respectively, by $\LDUI_d$, $\CLDUI_d$ and $\LDOI_d$. We can now begin to set up important bijections between the newly introduced vector spaces and certain families of matrix pairs/triples defined as follows: \cite[Propositions 6.4 and 7.2]{nechita2021graphical}
\begin{align}
    \MLDUI{d} &\coloneqq \{ (A,B) \in \M{d}\times \M{d} \, \big| \, \operatorname{diag}(A)=\operatorname{diag}(B) \}, \label{eq:MLDUI}\\
    \MLDOI{d} &\coloneqq \{ (A,B,C) \in \M{d}\times \M{d}\times \M{d} \, \big| \, \operatorname{diag}(A)=\operatorname{diag}(B)=\operatorname{diag}(C)\}. \label{eq:MLDOI}
\end{align}

\begin{proposition} \label{prop:LDOI-ABC}
Let $\MLDUI{d}$ and $\MLDOI{d}$ be vector spaces defined in Eqs.~\eqref{eq:MLDUI} and \eqref{eq:MLDOI}, endowed with the usual component-wise addition and scalar multiplication. Then,
\begin{itemize}
    \item $\LDUI_d$ and $\CLDUI_d$ are isomorphic as vector spaces to $\MLDUI{d}$,
    \item $\LDOI_d$ is isomorphic as a vector space to $\MLDOI{d}$.
\end{itemize}
\end{proposition}
We take a moment to state the bijections from Proposition~\ref{prop:LDOI-ABC} more explicitly. Notice that given $A\in \M{d}$, $\widetilde{A}$ denotes the matrix with zero diagonal but with the same off-diagonal entries as $A$: $\widetilde{A} := A - \operatorname{diag}A$, or, in coordinates
\begin{equation} \label{eq:tilde}
    \widetilde{A}_{ij} = \begin{cases}
    \,\,\,\, 0 , \quad &\text{if } i = j\\
     A_{ij} , \quad &\text{otherwise }. 
\end{cases}
\end{equation}

\begin{align} 
    \bullet \quad X^{(1)} : \MLDUI{d} &\rightarrow \LDUI_d \nonumber \\
    (A,C) &\mapsto X^{(1)}_{(A,C)} \coloneqq \includegraphics[align=c]{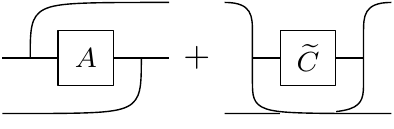} = \sum_{i,j=1}^d A_{ij} \ketbra{ij}{ij} + \sum_{1\leq i\neq j\leq d} C_{ij} \ketbra{ij}{ji}, \nonumber \\
    \bullet \quad X^{(2)} : \MLDUI{d} &\rightarrow \CLDUI_d \nonumber \\
    (A,B) &\mapsto X^{(2)}_{(A,B)} \coloneqq \includegraphics[align=c]{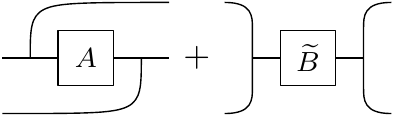} = \sum_{i,j=1}^d A_{ij} \ketbra{ij}{ij} + \sum_{1\leq i\neq j\leq d} B_{ij} \ketbra{ii}{jj} \nonumber, \\
    \bullet \quad X^{(3)} : \MLDOI{d} &\rightarrow \LDOI_d \nonumber \\
    (A,B,C) &\mapsto X^{(3)}_{(A,B,C)} \coloneqq \includegraphics[align=c]{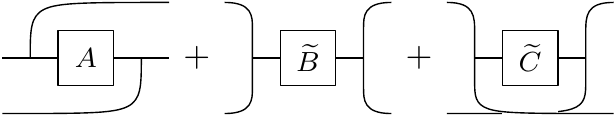} \nonumber \\ 
    &\hspace{2.05cm}= \sum_{i,j=1}^d A_{ij} \ketbra{ij}{ij} + \sum_{1\leq i\neq j\leq d} B_{ij} \ketbra{ii}{jj} + \sum_{1\leq i\neq j\leq d} C_{ij} \ketbra{ij}{ji}. \nonumber
\end{align}
Coordinate-wise, we have (note that $X(i_1i_2,j_1j_2) = \langle i_1 i_2|X|j_1 j_2 \rangle$):
\begin{align}
\label{eq:X1-coord}    X^{(1)}_{(A,C)}(i_1i_2, j_1j_2) &= \begin{cases}
    A_{i_1i_2}, &\quad \text{ if } i_1 = j_1 \text{ and } i_2 = j_2\\
    C_{i_1j_1}, &\quad \text{ if } i_1 = j_2 \text{ and } i_2 = j_1\\
    0, &\quad \text{ otherwise } 
    \end{cases}\\
\label{eq:X2-coord}    X^{(2)}_{(A,B)}(i_1i_2, j_1j_2) &= \begin{cases}
    A_{i_1i_2} &\quad \text{ if } i_1 = j_1 \text{ and } i_2 = j_2\\
    B_{i_1j_1} &\quad \text{ if } i_1 = i_2 \text{ and } j_1 = j_2\\
    0, &\quad \text{ otherwise }
    \end{cases}\\  
\label{eq:X3-coord}    X^{(3)}_{(A,B,C)}(i_1i_2, j_1j_2) &= \begin{cases}
    A_{i_1i_2} &\quad \text{ if } i_1 = j_1 \text{ and } i_2 = j_2\\
    B_{i_1j_1} &\quad \text{ if } i_1 = i_2 \text{ and } j_1 = j_2\\
    C_{i_1j_1} &\quad \text{ if } i_1 = j_2 \text{ and } i_2 = j_1\\
    0, &\quad \text{ otherwise }
    \end{cases}
\end{align}
The general matrix form of a $3\otimes 3$ LDOI matrix is exhibited below (dots represent zeros):
\begin{equation} \label{eq:LDOI-block-3}
    X^{(3)}_{(A,B,C)} =  \left(\mathcode`0=\cdot
\begin{array}{ *{3}{c} | *{3}{c} | *{3}{c} }
   A_{11} & 0 & 0 & 0 & B_{12} & 0 & 0 & 0 & B_{13} \\
   0 & A_{12} & 0 & C_{12} & 0 & 0 & 0 & 0 & 0 \\
   0 & 0 & A_{13} & 0 & 0 & 0 & C_{13} & 0 & 0 \\\hline
   0 & C_{21} & 0 & A_{21} & 0 & 0 & 0 & 0 & 0 \\
   B_{21} & 0 & 0 & 0 & A_{22} & 0 & 0 & 0 & B_{23} \\
   0 & 0 & 0 & 0 & 0 & A_{23} & 0 & C_{23} & 0 \\ \hline
   0 & 0 & C_{31} & 0 & 0 & 0 & A_{31} & 0 & 0 \\
   0 & 0 & 0 & 0 & 0 & C_{32} & 0 & A_{32} & 0 \\
   B_{31} & 0 & 0 & 0 & B_{32} & 0 & 0 & 0 & A_{33} \\
  \end{array}
\right).
\end{equation}

We refer the reader to Proposition \ref{prop:block-structure} for the block-structure of the above matrices.
The next proposition identifies orthogonal projections on the LDUI/CLDUI/LDOI subspaces with certain local diagonal unitary/orthogonal averaging operations \cite[Propositions 6.4, 7.2]{nechita2021graphical}.

\begin{proposition} \label{prop:LDUI/CLDUI/LDOI-projections}
Denote the orthogonal projections on the vector subspaces $\LDUI_d, \CLDUI_d$ and $\LDOI_d$ in $\M{d}\otimes \M{d}$ by $\operatorname{Proj}_{\LDUI}$, $\operatorname{Proj}_{\CLDUI}$ and $\operatorname{Proj}_{\LDOI}$ respectively. For an arbitrary matrix $X\in \M{d}\otimes \M{d}$, define matrices $A,B,C\in \M{d}$ entrywise as $A_{ij}=\langle ij|X|ij\rangle$, $B_{ij}=\langle ii|X|jj\rangle$, and $C_{ij}=\langle ij|X|ji\rangle$ for $i,j\in [d]$. Then,
\begin{align*}
    \operatorname{Proj}_{\LDUI}(X) &=  \mathbb{E}_U [(U\otimes U)X(U^*\otimes U^*)] = X^{(1)}_{(A,C)},  \\
    \operatorname{Proj}_{\CLDUI}(X) &= \mathbb{E}_U [(U\otimes U^*)X(U^*\otimes U)] = X^{(2)}_{(A,B)},   \\
    \operatorname{Proj}_{\LDOI}(X) &= \mathbb{E}_O [(O\otimes O)X(O\otimes O)] = X^{(3)}_{(A,B,C)}.
\end{align*}
where $U\in \mathcal{DU}_d$ (resp. $O\in \mathcal{DO}_d$) is a random diagonal unitary (resp. orthogonal) matrix, and $\mathbb{E}_U$ (resp. $\mathbb{E}_O$) denotes the expectation with respect to the distribution of $U$ (resp. $O$).   
\end{proposition}

We refer the reader to \cite[Theorems 4.8 and 5.5]{nechita2021graphical} for a nice graphical method to compute general expectations of the form given in Proposition~\ref{prop:LDUI/CLDUI/LDOI-projections}.

\begin{remark} \label{remark:(C)LDUIsubspaceLDOI}
For $(A,B) \in \MLDUI{d}$, the above stated bijections imply that
\begin{equation}
    X^{(1)}_{(A,B)} = X^{(3)}_{(A,\operatorname{diag}A,B)}, \qquad X^{(2)}_{(A,B)} = X^{(3)}_{(A,B,\operatorname{diag}A)}, \qquad 
    [X^{(1)}_{(A,B)}]^\Gamma = X^{(2)}_{(A,B)}.
\end{equation}
Hence, $\LDUI_d$ and $\CLDUI_d$ are vector subspaces of $\LDOI_d$.
\end{remark}

Next, we introduce the notions of \emph{completely positive}, \emph{pairwise completely positive} and \emph{triplewise completely positive} matrices and link them to the separability problem for matrices in $\LDUI_d$, $\CLDUI_d$ and $\LDOI_d$ \cite[Lemmas 6.6 and 7.5]{nechita2021graphical}. Recall that a positive semi-definite bipartite matrix $X \in \mathcal{M}_d(\mathbb{C}) \otimes \mathcal{M}_d(\mathbb{C})$ is said to be \emph{separable} if there exists a family of vectors $\{ \ket{v_k}, \ket{w_k} \}_{k \in I} \subseteq \mathbb{C}^d$ for a finite index set $I$, such that Eq.~\eqref{eq:sep} holds (observe that $V$ and $W$ are matrices in $\M{d,|I|}$ with columns given by the vectors $\{ \ket{v_k} \}_{k\in I}$ and $\{ \ket{w_k} \}_{k \in I}$ respectively). A positive semi-definite $X\in \M{d}\otimes \M{d}$ which is not separable is called \emph{entangled}.
\begin{equation} \label{eq:sep}
    X = \sum_k |v_k\rangle\langle v_k| \otimes |w_k\rangle\langle w_k| = \includegraphics[align=c]{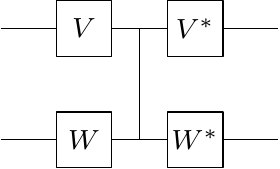}
\end{equation}
\begin{definition}[Completely positive matrices]
A matrix $A\in \Mreal{d}$ is said to be \emph{Completely Positive (CP)} if there exists an entrywise non-negative matrix $V\in \EWP_{d,d'} \subset \Mreal{d,d'}$, for arbitrary $d'\in \mathbb{N}$, such that the following decomposition holds: $A = VV^\top$.
\end{definition}

We would like to warn the reader at this point that the notion of a completely positive \emph{matrix} defined above is unrelated from the one of a completely positive \emph{map}, discussed in Section \ref{sec:DOC}.

\begin{definition}[Pairwise completely positive matrices] \label{def:pcp}
A matrix pair $(A,B) \in \MLDUI{d}$ is said to be \emph{Pairwise Completely Positive (PCP)} if there exist matrices $V,W \in \mathcal{M}_{d,d'}(\mathbb{C})$, for arbitrary $d' \in \mathbb{N}$, such that the following decomposition holds:
\begin{equation}\label{eq:def-PCP}
    A = (V \odot \overbar{V})(W \odot \overbar{W})^*, \qquad B = (V \odot W) (V \odot W)^*.
\end{equation}
\end{definition}

\begin{definition}[Triplewise completely positive matrices] \label{def:tcp}
A matrix triple $(A,B,C) \in \MLDOI{d}$ is said to be \emph{Triplewise Completely Positive (TCP)} if there exist matrices $V,W \in \mathcal{M}_{d,d'}(\mathbb{C})$, for arbitrary $d' \in \mathbb{N}$, such that the following decomposition holds:
\begin{equation}\label{eq:def-TCP}
    A = (V \odot \overbar{V})(W \odot \overbar{W})^*, \qquad B = (V \odot W) (V \odot W)^*, \qquad C = (V \odot \overbar{W})(V \odot \overbar{W})^*.
\end{equation}
\end{definition}

The next result is due to \cite[Theorem 3.4]{johnston2019pairwise}, where the authors prove that the notion of PCP matrices generalizes that of CP matrices. In our discussion of partial transpose invariant LDOI matrices in Example~\ref{eg:states-PTinv}, we will extend this generalization to TCP matrices as well.

\begin{theorem}\label{theorem:CP<=PCP}
For $A\in \Mreal{d}$, the following equivalence holds: $A$ is \emph{CP} $\iff (A,A)$ is \emph{PCP}. 
\end{theorem}

We now arrive at the crucial link between separability of LDOI matrices and the different notions of completely positive matrices introduced.

\begin{theorem} \label{theorem:LDOI-sep}
Consider a triple $(A,B,C) \in \MLDOI{d}$. Then
\begin{itemize}
    \item $X^{(1)}_{(A,B)} \in \LDUI_d$ is separable $\iff X^{(2)}_{(A,B)} \in \CLDUI_d$ is separable $\iff (A,B)$ is \emph{PCP},
    \item $X^{(3)}_{(A,B,C)} \in \LDOI_d$ is separable $\iff (A,B,C)$ is \emph{TCP}.
\end{itemize}
\end{theorem}

\begin{proof}
    We prove the equivalence for LDOI matrices, and urge the readers to mimic the same proof for LDUI/CLDUI matrices. Assume first that $X^{(3)}_{(A,B,C)}$ is separable and hence can be decomposed as in Eq.~\eqref{eq:sep}, with matrices $V,W \in \M{d,|I|}$. Using the coordinate-wise relations in Eq.~\eqref{eq:X3-coord}, it is easy to infer the TCP decomposition of $(A,B,C)$: 
    \begin{align*}
        \includegraphics[align=c]{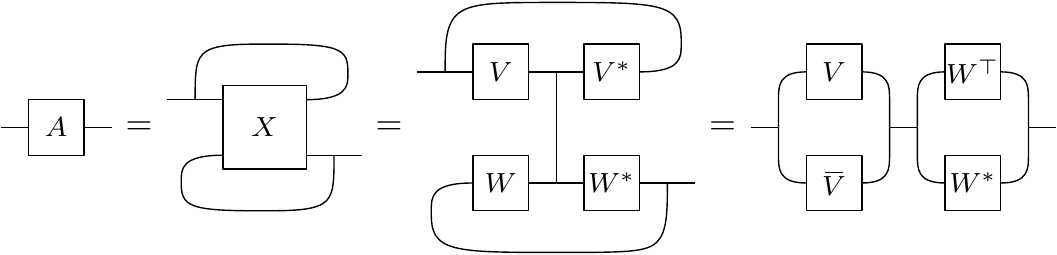} &= (V \odot \overbar{V})(W \odot \overbar{W})^*, \\[0.5cm]
        \includegraphics[align=c]{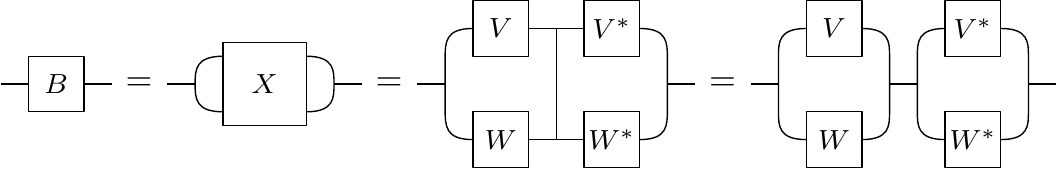} &= (V \odot W) (V \odot W)^*, \\[0.5cm]
        \includegraphics[align=c]{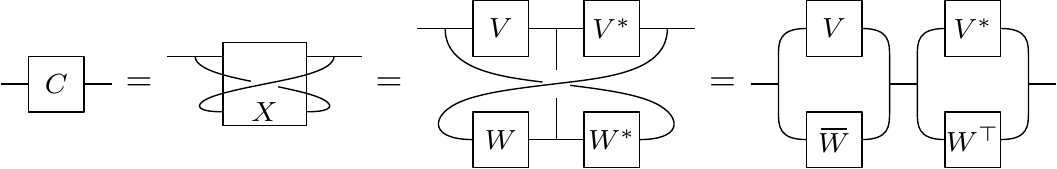} &= (V \odot \overbar{W})(V \odot \overbar{W})^*. \\[0.3cm]
    \end{align*}
Conversely, assume that $(A,B,C)$ is TCP and let $V,W\in \M{d,d'}$ form its TCP decomposition. We construct a separable matrix $X\in \M{d}\otimes \M{d}$ using the $V,W$ matrices, as was done in Eq.~\eqref{eq:sep}. It is then not too difficult to see that 
    \begin{equation*}
        X^{(3)}_{(A,B,C)} = \mathbb{E}_O [(O\otimes O)X(O\otimes O)] = \operatorname{Proj}_{\LDOI}(X),
    \end{equation*}
    and the proof concludes with the observation that since $\operatorname{Proj}_{\LDOI}$ is a (classically correlated) local operation, it preserves separability (see Proposition~\ref{prop:LDUI/CLDUI/LDOI-projections}). 
\end{proof}

Several elementary results on TCP matrix triples are presented in \cite[Appendix B]{nechita2021graphical}. Some of these results, which will be of relevance to us in later sections, are presented below. Notice that the \emph{entrywise} and \emph{trace} (or \emph{nuclear}) norms on $\mathcal{M}_d(\mathbb{C})$ are denoted by $\Vert . \Vert_1$ and $\Vert . \Vert_{\operatorname{Tr}}$ respectively:
$$\|X\|_1 := \sum_{i,j=1}^d |X_{ij}| \qquad \text{ and } \qquad \|X\|_{\operatorname{Tr}} := \Tr \sqrt{XX^*}.$$

\begin{lemma} \label{lemma:tcp-properties}
Let $(A,B,C)\in \MLDOI{d}$ be \emph{TCP}. Then
\begin{enumerate}
    \item both $(A,B)$ and $(A,C)$ are \emph{PCP},
    \item $A\in \EWP_d$ and $B, C \in \PSD_d$,
    \item $A_{ij}A_{ji} \geq \vert B_{ij} \vert^2 $ and $A_{ij}A_{ji} \geq \vert C_{ij} \vert^2 \,\, \forall i,j \in [d]$,
    \item $\Vert A \Vert_1 - \Vert A \Vert_{\operatorname{Tr}} \geq \Vert B \Vert_1 - \Vert B \Vert_{\operatorname{Tr}}$ and $\Vert A \Vert_1 - \Vert A \Vert_{\operatorname{Tr}} \geq \Vert C \Vert_1 - \Vert C \Vert_{\operatorname{Tr}}$,
    \item $\Vert A \Vert_1 - \Vert A \Vert_{\operatorname{Tr}} \geq 2\sum_{1\leq i < j \leq d} \operatorname{max}\{ \vert B_{ij} \vert, \vert C_{ij} \vert \}$.
\end{enumerate}
\end{lemma}

Before proceeding further, it is essential to recall some definitions. A matrix $Z\in \mathcal{M}_d(\mathbb{C})$ is said to be \emph{row} (resp. \emph{column}) \emph{diagonally dominant} if $Z_{ii} \geq \sum_{j\neq i} \vert Z_{ij} \vert$ for each $i\in [d]$ (resp. $Z_{ii} \geq \sum_{j\neq i} \vert Z_{ji} \vert$ for each $i\in [d]$). A matrix $Z \in \M{d}$ is said to be \emph{diagonally dominant} if it is both row and column diagonally dominant. The \emph{comparison matrix} $M(Z)$ of $Z$ in $\mathcal{M}_d(\mathbb C)$ is defined entrywise as follows: 
\begin{equation}
    M(Z)_{ij} = \begin{cases}
    \,\,\,\, \vert Z_{ij} \vert, \quad &\text{if } i = j\\
    -\vert Z_{ij} \vert, \quad &\text{otherwise }.
\end{cases}
\end{equation}
The next result \cite[Theorem 4.4]{johnston2019pairwise} provides easily verifiable sufficient conditions for a pair $(A,B)\in \MLDUI{d}$ to be PCP and hence yields a non-trivial test for separability of the associated matrices in $\LDUI_d$ and $\CLDUI_d$.

\begin{lemma} \label{lemma:PCP-comparison}
Consider a matrix pair $(A,B)\in \MLDUI{d}$ such that $A\in \EWP_d$, $B\in \PSD_d $ and $A_{ij}A_{ji} \geq |B_{ij}|^2 \,\, \forall i,j\in [d]$. Then,
$$ B \text{ is diagonally dominant } \implies M(B)\in \PSD_d \implies (A,B) \in \PCP_d.$$
\end{lemma}

Lemma~\ref{lemma:LDOI-psd-ppt} below collects some important properties of matrices in $\LDOI_d$ \cite[Lemma 7.6]{nechita2021graphical}. Analogous results for matrices in $\LDUI_d$ and $\CLDUI_d$ can be derived with the help of the following Lemma and Remark~\ref{remark:(C)LDUIsubspaceLDOI}, see \cite[Lemmas 6.7 and 6.8]{nechita2021graphical}. Recall that if we define the \emph{realignment map} $R:\M{d}\otimes \M{d} \rightarrow \M{d}\otimes \M{d}$ as $R(\ketbra{i}{j}\otimes \ketbra{k}{l}) = \ketbra{i}{k}\otimes \ketbra{j}{l} \,\, \forall i,j\in [d]$, then all separable matrices $X\in \M{d}\otimes \M{d}$ satisfy: $||R(X)||_{\operatorname{Tr}}\leq \operatorname{Tr}X$, see \cite{Kai2002realignment, Rudolph2000realignment}. This is known as the \emph{realignment criterion} of separability. In this paper, we use the notation $X^R := R(X)$.

\begin{lemma} \label{lemma:LDOI-psd-ppt}
Consider an arbitrary $X^{(3)}_{(A,B,C)}\in \LDOI_d$. Then,
\begin{enumerate}
    \item $X^{(3)}_{(A,B,C)}\in \PSD_{d^2}\iff A \in \EWP_d$, $B \in \PSD_d$, $C\in \Msa{d}$, and $A_{ij}A_{ji} \geq \vert C_{ij} \vert^2 \,\forall i,j \in~[d]$,
    \item $X^{(3)\,\Gamma}_{(A,B,C)} \in \PSD_{d^2} \iff A\in \EWP_d$, $B\in \Msa{d}$, $C\in \PSD_d$, and $A_{ij}A_{ji} \geq \vert B_{ij} \vert^2 \,\forall i,j \in~[d]$,
    \item $X^{(3)}_{(A,B,C)}$ satisfies the realignment criterion $\iff$ condition (5) of Lemma~\ref{lemma:tcp-properties} holds.
\end{enumerate}
\end{lemma}

In Quantum Mechanics, states $\rho$ of a physical system (called \emph{quantum states}) are modelled as bounded operators on a separable Hilbert space $\mathcal{H}$, i.e., $\rho\in \mathcal{B}(\mathcal{H})$, which are positive semi-definite and trace-class with $\operatorname{Tr}\rho=1$, see \cite[Chapter 2]{watrous2018theory}, \cite[Chapter 2]{holevo2019channels}. For multiparty systems, the Hilbert spaces are naturally required to acquire a tensor product structure, and the relevant states then lie in the space $\mathcal{B}(\mathcal{H}_1 \otimes \dotsb \otimes \mathcal{H}_n)$. In a finite dimensional bipartite setting: $\mathcal{H}_1\simeq \mathbb{C}^{d_1}, \mathcal{H}_2\simeq \mathbb{C}^{d_2}$, it is obvious that quantum states of a physical system are just positive semi-definite matrices $\rho\in \M{d_1} \otimes \M{d_2}$ with unit trace. The next lemma exploits the results from Lemma~\ref{lemma:LDOI-psd-ppt} to define the class of quantum states within the family of local diagonal orthogonal invariant matrices in $\M{d}\otimes \M{d}$.

\begin{lemma}
A matrix $X^{(3)}_{(A,B,C)}\in \LDOI_d$ is a quantum state $\iff A\in \EWP_d$, $B\in \PSD_d$ and $C\in \Msa{d}$ such that $\sum_{i,j=1}^d A_{ij} = || A ||_1 =1$ and $A_{ij}A_{ji} \geq \vert C_{ij} \vert^2 \,\, \forall i,j \in [d]$.
\end{lemma}

\section{Important classes of LDOI matrices}\label{sec:LDOIexamples}

This section contains a plethora of examples of (classes of) local diagonal orthogonal invariant matrices. We list well-known examples of bipartite quantum states which fit into this category, many classes from the literature, as well as some new ones. The fact that so many examples of bipartite quantum states fit into this framework motivates the study of the general properties of LDOI matrices, which will form the subject matter of Sections ~\ref{sec:linear-structure} and \ref{sec:convex-structure} later. We collect all the examples in Table \ref{tbl:LDOI} which can be found at the end of this section. In what follows, $\mathbb{I}_d\in \M{d}$ is the identity matrix and $\mathbb{J}_d\in \M{d}$ is the matrix in which all entries are equal to one.

\begin{example}[\emph{Diagonal matrices}]\ \\[0.1cm] \label{eg:states-diag}
It is easy to see that a bipartite matrix $X$ lies in $\LDUI_d \cap \CLDUI_d$ if and only if it is diagonal. The associated triple $(A,B,C)\in \MLDOI{d}$ then has an arbitrary $A\in \M{d}$ and $B=C=\operatorname{diag}A$. In particular, this class contains tensor products of diagonal matrices $\operatorname{diag}\ket{y} \otimes \operatorname{diag}\ket{z}$ with $\ket{y},\ket{z} \in \mathbb C^d$, where  $A = \ketbra{y}{z}$. By fixing $\ket{y} = \ket{z} = \ket{\operatorname{diag}\mathbb{I}_d} \in \mathbb{C}^d$, we obtain the identity matrix $\mathbb{I}_d \otimes \mathbb{I}_d$ with $A=\mathbb{J}_d$. Being diagonal, these matrices are separable if and only if they are PPT if and only if they are positive semi-definite if and only if $A\in \EWP_d$.
\end{example}

\begin{example}[\emph{Unit rank and tensor product states}]\ \\[0.1cm] \label{eg:states-unitrank}
In this example, we deal with LDOI matrices which factorize either as a unit rank matrix or as a tensor product, see also Remark~\ref{cor:rank-X}.

\begin{proposition}
The \emph{LDOI} matrices of unit rank are of the form $\ketbra{Y}{Z}$, where $Y,Z$ are either
\begin{itemize}
    \item of the `diagonal' form: $\ket{Y}=\sum_{i=1}^d y_i \ket{ii}$ and $\ket{Z}=\sum_{i=1}^d z_i \ket{ii}$ in $\C{d}\otimes \C{d}$,
    \item or are supported on $\mathbb C \ket{ij} \oplus \mathbb C \ket{ji}$ for some $i \neq j$.
\end{itemize}
 In the former case, $B=\ketbra{y}{z}$ and $A=C=\operatorname{diag}B$, where $\ket{y}=\sum_{i=1}^d y_i \ket{i}$ and $\ket{z}=\sum_{i=1}^d z_i \ket{i}$ are vectors in $\C{d}$, the matrices being \emph{CLDUI}. In the latter case, the matrices are \emph{LDUI}. The only unit rank matrices which are simultaneously \emph{CLDUI} and \emph{LDUI} are of the form
\begin{equation}\label{eq:iiii}
    X = \lambda \ketbra{ij}
\end{equation}
for some non-zero $\lambda \in \mathbb C$ and $i,j \in [d]$.
\end{proposition}
\begin{proof}
    Consider a unit rank LDOI matrix $X = \ketbra{Y}{Z}$, with $0\neq \ket{Y},\ket{Z} \in \C{d}\otimes \C{d}$. From the LDOI property (see Definition~\ref{def:LDUI-CLDUI-LDOI}), we see that for all sign vectors $s \in \{\pm 1\}^d$ and all $i,j,a,b\in [d]$, $s_is_j s_a s_b Y_{ij} \bar Z_{ab} = Y_{ij}\bar Z_{ab}$. This relation holds if $\ket{Y}$ and $\ket{Z}$ are of the diagonal form, in which case
    \begin{equation}\label{eq:rank-one-CLDUI}
    X = \sum_{i,j=1}^d y_i \overbar{z_j} \ketbra{ii}{jj}
    \end{equation}
    is CLDUI. The associated matrices $A,B$ and $C$ can also be shown to have the required form, see Proposition~\ref{prop:LDOI-ABC} and the bijections following it. If $\ket{Y}$ is not of the diagonal form, then $Y_{ij} \neq 0$ for some $i \neq j$. This implies that $Z_{ab} = 0$ for all $(a,b) \notin \{(i,j),(j,i)\}$, and we obtain rank one LDUI matrices of the type
    \begin{equation}\label{eq:rank-one-LDUI}
        X = \left(\alpha \ket{ij} + \beta \ket{ji} \right)\left(\gamma \bra{ij} + \delta \bra{ji} \right),
    \end{equation}
    where $\alpha, \beta, \gamma$ and $\delta$ are complex numbers. Finally, selecting those matrices from Eqs.~\eqref{eq:rank-one-CLDUI} and \eqref{eq:rank-one-LDUI} which are both LDUI and CLDUI gives us precisely the matrices in Eq.~\eqref{eq:iiii}. 
\end{proof}
\begin{remark}
    The set of LDOI matrices $X$ having a tensor product structure $X = Y \otimes Z$ can be deduced from the result above by applying the realignment operation, see the Definition used in Lemma~\ref{lemma:LDOI-psd-ppt}.
\end{remark}
\end{example}

\begin{example}[\emph{Werner and Isotropic matrices}] \cite{Werner1989, Horodecki1999iso}\ \\[0.1cm] \label{eg:states-wer-iso}
A matrix $X\in \M{d}\otimes \M{d}$ is called a Werner (resp. isotropic) matrix if it satisfies $(U\otimes U)X(U^*\otimes U^*) = X$ (resp.~$(U\otimes \overbar{U})X(U\otimes \overbar{U})^* = X$) for all unitary matrices $U\in \M{d}$. It is obvious from the definition that the Werner matrices lie in $\LDUI_d$, while the isotropic matrices lie in $\CLDUI_d$ (see Definition~\ref{def:LDUI-CLDUI-LDOI}). The structure of these matrices is known to be of the following form, for $a,b\in \mathbb{C}$:
\begin{equation}
 X^{\mathsf{wer}}_{a,b} = a (\mathbb{I}_d \otimes \mathbb{I}_d) + b \sum_{i,j=1}^d \ketbra{ij}{ji}, \qquad\qquad X^{\mathsf{iso}}_{a,b} = a (\mathbb{I}_d \otimes \mathbb{I}_d) + b \sum_{i,j=1}^d \ketbra{ii}{jj}.
\end{equation}
It is straightforward to check that $X^{\mathsf{wer}}_{a,b} = X^{(1)}_{(A,B)}$ and $X^{\mathsf{iso}}_{a,b} = X^{(2)}_{(A,B)}$, for $A = b\,\mathbb{I}_d + a\mathbb{J}_d$ and $B = a\mathbb{I}_d + b\mathbb{J}_d$. It is equally easy to see that self-adjointness forces the parameters $a,b$ to be real. Now, observe that $B\in \PSD_d$ if and only if $a\geq 0$ and $\operatorname{det}(B) = a^{d-1}(a + db) \geq 0$, i.e., $B\in \PSD_d$ if and only if $a\geq 0$ and $b\geq -a/d$. Also notice that $A_{ij}A_{ji}\geq |B_{ij}|^2 \,\, \forall i,j\in [d]$ if and only if $-a\leq b\leq a$. Combining everything together, we deduce that $X^{\mathsf{wer}}_{a,b}$ and $X^{\mathsf{iso}}_{a,b}$ are PPT if and only if $a\geq 0$ and $-a/d \leq b \leq a$. We now prove that this also suffices to guarantee separability of the concerned matrices.

\begin{proposition} \label{prop:wer-iso-sep}
Let $X^{\mathsf{wer}}_{a,b}, X^{\mathsf{iso}}_{a,b}\in \M{d}\otimes \M{d}$ be the Werner and isotropic matrices, parameterized by the pair $(a,b) \in \C{2}$. Then, the following equivalences hold:
\begin{align*}
    X^{\mathsf{wer}}_{a,b} \text{ is \emph{PPT}} \iff X^{\mathsf{wer}}_{a,b} \text{ is separable} &\iff X^{\mathsf{iso}}_{a,b} \text{ is separable} \iff X^{\mathsf{iso}}_{a,b} \text{ is \emph{PPT}} \\
    &\iff a\geq 0 \text{ and } -a/d \leq b \leq a.
\end{align*}
\end{proposition}

\begin{proof}
The idea is to show that the associated pair $(A,B)$ is PCP in the range determined by $a$ and $b$. To this end, we first fix $a\geq 0$ and assume that $b=-a/d$. It is then evident that $B$ is diagonally dominant and hence Lemma~\ref{lemma:PCP-comparison} tells us that $(A,B)$ is PCP. On the other hand, if $b=a$, we have $$A=B=a\mathbb{I}_d + a\mathbb{J}_d = a\sum_{i=1}^d \ketbra{i}{i} + a\ketbra{\operatorname{diag}\mathbb{I}_d}{\operatorname{diag}\mathbb{I}_d}.$$ It is then clear that $A=B \text{ is CP}\implies (A,B)$ is PCP, see Theorem~\ref{theorem:CP<=PCP}. Separability on the entire interval $-a/d \leq b \leq a$ then follows from convexity, see also \cite[Example 2]{johnston2019pairwise}.
\end{proof}
\end{example}

\begin{example}[\emph{Mixtures of Dicke states or diagonal symmetric matrices}] \cite{yu2016separability, tura2018separability} \label{eg:states-dicke} 
The symmetric subspace in $\C{d}\otimes \C{d}$ is spanned by vectors of the form 
\begin{equation}
    \ket{\psi_{ij}} = \begin{cases}
     (\ket{ij} + \ket{ji})/\sqrt{2}, \quad &\text{if } i < j\\
     \,\,\ket{ii}, \quad &\text{if } i=j
\end{cases}
\end{equation}
which constitute what is called the Dicke basis for the symmetric subspace. Bipartite matrices which are diagonal in the Dicke basis are known as diagonal symmetric matrices:
\begin{equation}
    X^{\mathsf{dicke}}_{Y} = \sum_{1\leq i\leq j \leq d} Y_{ij} \ketbra{\psi_{ij}}{\psi_{ij}},
\end{equation}
where $Y\in \Mreal{d}$ is an arbitrary symmetric matrix. By defining $A = \operatorname{diag}(Y) + \widetilde{Y}/2$, it becomes clear that $X^{\mathsf{dicke}}_{Y} = X^{(1)}_{(A,A)} \in \LDUI_d$. From Lemma~\ref{lemma:LDOI-psd-ppt}, it is evident that $X^{\mathsf{dicke}}_{Y}$ is PPT if and only if $A\in \EWP_d \cap \PSD_d \coloneqq \mathsf{DNN}_d$, i.e., if and only if $A$ is \emph{doubly non-negative}. The equivalence of separability of $X^{\mathsf{dicke}}_{Y}$ and complete positivity of $A$ can be similarly obtained from Theorem~\ref{theorem:CP<=PCP}. Since the convex cones of doubly non-negative and completely positive matrices are equal in dimensions $d\leq 4$ \cite[Theorem 2.4]{berman2003completely}, we conclude that $X^{\mathsf{dicke}}_{Y}\in \M{d}\otimes \M{d}$ with $d\leq 4$ is separable if and only if it is PPT. For $d\geq 5$, every $A\in \mathsf{DNN}_d$ which is not completely positive gives rise to a PPT entangled diagonal symmetric matrix.
\end{example}

\begin{example}[\emph{Partial transpose invariant LDOI matrices}]\ \\[0.1cm] \label{eg:states-PTinv}
Using Proposition~\ref{prop:leg-permutations}, it is straightforward to infer that an LDOI matrix is invariant under partial transposition with respect to the first (resp. second) subsystem if and only if the associated triple $(A,B,C)\in \MLDOI{d}$ satisfies: $C=B^\top$ (resp. $C=B$). Moreover, by definition, these matrices are PPT if and only if they are positive semi-definite, which in turn is equivalent to the condition that $A\in \EWP_d, B\in \PSD_d$ and $A_{ij}A_{ji}\geq |B_{ij}|^2 \,\, \forall i,j\in [d]$, see Lemma~\ref{lemma:LDOI-psd-ppt}. Finally, Proposition~\ref{prop:AB-ABB} below (combined with Theorem~\ref{theorem:LDOI-sep}) shows that separability of these matrices is equivalent to the separability of the corresponding LDUI/CLDUI matrices with matrix pairs $(A,B)\in \MLDUI{d}$.

\begin{proposition} \label{prop:AB-ABB}
For $(A,B)\in \MLDUI{d}$, the following sequence of equivalences hold:
\begin{align*}
    (A,B,B^\top) \text{ is \emph{TCP}} \iff (A,B) \text{ i}&\text{s \emph{PCP}} \iff (A,B,B) \text{ is \emph{TCP}}  \\
    &\Big\Updownarrow \\
    (A,B^\top,B) \text{ is \emph{TCP}} \iff (A,B^\top) \text{ i}&\text{s \emph{PCP}} \iff  (A,B^\top,B^\top) \text{ is \emph{TCP}}.
\end{align*}
\end{proposition}
\begin{proof}
We begin with the first row of equivalences. Since the implications pointing from the ends to the center are trivially obtained from part (1) of Lemma~\ref{lemma:tcp-properties}, we start from the center and assume that $(A,B)$ is PCP with $V,W\in \mathcal{M}_{d,d'}(\mathbb C)$ forming its PCP decomposition, see Definition~\ref{def:pcp}. Now, define matrices $V', W', V'', W''\in \M{d,d'}$ entrywise as follows: 
\begin{alignat*}{2}
V'_{ij} &= V_{ij}\operatorname{phase}(W_{ij}), \qquad\qquad &&W'_{ij} = \vert W_{ij} \vert, \\
V''_{ij} &= |V_{ij}|, \qquad\qquad  &&W''_{ij} = W_{ij}\operatorname{phase}(V_{ij}),
\end{alignat*}
where $\operatorname{phase}(V_{ij})$ and $\operatorname{phase}(W_{ij})$ are the complex phases of the entries of $V$ and $W$: $V_{ij} = \vert V_{ij} \vert \operatorname{phase}(V_{ij})$ and $W_{ij} = \vert W_{ij} \vert \operatorname{phase}(W_{ij})$. Now, observe that since $W'$ is entrywise non-negative (and hence $W' = \overbar{W'}$), $V',W'$ form a TCP decomposition of $(A,B,B)$ as in Definition~\ref{def:tcp}. Similarly, $V'',W''$ form a TCP decomposition of $(A,B,B^\top)$. This establishes all the equivalences in the first row. An identical argument does the same for the second row as well. Now, to connect the two rows, we observe that if $V,W$ form a PCP decomposition of $(A,B)$, then $V'',\overbar{W''}$ (as constructed above) form a PCP decomposition of $(A,B^\top)$.  
\end{proof}

Using Proposition~\ref{prop:AB-ABB}, the conclusion of Theorem~\ref{theorem:CP<=PCP} can be trivially extended to TCP matrices.
\begin{theorem} \label{theorem:CP<=PCP<=TCP}
For $A\in \M{d}$, the following equivalences hold:
\begin{equation*}
    A \text{ is \emph{CP}} \iff (A,A) \text{ is \emph{PCP}} \iff (A,A,A) \text{ is \emph{TCP}}.
\end{equation*}
\end{theorem}
\end{example}

\begin{example}[\emph{LDOI matrices with} $A=\mathbb{J}_d$]\ \\[0.1cm] \label{eg:states-A=J}
 In Example~\ref{eg:states-diag}, it was shown that $\mathbb{I}_d \otimes \mathbb{I}_d \in \LDOI_d$ with $A=\mathbb{J}_d$ and $B=C=\operatorname{diag}A=\mathbb{I}_d$. In this example, we investigate the general class of matrices in $\LDOI_d$ with $(A,B,C)\in \MLDOI{d}$ such that $A=\mathbb{J}_d$. For the moment, let us restrict ourselves to matrices in $\LDUI_d$ and $\CLDUI_d$. Then, a quick application of Lemma~\ref{lemma:LDOI-psd-ppt} is adequate to deduce that such matrices are positive semi-definite if and only if they are PPT if and only if the associated matrix $B$ is a \emph{correlation} matrix.
 \begin{definition}\label{def:corr}
     $B\in\M{d}$ is said to be a \emph{correlation} matrix if $B\in \PSD_d$ and $\operatorname{diag}B=\mathbb{I}_d$
\end{definition}
 
We collect all $d\times d$ correlation matrices in the set $\mathsf{Corr}_d \coloneqq \{ Z\in \PSD_d \, : \, \operatorname{diag}Z = \mathbb{I}_d \}$. We now show that for LDUI/CLDUI matrices with $A=\mathbb{J}_d$, the PPT propetry is equivalent to separability.
\begin{proposition} \label{prop:A=J-LDUI/CLDUI-sep}
For $(\mathbb{J}_d,B)\in \MLDUI{d}$, the following are equivalent for $i=1,2$:
\begin{equation}
    B\in \mathsf{Corr}_d \iff  X^{(i)}_{(\mathbb{J}_d,B)} \text{ is PPT} \iff X^{(i)}_{(\mathbb{J}_d,B)} \text{ is separable}.
\end{equation}
\end{proposition}


\begin{proof}
    The first equivalence is a straightforward consequence of Lemma~\ref{lemma:LDOI-psd-ppt} (combined with Remark~\ref{remark:(C)LDUIsubspaceLDOI}): since $A = \mathbb J_d$, the condition $1 \geq |B_{ij}|^2$ corresponds to the $i,j$-minor of $B$ being non-negative. In the final equivalence, the reverse implication is trivial to obtain. To establish the forward implication, we need to show that $B\in \mathsf{Corr}_d \implies (\mathbb{J}_d, B)$ is PCP (see Theorem~\ref{theorem:LDOI-sep}). Hence, assume that $B\in \mathsf{Corr}_d$ and consider a decomposition of the form $B = WW^*$, where $W\in \M{d,r}$ with $r\geq \operatorname{rank}(B)$. Now, by choosing $V=\mathbb{J}_{d,r}$ (the $d\times r$ matrix with all entries equal to one), we see that $V,W$ form a PCP decomposition of $(\mathbb{J}_d, B)$:
    \begin{align*}
        [(V\odot \overbar{V})(W\odot \overbar{W})^*]_{ij} &= \sum_{k=1}^r |V_{ik}|^2 |W_{jk}|^2 = \sum_{k=1}^r |W_{jk}|^2 = B_{jj} = 1, \\
        [(V\odot W)(V\odot W)^*]_{ij} &= [WW^*]_{ij} = B_{ij},
    \end{align*}
    and the proof is complete.
\end{proof}
 
 It seems wise to pause here for a moment to collect some useful facts about the set of correlation matrices $\mathsf{Corr}_d$ in $\M{d}$:
 \begin{itemize}
     \item By definition, $Z\in \mathsf{Corr}_d \iff Z\in \PSD_d$ and $\operatorname{diag}Z = \mathbb{I}_d$.
     \item $\mathsf{Corr}_d$ is a compact convex set, with the rank one matrices $\ketbra{z}{z}$ being the obvious extreme points, for $\ket{z}\in \mathbb{T}^d \coloneqq \{\ket{y} \in \mathbb{C}^d : |y_i| = 1 \,\, \forall i\}$.
     \item For $d\leq 3$, it can be shown that the rank one matrices are the only extreme correlation matrices, and  hence $\mathsf{Corr}_d = \operatorname{conv} \{ \ketbra{z}{z} : \ket{z}\in \mathbb{T}^d \}$. However, for $d\geq 4$, other higher rank extreme points exist, and the preceding conclusion becomes false, see \cite{Christensen1979corr, Loewy1980corr, Grone1990corr, Li1994corr}.
 \end{itemize}
\end{example}

Now, if we consider the general LDOI matrices with $A=\mathbb{J}_d$, we can quickly infer that as before, these matrices are PPT if and only if the associated matrices $B,C \in \mathsf{Corr}_d$. However, we do not know whether this also suffices to guarantee separability. We do have some partial results in this direction:
\begin{itemize}
    \item $B\in \mathsf{Corr}_d \iff X^{(3)}_{(\mathbb{J}_d, B,B)}$ is separable, see Proposition~\ref{prop:AB-ABB}.
    \item $B,C \in \operatorname{conv}\{\ketbra{z}{z} : \ket{z}\in \mathbb{T}^d \} \implies X^{(3)}_{(\mathbb{J}_d, B,C)}$ is separable. Let us quickly prove this. Assume $B=\ketbra{b}{b}$ and $C=\ketbra{c}{c}$ are arbitrary rank one correlations, where $\ket{b},\ket{c}\in \mathbb{T}^d$. Then, we can easily construct vectors $\ket{v},\ket{w}\in \mathbb{T}^d$ such that $\ket{v\odot w}=\ket{b}$ and $\ket{v\odot \overbar{w}}=\ket{c}$. Hence, we have $\mathbb{J}_d = \ketbra{v\odot \overbar{v}}{w\odot \overbar{w}}$, $B=\ketbra{v\odot w}{v\odot w}$ and $C=\ketbra{v\odot \overbar{w}}{v\odot \overbar{w}}$, implying that $(\mathbb{J}_d, B,C)$ is TCP. The desired result then follows from convexity.
\end{itemize}

\begin{example}[\emph{Canonical NPT states}]\ \\[0.1cm] \label{eg:states-CanNPT}
Extraction of maximally entangled states from several copies of a given bipartite quantum state through the use of local operations and classical communication (LOCC) forms a central task in numerous quantum communication and cryptographic protocols. In the asymptotic limit, if a state $\rho\in \M{d_1}\otimes \M{d_2}$ allows for a non-zero rate of extraction (defined as the ratio of the extracted number of maximally entangled states to the number of input states), it is said to be \emph{distillable}. Non-distillable entangled bipartite states are called \emph{bound entangled}.  It is well known that distillable states $\rho$ are negative under partial transposition (NPT), i.e., $(\operatorname{id}\otimes \top)\rho$ is not positive semi-definite. However, it is not known whether every NPT state is distillable or not. Put differently, the existence of an NPT bound entangled state is uncertain (see the excellent review articles \cite{Horodecki2001distillation, Lewenstein2000distillation} for a more precise formulation of these concepts). In \cite{Shor2000NPT}, the authors show that every NPT state $\rho\in \M{d_1}\otimes \M{d_2}$ is LOCC-transformable to an NPT state in the following family of states in $\M{d}\otimes \M{d}$ ($d\leq \operatorname{min}(d_1,d_2)$), where for $i\neq j$, $\ket{\psi^\pm_{ij}} = (\ket{ij} \pm \ket{ji})/\sqrt{2}$:
\begin{equation}
    \rho_{a,b,c} = a\sum_{i=1}^d \ketbra{ii}{ii} + b\sum_{1\leq i<j \leq d} \ketbra{\psi^-_{ij}}{\psi^-_{ij}} + c\sum_{1\leq i<j \leq d} \ketbra{\psi^+_{ij}}{\psi^+_{ij}}
\end{equation}
and $a,b,c$ are real parameters. It is not too hard to discern that the states $\rho_{a,b,c}\in \M{d}\otimes \M{d}$ are LDUI, with the associated matrices $A,B \in \Mreal{d}$ defined as follows:
\begin{equation}
    A_{ij} = \begin{cases}
    a,  & \text{if } i=j \\
    (c+b)/2,  & \text{otherwise},
\end{cases} \quad B_{ij} = \begin{cases}
    a,& \text{if } i=j \\
    (c-b)/2,  & \text{otherwise} .
\end{cases}
\end{equation}
Hence, to show that every NPT bipartite state is distillable, it suffices to prove distillability for an arbitrary NPT LDUI state in the above family. The parameter ranges within which the above states are PPT/NPT can easily be calculated using Lemma~\ref{lemma:LDOI-psd-ppt}, see also \cite[Figure 2]{Shor2000NPT}.
\end{example}

\begin{example}[\emph{PPT entangled edge states}]\ \\[0.1cm] \label{eg:states-edge}
A PPT entangled matrix $X\in \M{d_1}\otimes \M{d_2}$ is said to be an \emph{edge state} if there are no product vectors $\ket{xy}\in \operatorname{range}(X)$ such that $\ket{\overbar{x}y}\in \operatorname{range}(X^\Gamma)$. These states defy the range criterion of separability in a very extreme fashion \cite{choi1982positive,Lewenstein2000edge, Lewenstein2001edge}. In recent years, there has been a great deal of interest in characterizing edge states $X$ based on their types, which are nothing but pairs of numbers $(p,q)$ such that $\operatorname{rank}(X)=p$ and $\operatorname{rank}(X^\Gamma)=q$. In particular, for the low dimensional $3\otimes 3$ system, all possible types have been identified, and examples for each type have been constructed (see \cite{Kiem2011edge} and references therein for a review of all the examples). Noticeably, the authors in \cite{Kye2012edge} show that except for the $(4,4)$ type, all other types of $3\otimes 3$ edge states can be generated by matrices of the form:
\begin{equation} \label{eq:LDOI-edge-3}
    X =  \left(\mathcode`0=\cdot
\begin{array}{ *{3}{c} | *{3}{c} | *{3}{c} }
   e^{i\theta} + e^{-i\theta} & 0 & 0 & 0 & -e^{i\theta} & 0 & 0 & 0 & -e^{-i\theta} \\
   0 & 1/b & 0 & \langle\eta|\xi\rangle & 0 & 0 & 0 & 0 & 0 \\
   0 & 0 & b & 0 & 0 & 0 & \langle\zeta|\xi\rangle & 0 & 0 \\\hline
   0 & \langle\xi|\eta\rangle & 0 & b & 0 & 0 & 0 & 0 & 0 \\
   -e^{-i\theta} & 0 & 0 & 0 & e^{i\theta} + e^{-i\theta} & 0 & 0 & 0 & -e^{i\theta} \\
   0 & 0 & 0 & 0 & 0 & 1/b & 0 & \langle\zeta|\eta\rangle & 0 \\ \hline
   0 & 0 & \langle\xi|\zeta\rangle & 0 & 0 & 0 & 1/b & 0 & 0 \\
   0 & 0 & 0 & 0 & 0 & \langle\eta|\zeta\rangle & 0 & b & 0 \\
   -e^{i\theta} & 0 & 0 & 0 & -e^{-i\theta} & 0 & 0 & 0 & e^{i\theta} + e^{-i\theta} \\
  \end{array}
\right)
\end{equation}
where $b>0, \,\, -\pi/3 < \theta < \pi/3 \,\, (\theta\neq 0)$ and
$\ket{\eta}, \ket{\zeta}, \ket{\xi}\in \C{3}$ are vectors such that
$$\langle\eta|\eta\rangle = \langle\zeta|\zeta\rangle = \langle\xi|\xi\rangle = e^{i\theta}+e^{-i\theta},$$
$$ |\langle\eta|\xi\rangle| \leq 1, \quad |\langle\xi|\zeta\rangle| \leq 1, \quad |\langle\zeta|\eta\rangle| \leq 1. $$
The zero pattern of these states immediately reveal that they are all $3\otimes 3$ LDOI matrices. The entries of the associated $A,B$ and $C$ matrices can be read off from the appropriate diagonal/off-diagonal elements of $X$.
\end{example}

We summarize all the examples discussed in this section in Table \ref{tbl:LDOI}.

\begin{table}[htb]
\begin{tabular}{@{}|l|l|l|l|l|l|@{}}
\toprule
  \cellcolor[HTML]{DCD0F4}\textbf{Ex.} & \cellcolor[HTML]{DCD0F4}\textbf{Name}                                                          & \cellcolor[HTML]{DCD0F4}\textbf{Defining Characteristic}                                                                                                & \cellcolor[HTML]{DCD0F4}\begin{tabular}[c]{@{}l@{}}\textbf{Ambient} \\ \textbf{Space}\end{tabular} & \cellcolor[HTML]{DCD0F4}\begin{tabular}[c]{@{}l@{}}\textbf{Associated} \\ \textbf{$(A,B,C)$}\end{tabular}                                                                                                                         & \cellcolor[HTML]{DCD0F4}\textbf{References} \\ \midrule
\ref{eg:states-diag} & Diagonal                                                                              & Diagonal matrices in $\M{d^2}$                                                                                                                & \begin{tabular}[c]{@{}l@{}}$\LDUI_d\,\, \cap$ \\ $\CLDUI_d$\end{tabular}               & \begin{tabular}[c]{@{}l@{}}$A\in \M{d}$\\ $B=\operatorname{diag}(A)$\end{tabular}                                                                             & ---                                \\ \midrule
\cellcolor[HTML]{DCD0F4}\ref{eg:states-unitrank} & \cellcolor[HTML]{DCD0F4}Unit rank                                                                               &\cellcolor[HTML]{DCD0F4}\begin{tabular}[c]{@{}l@{}}Unit rank LDOI matrices  \\$X=\ketbra{Y}{Z}$; $\ket{Y},\ket{Z}\in \C{d}\otimes \C{d}$ \end{tabular}                                                             &\cellcolor[HTML]{DCD0F4}$\LDOI_d$ &\cellcolor[HTML]{DCD0F4}\begin{tabular}[c]{@{}l@{}} see \\Example~\ref{eg:states-unitrank} \end{tabular}                                                                        & \cellcolor[HTML]{DCD0F4}---                                \\ \midrule
\ref{eg:states-unitrank} &  \begin{tabular}[c]{@{}l@{}}Tensor \\ product \end{tabular}                                                                              & \begin{tabular}[c]{@{}l@{}}Tensor product LDOI matrices \\ $X=Y\otimes Z$; $Y,Z\in \M{d}$ \end{tabular}                                                                                                                & $\LDOI_d$ & \begin{tabular}[c]{@{}l@{}} see \\Example~\ref{eg:states-unitrank} \end{tabular}                                                                        & ---                                \\ \midrule
\cellcolor[HTML]{DCD0F4}\ref{eg:states-wer-iso} & \cellcolor[HTML]{DCD0F4}Werner                                                        & \cellcolor[HTML]{DCD0F4}\begin{tabular}[c]{@{}l@{}}$(U\otimes U)X(U^*\otimes U^*) = X$ for\\ all unitary matrices $U\in \M{d}$\end{tabular}   & \cellcolor[HTML]{DCD0F4}$\LDUI_d$                                                & \cellcolor[HTML]{DCD0F4}\begin{tabular}[c]{@{}l@{}}$A = b\,\mathbb{I}_d + a\mathbb{J}_d$ \\ $C = a\mathbb{I}_d + b\mathbb{J}_d$\\ $a,b \in \C{ }$\end{tabular} & \cite{Werner1989}   \cellcolor[HTML]{DCD0F4}           \\ \midrule
\ref{eg:states-wer-iso} & Isotropic                                                                             & \begin{tabular}[c]{@{}l@{}}$(U\otimes U^*)X(U^*\otimes U) = X$ for \\ all unitary matrices $U\in \M{d}$\end{tabular}                          & $\CLDUI_d$                                                                       & \begin{tabular}[c]{@{}l@{}}$A = b\,\mathbb{I}_d + a\mathbb{J}_d$\\ $B = a\mathbb{I}_d + b\mathbb{J}_d$\\ $a,b \in \C{ }$\end{tabular}                          &      \cite{Horodecki1999iso}                              \\ \midrule
\cellcolor[HTML]{DCD0F4}\ref{eg:states-dicke} & \cellcolor[HTML]{DCD0F4}\begin{tabular}[c]{@{}l@{}}Diagonal \\ symmetric\end{tabular} & \cellcolor[HTML]{DCD0F4}\begin{tabular}[c]{@{}l@{}}Diagonal in the Dicke basis of \\ the $\C{d}\otimes \C{d}$ symmetric subspace\end{tabular} & \cellcolor[HTML]{DCD0F4}$\LDUI_d$                                                & \cellcolor[HTML]{DCD0F4}\begin{tabular}[c]{@{}l@{}}$A\in \M{d}$\\ $C=A$\end{tabular}                                                                            & \begin{tabular}[c]{@{}l@{}} \cite{tura2018separability}, \\ \cite{yu2016separability} \end{tabular} \cellcolor[HTML]{DCD0F4}           \\ \midrule
\ref{eg:states-PTinv} & PT invariant                                                                          & \begin{tabular}[c]{@{}l@{}}LDOI matrices invariant under \\ partial transposition\end{tabular}                                                & $\LDOI_d$                                                                        & \begin{tabular}[c]{@{}l@{}l@{}}$A,B\in \M{d}$ \\ $C=B^\top$ or \\ $C=B$\end{tabular}                                                                                   & ---                                \\ \midrule
\cellcolor[HTML]{DCD0F4}\ref{eg:states-A=J} & \cellcolor[HTML]{DCD0F4}LDOI A=J                                                      & \cellcolor[HTML]{DCD0F4}LDOI matrices with $A=\mathbb{J}_d$                                                                                  & \cellcolor[HTML]{DCD0F4}$\LDOI_d$                                                & \cellcolor[HTML]{DCD0F4}\begin{tabular}[c]{@{}l@{}}$A=\mathbb{J}_d$\\ $B,C\in \M{d}$\end{tabular}                                                               & \cellcolor[HTML]{DCD0F4}---        \\ \midrule
\ref{eg:states-CanNPT} & \begin{tabular}[c]{@{}l@{}}Canonical \\ NPT\end{tabular}                                                                        & \begin{tabular}[c]{@{}l@{}}Every NPT state is LOCC-\\ transformable to a state in this \\ class\end{tabular}                                  & $\LDUI_d$                                                                        &  \begin{tabular}[c]{@{}l@{}} see \\Example~\ref{eg:states-CanNPT}\end{tabular}                                                                                                                         &   \cite{Shor2000NPT}                                \\ \midrule
\cellcolor[HTML]{DCD0F4}\ref{eg:states-edge} & \cellcolor[HTML]{DCD0F4}$3\otimes 3$ edge                                             & \cellcolor[HTML]{DCD0F4}\begin{tabular}[c]{@{}l@{}}$3\otimes 3$ edge states of all types \\ except $(4,4)$\end{tabular}                       & \cellcolor[HTML]{DCD0F4}$\LDOI_d$                                                & \cellcolor[HTML]{DCD0F4}\begin{tabular}[c]{@{}l@{}} see \\Example~\ref{eg:states-edge} \end{tabular}                                                                                                & \cite{Kye2012edge} \cellcolor[HTML]{DCD0F4}           \\ \bottomrule
\end{tabular}
\medskip
\caption{Important classes of bipartite matrices in $\LDOI_d$. For more details, see the appropriate examples.}
\label{tbl:LDOI}
\end{table}

\section{Linear structure of LDOI matrices}\label{sec:linear-structure}

We discuss in this section the linear structure of the sets of bipartite matrices with the invariance properties discussed in Section \ref{sec:LDUI-CLDUI-LDOI}, focusing on different notions of symmetry. We shall make no reference to any positivity or separability notions, this being the topic of the next section. 

It was shown in \cite[Sections 6-7]{nechita2021graphical} that the sets $\CLDUI_d$, $\LDUI_d$, $\LDOI_d$ are $\mathbb C$-vector spaces with dimensions
\begin{align*}
    \dim_{\mathbb C}  \CLDUI_d = \dim_{\mathbb C}  \LDUI_d &= 2d^2-d,\\
    \dim_{\mathbb C}  \LDOI_d &= 3d^2-2d.
\end{align*}
This fact is a consequence of a simple dimension counting for the triples $(A,B,C)$ using the different diagonal restrictions. For example, in the $\CLDUI_d$ case, the elements are parametrized by pairs $(A,B)$ with $A,B$ being $d \times d$ complex matrices (having, in total, $2d^2$ complex parameters), with the restriction that $\operatorname{diag}(A) = \operatorname{diag}(B)$ (fixing $d$ complex parameters), resulting in a total of $2d^2-d$ complex parameters, see Eq.~\eqref{eq:MLDUI} and Proposition \ref{prop:LDOI-ABC}. It is also trivial to see from their definition that $\CLDUI_d$ and $\LDUI_d$ are vector subspaces of $\LDOI_d$ (see Remark~\ref{remark:(C)LDUIsubspaceLDOI}). Moreover, $\CLDUI_d \cap \LDUI_d = \operatorname{diag}(\mathcal M_{d^2}(\mathbb C))$ is precisely the subspace of diagonal matrices in $\M{d}\otimes \M{d}$, see also Example \ref{eg:states-diag}. We represent the general position of these vector spaces in Figure \ref{fig:linear-sets}.
\begin{figure}[H]
    \centering
    \includegraphics{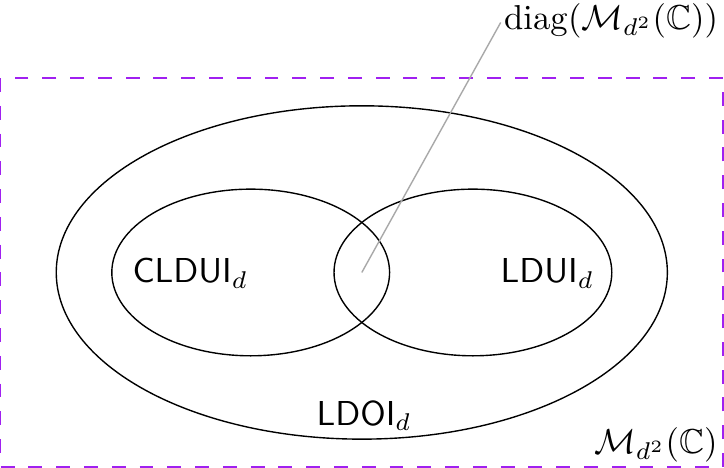}
    \caption{The relative position of the $\mathbb C$-vector spaces $\CLDUI_d$, $\LDUI_d$, $\LDOI_d$.}
    \label{fig:linear-sets}
\end{figure}   

As elements of $\mathcal M_{d^2}(\mathbb C)$, the LDUI/CLDUI/LDOI matrices have interesting block structures, which can be derived by using the explicit form of the isomorphisms stated after Proposition~\ref{prop:LDOI-ABC}, see also \cite[Eq.~(41)]{nechita2021graphical}.
\begin{proposition}\label{prop:block-structure}
For all triples $(A,B,C) \in \MLDOI{d}$, we have the following decompositions:
\begin{align*}
    X^{(1)}_{(A,C)} &= \left(\bigoplus_i [A_{ii}] \right) \oplus \left( \bigoplus_{i < j} \begin{bmatrix} A_{ij} & C_{ij} \\ C_{ji} & A_{ji} \end{bmatrix} \right),\\
    X^{(2)}_{(A,B)} &= B \oplus \bigoplus_{i \neq j}[A_{ij}],\\
    X^{(3)}_{(A,B,C)} &= B \oplus \left( \bigoplus_{i < j} \begin{bmatrix} A_{ij} & C_{ij} \\ C_{ji} & A_{ji} \end{bmatrix} \right).
\end{align*}
\end{proposition}
\begin{corollary}\label{cor:rank-X}
The rank of the invariant matrices $X^{(1,2,3)}$ above are as follows:
\begin{align*}
    \operatorname{rank} X^{(1)}_{(A,C)} &= |\{i \in [d] \, : \, A_{ii} \neq 0\}| + \sum_{i<j} \operatorname{rank} \begin{bmatrix} A_{ij} & C_{ij}, \\ C_{ji} & A_{ji} \end{bmatrix} \\
    \operatorname{rank} X^{(2)}_{(A,B)} &= \operatorname{rank} B  + |\{ (i,j) \in [d]^2 \, : \, i \neq j \text{ and } A_{ij} \neq 0\}|,\\
    \operatorname{rank} X^{(3)}_{(A,B,C)} &= \operatorname{rank} B + \sum_{i<j} \operatorname{rank} \begin{bmatrix} A_{ij} & C_{ij} \\ C_{ji} & A_{ji} \end{bmatrix}.
\end{align*}
\end{corollary}

We now consider various symmetries of the vector spaces $\LDUI_d, \CLDUI_d$ and $\LDOI_d$, obtained by permuting the tensor legs (i.e.~the wires in the graphical representation corresponding to the different tensor factors) of the relevant bipartite matrices. 

\begin{proposition}\label{prop:leg-permutations}
The vector subspaces $\CLDUI_d, \LDUI_d, \LDOI_d \subseteq \mathcal M_{d^2}(\mathbb C)$ are invariant under pairwise leg permutations: 
\begin{alignat*}{2}
F X^{(3)}_{(A,B,C)} F &= X^{(3)}_{(A^\top,B,C^\top)}, \qquad\qquad \left(X^{(3)}_{(A,B,C)}\right)^\top &&= X^{(3)}_{(A,B^\top,C^\top)}, \\
\left(X^{(3)}_{(A,B,C)}\right)^\times &= X^{(3)}_{(A^\top,B^\top,C)}, \qquad\qquad
\left(X^{(3)}_{(A,B,C)}\right)^R &&= X^{(3)}_{(B,A,C)},
\end{alignat*}
where $F$ is the flip operator ($F \ket{ab} = \ket{ba}$), $\times$ permutes the legs of a 4-tensor diagonally and $R$ is the realignment operation. Moreover, the space $\LDOI_d$ has the following additional symmetries:
\begin{alignat*}{2}
\left(X^{(3)}_{(A,B,C)}\right)^\Gamma &= X^{(3)}_{(A,C,B)}, \qquad\qquad
FX^{(3)}_{(A,B,C)} &&= X^{(3)}_{(C^\top,B,A^\top)}, \\
\left(X^{(3)}_{(A,B,C)}\right)^\text{\reflectbox{$\Gamma$}}&= X^{(3)}_{(A,C^\top,B^\top)}, \quad\qquad X^{(3)}_{(A,B,C)}F &&= X^{(3)}_{(C,B,A)}.
\end{alignat*} 
\end{proposition}
\begin{proof}
    All the above relations are easily verified using algebra, but expressing them in the graphical language of tensor networks is more insightful. We leave the detailed proofs to the reader, and provide a graphical proof for the flip $F X^{(3)}_{(A,B,C)} F$ operation in Figure \ref{fig:F-X-F}. Finally, Remark~\ref{remark:(C)LDUIsubspaceLDOI} can be used to derive analogous results for the $\CLDUI_d$ and $\LDUI_d$ subspaces.
\begin{figure}[H]
    \centering
    \includegraphics{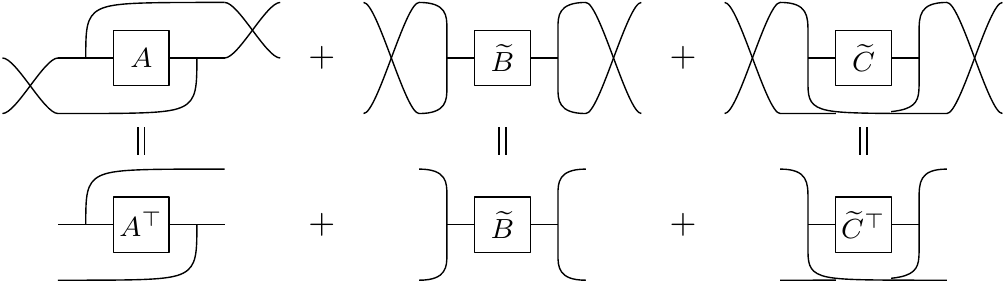}
    \caption{The conjugation of the LDOI matrix $X^{(3)}_{(A,B,C)}$ by the flip operator. The two-sided action of the flip operator on each of the matrices $A, \widetilde{B}, \widetilde{C}$ amounts to transposing the first and the last. \qedhere}
    \label{fig:F-X-F}
\end{figure}    
\end{proof}

\begin{remark}
There are more tensor-leg permutations that leave invariant the $\LDOI_d$ matrices, by permuting the $A,B,C$ matrices and adding transposes to them. The $4!=24$ possible leg permutations correspond to $3! \times 2^3 / 2$ such operations on the ordered triples $(A,B,C)$, where the $3!$ factor corresponds to permuting the matrices, $2^3$ corresponds to the choice of adding or not a transpose on each letter, and the division by 2 corresponds to the requirement that the number of transposes should be even. We have presented in Proposition \ref{prop:leg-permutations} the most relevant ones for the quantum information theory, see Proposition \ref{prop:leg-permutations-map} for the corresponding linear map symmetries. 
\end{remark}

The considerations above allow us to characterize \emph{self-adjoint} LDOI matrices as follows. 

\begin{proposition}
$X^{(3)}_{(A,B,C)} \in \LDOI_d$ is self-adjoint $\iff A\in \Mreal{d}$ and $B,C\in \Msa{d}$. In particular, the real vector spaces of self-adjoint invariant matrices have the following dimensions: 
    \begin{align*}
    \dim_{\mathbb R}  \CLDUI_d^{sa} = \dim_{\mathbb R}  \LDUI_d^{sa} &= 2d^2-d,\\
    \dim_{\mathbb R}  \LDOI_d^{sa} &= 3d^2-2d.
\end{align*}
\end{proposition}
\begin{proof}
    The main assertion follows from Proposition~\ref{prop:leg-permutations}:
    $$\left(X^{(3)}_{(A,B,C)}\right)^* = \overline{\left(X^{(3)}_{(A,B,C)}\right)^\top} = \overline{X^{(3)}_{(A,B^\top,C^\top)}} = X^{(3)}_{(\bar A, B^*,C^*)}. $$
\end{proof}

\begin{proposition}
A matrix $X=X^{(3)}_{(A,B,C)} \in \LDOI_d$ is \emph{symmetric} (i.e.~$X=FXF$) if and only if
$$A=A^\top \quad \text{ and } \quad C = C^\top.$$
Moreover, $X$ is \emph{Bose-symmetric} (i.e.~$X = P_s X P_s$, where $P_s$ is the orthogonal projection on the symmetric subspace of $\mathbb C^d \otimes \mathbb C^d$) if and only if
$$A=C=A^\top=C^\top.$$
\end{proposition}
\begin{proof}
    For the first claim, use Proposition~\ref{prop:leg-permutations}. For the second claim, use $P_s = (I+F)/2$ and Proposition \ref{prop:leg-permutations} to get
    \begin{align*}
        4A &= A + C^\top + C + A^\top,\\
        4C &= C + A^\top + A + C^\top,
    \end{align*}
    from which the conclusion immediately follows. 
\end{proof}

We now show that the $\CLDUI_d$, $\LDUI_d$, and $\LDOI_d$ vector spaces are stable under a modified notion of direct sum, that we introduce next. 

\begin{definition} \label{def:biparty_directsum}
    Given two bipartite matrices $X_i \in \mathcal M_{d_i^2}(\mathbb C)$, $i=1,2$, we define their \emph{bipartite direct sum} $X_1 \boxplus X_2 \in \mathcal M_{(d_1+d_2)^2}(\mathbb C)$ by 
    $$[X_1 \boxplus X_2]((i_1,i_2), (j_1,j_2)) = \begin{cases}
    X_1((i_1,i_2), (j_1,j_2)) & \quad \text{ if } i_1,i_2,j_1,j_2 \leq d_1\\
    X_2((i_1-d_1,i_2-d_1), (j_1-d_1,j_2-d_1)) & \quad \text{ if } i_1,i_2,j_1,j_2 > d_1\\
    0 & \quad \text{ otherwise.}\end{cases}$$
\end{definition}

With this definition in hand, we have the following result, showing that the vector spaces of invariant bipartite matrices are stable under bipartite direct sum. 

\begin{proposition} \label{prop:directsum_triple}
Given matrices $A_i,B_i,C_i \in \mathcal M_{d_i}(\mathbb C)$, $i=1,2$, the following relations hold: 
\begin{align*}
    X^{(1)}_{(A_1, B_1)} \boxplus X^{(1)}_{(A_2, B_2)} &= X^{(1)}_{(A_1 \oplus A_2, B_1 \oplus B_2)},\\
    X^{(2)}_{(A_1, B_1)} \boxplus X^{(2)}_{(A_2, B_2)} &= X^{(2)}_{(A_1 \oplus A_2, B_1 \oplus B_2)},\\
    X^{(3)}_{(A_1, B_1, C_1)} \boxplus X^{(2)}_{(A_2, B_2, C_2)} &= X^{(3)}_{(A_1 \oplus A_2, B_1 \oplus B_2, C_1 \oplus C_2)}.
\end{align*}
\end{proposition}
\begin{proof}
    The result directly follows from the explicit (coordinate-wise) expressions of the bijections $X^{(i)}$ for $i=1,2,3$, see Equations \eqref{eq:X1-coord}, \eqref{eq:X2-coord} and \eqref{eq:X3-coord}. 
\end{proof}

\begin{remark}
The fact that the vector spaces of diagonal unitary-invariant matrices are stable under the notion of direct sum defined above is a consequence of the fact that diagonal unitary (resp.~orthogonal) matrices of size $d_1+d_2$ are direct sums of diagonal unitary (resp.~orthogonal) matrices of sizes $d_1$ and $d_2$. The same does not hold for \emph{tensor products}, so we do not have a similar stability property with respect to tensor products. 
\end{remark}

For $A\in \M{d}$ and some index set $I\subset [d]$, $A[I]\in \M{|I|}$ denotes the principal submatrix of $A$ with rows and columns indexed by the indices present in $I$. In the same vein, for an index set $I\subset [d]$ and a triple $(A,B,C)\in \MLDOI{d}$, we define the principal subtriple of $(A,B,C)$ associated with the index set $I$ to be the triple $(A[I],B[I],C[I])\in \MLDOI{|I|}$. Our next result discusses the effect of this action of taking principal subtriples in the $\LDOI_d$ space.

\begin{proposition}\label{prop:principal_triple}
For a triple $(A,B,C)\in \MLDOI{d}$ and an index set $I\subset [d]$, define the projector $P=\sum_{i\in I} \ketbra{i}{i} \in \Msa{d}$. Then, the following relation holds:
\begin{equation*}
    (P\otimes P)X^{(3)}_{(A,B,C)}(P\otimes P) = X^{(3)}_{(A[I],B[I],C[I])},
\end{equation*}
where the locally projected $(P\otimes P)X^{(3)}_{(A,B,C)}(P\otimes P)$ is understood to lie in $\M{|I|}\otimes \M{|I|}$. 
\end{proposition}
\begin{proof}
    Follows trivially from the explicit form of the bijections, stated after Proposition~\ref{prop:LDOI-ABC}.
\end{proof}

We will see in the next section that the cones of PCP and TCP matrices in $\MLDOI{d}$ are stable under the discussed operations of taking direct sums and principal subtriples. 
We now conclude this section by discussing the partial actions on $\LDOI_d$ of the two conditional expectations on the diagonal and scalar matrices defined below:
\begin{alignat*}{2}
\operatorname{diag} : \mathcal M_d(\mathbb C) &\to \mathcal M_d(\mathbb C) \qquad\qquad\qquad \operatorname{trace} : \mathcal M_d(\mathbb C) &&\to \mathcal M_d(\mathbb C)  \\
Z\,\, &\mapsto \sum_{i=1}^dZ_{ii}\ketbra{i}{i}, \qquad\qquad\qquad\qquad\quad Z &&\mapsto \frac{\operatorname{Tr} Z}{d}I_d.
\end{alignat*}
    
\begin{proposition}\label{prop:conditional-expectations}
    For $(A,B,C)\in \MLDOI{d}$, define the diagonal matrices $A^{\mathsf{row}}$ and $A^{\mathsf{col}}$ in $\M{d}$ with entries equal to the row and column sums of $A$ respectively: $A^{\mathsf{row}}_{ii}=\sum_{j=1}^d A_{ij}$ and $A^{\mathsf{col}}_{ii}=\sum_{j=1}^d A_{ji} \,\, \forall i\in [d]$. Then, the following relations hold:
    \begin{align*}
        [\operatorname{id} \otimes \operatorname{Tr}] X^{(3)}_{(A,B,C)} &= A^{\mathsf{row}},\\
        [\operatorname{Tr} \otimes \operatorname{id}]X^{(3)}_{(A,B,C)} &= A^{\mathsf{col}}.
    \end{align*}
    The action of the partial conditional expectations $\operatorname{diag}$ and $\operatorname{trace}$ are given respectively by
    \begin{alignat*}{2}
        [\operatorname{id} \otimes \operatorname{diag}]X^{(3)}_{(A,B,C)} &= X^{(3)}_{(A, \operatorname{diag}A, \operatorname{diag}A)}, \qquad\qquad  [\operatorname{id} \otimes \operatorname{trace}]X^{(3)}_{(A,B,C)} &&= \frac 1 d  X^{(3)}_{(A^{\mathsf{row}},A^{\mathsf{row}},A^{\mathsf{row}})},  \\
        [\operatorname{diag} \otimes \operatorname{id}]X^{(3)}_{(A,B,C)} &= X^{(3)}_{(A, \operatorname{diag}A, \operatorname{diag}A)}, \qquad\qquad [\operatorname{trace} \otimes \operatorname{id}]X^{(3)}_{(A,B,C)} &&= \frac 1 d X^{(3)}_{(A^{\mathsf{col}},A^{\mathsf{col}},A^{\mathsf{col}})}.
    \end{alignat*}
    We have thus (see also \cite[Lemmas 6.9 and 7.7]{nechita2021graphical})
    $$\operatorname{Tr} X^{(3)}_{(A,B,C)} = \operatorname{Tr}A^{\mathsf{row}} = \operatorname{Tr}A^{\mathsf{col}} = \sum_{i,j=1}^d A_{ij}.$$
\end{proposition}
\begin{proof}
    These are simple computations which can be easily obtained using the graphical calculus for tensors, see e.g.~Figure \ref{fig:Tr-X} for the trace formula.
\begin{figure}[htbp]
    \centering
    \includegraphics{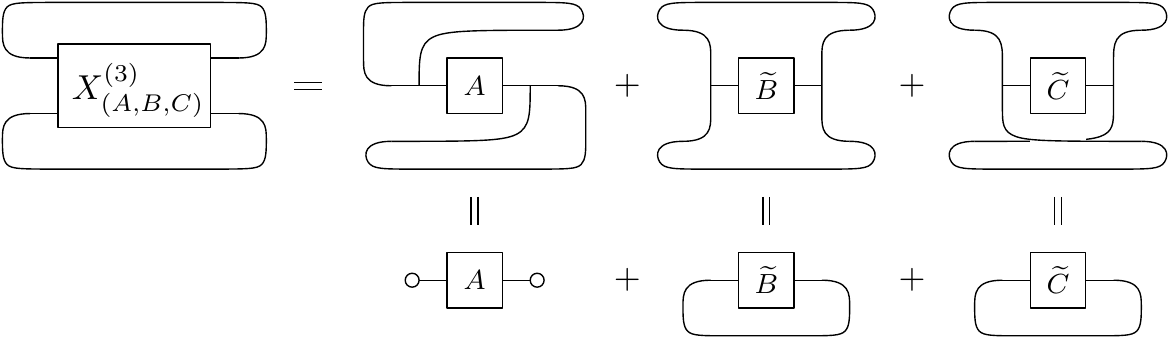}
    \caption{ The trace of an LDOI matrix. In the second row of diagrams, the first term corresponds to $\sum_{i,j} A_{ij}$, while the second and the third ones are zero, since $\tilde B$ and $\tilde C$ have zeros on the diagonal.}
    \label{fig:Tr-X}
\end{figure}
\end{proof}

\section{Convex structure of LDOI matrices}\label{sec:convex-structure}
We study in this section the properties of the three sets of local diagonal invariant matrices from the point of view of convex geometry; more precisely, we shall investigate the convex cone of positive semidefinite LDUI/CLDUI/LDOI matrices, as well as other notions of positivity related to quantum entanglement. We start with some basic facts about convex cones. 

\begin{definition}
    Let $V$ be a real vector space. A \emph{convex cone} $\mathcal C$ is a subset of $V$ having the following two properties: 
    \begin{itemize}
        \item if $x \in \mathcal C$ and $\lambda \in \mathbb R_+ = [0, \infty)$, then $\lambda x \in \mathcal C$.
        \item if $x,y \in \mathcal C$, then $x+y \in \mathcal C$.
    \end{itemize}
    In particular, $0 \in \mathcal C$. The cone $\mathcal C$ is said to be \emph{pointed} if $\mathcal C \cap (-\mathcal C) = \{0\}$; in other words, $\mathcal C$ is pointed if it does not contain any line. 
    
    Given a cone $\mathcal C$, we define its \emph{dual cone} by
    $$\mathcal C^* \:= \{\alpha \in V^* \, : \, \langle \alpha, x \rangle \geq 0, \, \forall x \in \mathcal C\} \subseteq V^*,$$
    where $V^*$ is the vector space dual to $V$. 
    
    For a vector $v \neq 0$, a half-line $\mathbb R_+ v \subseteq \mathcal C$ is called an \emph{extremal ray} of $\mathcal C$ (we write $\mathbb R_+ v \in \operatorname{ext} \mathcal C$) if
    $$v = x+y \text { for } x,y \in \mathcal C \implies x,y \in \mathbb R_+ v.$$
\end{definition}

\begin{example}
The main examples we shall be concerned with in this work are the cone of \emph{entrywise non-negative matrices}
$$\EWP_d = \{A \in \mathcal M_d(\mathbb R) \, : \, A_{ij} \geq 0, \, \forall i,j \in [d]\}$$
and the cone of \emph{positive semidefinite matrices}
$$\PSD_d = \{B \in \Msa{d} \, : \, \langle x| B |x \rangle \geq 0, \, \forall \ket{x} \in \mathbb C^d\}.$$
Importantly, the two cones $\EWP_d \subseteq \mathcal M_d(\mathbb R)$ and $\PSD_d \subseteq \mathcal M_d^{sa}(\mathbb C)$ are self-dual. Their extremal rays are as follows: 
$$\operatorname{ext} \EWP_d = \{\mathbb R_+ \ketbra{i}{j}\}_{i,j\in [d]} \quad \text{ and } \quad \operatorname{ext} \PSD_d = \{\mathbb R_+ \ketbra{x}{x}\}_{\ket{x} \in \mathbb C^d, \, \ket{x} \neq 0}.$$
\end{example}

Let $V$ and $W$ be vector spaces and consider convex cones $\mathcal B \subseteq V$ and $\mathcal C \subseteq W$. We define the following two operations on convex cones:
\begin{align}
    \text{Cartesian product:}\quad &\mathcal B \times \mathcal C = \{(b,c) \, : b \in \mathcal B, \, c \in \mathcal C\} \subseteq V \oplus W,\\
    \text{Direct sum:}\quad &\mathcal B \oplus \mathcal C = \operatorname{conv}\left(\{(b,0) \, : b \in \mathcal B\} \cup \{(0,c)\, : \, c \in \mathcal C\}\right) \subseteq V \oplus W.
\end{align}
For the corresponding construction for convex sets, see \cite{bremner1997complexity} or \cite[Section 3.1]{bluhm2020compatibility}. These constructions are dual to each other, see the references above:
$$(\mathcal B \times \mathcal C)^* = \mathcal B^* \oplus \mathcal C^* \quad \text{ and } \quad (\mathcal B \oplus \mathcal C)^* = \mathcal B^* \times \mathcal C^*.$$

We have the following result, giving the extremal rays of a Cartesian product of cones in terms of the extremal rays of the factors. By duality, a similar result could be given for the facets of a direct sum of convex cones.
\begin{proposition}
For two pointed cones $\mathcal B, \mathcal C$, we have 
$$\operatorname{ext}(\mathcal B \times \mathcal C) = \left(\operatorname{ext}(\mathcal B), 0 \right) \sqcup \left(0,\operatorname{ext}(\mathcal C) \right).$$
\end{proposition}
\begin{proof}
    We prove the equality of the two sets by showing separately the two inclusions. For ``$\subseteq$'', consider an extremal ray $\mathbb R_+(b,c)$ of $\mathcal B \times \mathcal C$. We first show that one of $b$ or $c$ must be null. Using $b \in \mathcal B$ and $c \in \mathcal C$, we also have 
    \begin{equation}\label{eq:half-b-c}
        \left( \frac{b}{2}, \frac{3c}{2}\right), \left( \frac{3b}{2}, \frac{c}{2}\right) \in \mathcal B \times \mathcal C.
    \end{equation}
    From $$(b,c) = \frac 1 2 \left[\left( \frac{b}{2}, \frac{3c}{2}\right) + \left( \frac{3b}{2}, \frac{c}{2}\right)\right],$$
    we deduce that the two vectors in \eqref{eq:half-b-c} must be collinear to $(b,c)$, and thus at least one of $b,c$ must be zero. Assuming, say, $c=0$, the extremality of $\mathbb R_+b$ inside $\mathcal B$ follows easily from the hypothesis. 
    
    For the reverse inclusion, let $\mathbb R_+b$ be an extremal ray in $\mathcal B$, and assume
    $$(b,0) = \frac 1 2 \left[ (b',c') + (b'',c'')\right].$$
    Since $\mathcal C$ is pointed, $c'=c''= 0$. Since $b$ is on an extremal ray in $\mathcal B$, it must be collinear to $b',b''$, proving the inclusion and finishing the proof.
\end{proof}

\bigskip

Let us now discuss the cones of positive semidefinite local diagonal invariant matrices, and their convex geometry. First, define the three convex cones of interest as sections of the positive semidefinite cone by the corresponding hyperplanes:
\begin{align*}
    \LDUI_d^+ &:= \PSD_{d^2} \cap \LDUI_d,\\
    \CLDUI_d^+ &:= \PSD_{d^2} \cap \CLDUI_d,\\
    \LDOI_d^+ &:= \PSD_{d^2} \cap \LDOI_d.
\end{align*}

The elements of these cones were characterized in Lemma \ref{lemma:LDOI-psd-ppt} (see also \cite[Lemma 7.6]{nechita2021graphical} or \cite[Theorem 5.2]{johnston2019pairwise} for the CLDUI case); this also follows from the following more general result, which gives the spectrum of an arbitrary local diagonal invariant matrix. 
\begin{proposition}\label{prop:spectra}
For any triple of matrices $(A,B,C) \in \mathcal M_{d}(\mathbb C)^{\times 3}_{\mathbb C^d}$, the spectra of the matrices $X^{(1,2,3)}$ are given by
\begin{align*}
    \operatorname{spec} X^{(1)}_{(A,C)} &= \{A_{ii}\}_{i \in d} \cup \bigcup_{i<j} \operatorname{spec} \begin{bmatrix} A_{ij} & C_{ij}, \\ C_{ji} & A_{ji} \end{bmatrix}\\
    \operatorname{spec} X^{(2)}_{(A,B)} &= \operatorname{spec} B \cup \{A_{ij}\}_{i \neq j},\\
    \operatorname{spec} X^{(3)}_{(A,B,C)} &= \operatorname{spec} B \cup \bigcup_{i<j} \operatorname{spec} \begin{bmatrix} A_{ij} & C_{ij} \\ C_{ji} & A_{ji} \end{bmatrix}.
\end{align*}
\end{proposition}
\begin{proof}
    The formulas follow immediately from the block structure of the matrices $X^{(1,2,3)}$, see Proposition \ref{prop:block-structure}. 
\end{proof}

\begin{theorem}
The extremal rays of the cones $\LDUI_d^+, \CLDUI_d^+,\LDOI_d^+ \subseteq \mathcal M_{d^2}^{sa}(\mathbb C)$ are as follows (we abuse notation below, by giving a vector representative $v$ for each extremal ray $\mathbb R_+v$): 
\begin{align*}
    \operatorname{ext}\LDUI_d^+ &= \{\ketbra{i}{i} \otimes \ketbra{i}{i}\}_{i\in [d]} \sqcup \bigsqcup_{1 \leq i < j \leq d} \{\ketbra{x_{ij}}{x_{ij}}\}_{\ket{x_{ij}} \in \mathbb C\ket{ij} \oplus \mathbb C \ket{ji}, \, \ket{x_{ij}} \neq 0},\\
    \operatorname{ext}\CLDUI_d^+ &= \{\ketbra{i}{i} \otimes \ketbra{j}{j}\}_{i \neq j \in [d]} \sqcup \{\ketbra{\operatorname{diag}(x)}{\operatorname{diag}(x)}\}_{\ket{x} \in \mathbb C^d, \, \ket{x} \neq 0},\\
    \operatorname{ext}\LDOI_d^+ &=\{\ketbra{\operatorname{diag}(x)}{\operatorname{diag}(x)}\}_{\ket{x} \in \mathbb C^d} \sqcup \bigsqcup_{1 \leq i < j \leq d} \{\ketbra{x_{ij}}{x_{ij}}\}_{\ket{x_{ij}} \in \mathbb C\ket{ij} \oplus \mathbb C \ket{ji}, \ket{x}_{ij} \neq 0}.
\end{align*}
where, for $\ket{x}\in \C{d}$, $\ket{\operatorname{diag}(x)}=\sum_{i=1}^d x_i \ket{ii} \in \C{d}\otimes \C{d}$.
\end{theorem}
\begin{proof}
The idea of the proof is to decompose the cones as Cartesian products of simpler cones, following the decompositions from Propositions \ref{prop:block-structure} and \ref{prop:spectra}. Let us start with the case of $\LDUI_d^+$. In terms of the matrices $(A,C)$, we have: 
$$\MLDUI{d} \supset \{ (A,C) \, : \, X^{(1)}_{(A,C)} \in \LDUI^+_d\} = \mathbb R^d_+ \times \bigtimes_{1\leq i<j \leq d}\PSD_2^{(i,j)},$$
where 
$$\PSD_2^{(i,j)} := \left\{A_{ij}, A_{ji} \in \mathbb R, C_{ij} \in \mathbb C, C_{ji} = \bar C_{ij} \, : \,  \begin{bmatrix} A_{ij} & C_{ij} \\ C_{ji} & A_{ji}\end{bmatrix} \in \PSD_2 \right\}.$$
We can thus identify an element of the cone $\LDUI^+_d$ with a vector of $1 + \binom{d}{2}$ coordinates, following the decomposition above. Using this block structure, the extremal rays of $\LDUI^+_d$ are the union of the sets of extremal rays of the different factors: 
$$\{(\ketbra{i}{i},0, \ldots, 0)\}_{i\in [d]} \sqcup \bigsqcup_{1 \leq i < j \leq d} \{(0,0,\ldots, 0,\ketbra{x_{ij}}{x_{ij}}, 0, \ldots, 0)\}_{\ket{x_{ij}} \in \mathbb C^2}.$$
Let us now write these elements in the usual form, using the $X^{(1)}$ isomorphism:
\begin{align*}
    (\ketbra{i}{i},0, \ldots, 0) \quad&\leadsto\quad A = C = \ketbra{i}{i} \quad\leadsto\quad X^{(1)}_{(A,C)} = \ketbra{i}{i} \otimes \ketbra{i}{i}\\
    (0,0,\ldots, 0,\ketbra{x_{ij}}{x_{ij}}, 0, \ldots, 0) \quad&\leadsto\quad  \begin{bmatrix} A_{ij} & C_{ij} \\ C_{ji} & A_{ji}\end{bmatrix} = \ketbra{x_{ij}}{x_{ij}}  \quad\leadsto\quad X^{(1)}_{(A,C)} = \ketbra{x_{ij}}{x_{ij}} \text{ on } \mathbb C\ket{ij} \oplus \mathbb C \ket{ji}.
\end{align*}

For $\CLDUI_d^+$, we start from the decomposition given in Proposition \ref{prop:block-structure} 
\begin{equation}\label{eq:decomposition-CLDUI+}
\MLDUI{d} \supset \{ (A,B) \, : \, X^{(2)}_{(A,B)} \in \CLDUI^+_d\} = \EWP^0_d \times \PSD_d,
\end{equation}
where $\EWP^0_d$ is the cone of entrywise positive matrices with zero diagonal; we have
$$\EWP^0_d = \bigtimes_{i,j \in [d], \, i \neq j} \mathbb R_+.$$
An extremal ray coming from the $i \neq j$ factor $\mathbb R_+$ above corresponds to 
$$A = \ketbra{i}{j} \text{ and } B=0 \quad\leadsto\quad X^{(2)}_{(A,B)} = \ketbra{i}{i} \otimes \ketbra{j}{j}.$$
An extremal ray $\ketbra{x}{x}$ for $\ket{x} \neq 0$ corresponding to the factor $\PSD_d$ in \eqref{eq:decomposition-CLDUI+} gives
$$B = \ketbra{x}{x} \text{ and } A=\operatorname{diag} B \quad\leadsto\quad X^{(2)}_{(A,B)} = \ketbra{\operatorname{diag}(x)}{\operatorname{diag}(x)},$$
where, for a vector $\ket{x} \in \mathbb C^d$, 
$$\ket{\operatorname{diag}(x)} = \sum_{i=1}^d x_i \ket{ii}.$$

Finally, for $\LDOI_d^+$, we have the decomposition
$$\MLDOI{d} \supset \{ (A,B,C) \, : \, X^{(3)}_{(A,B,C)} \in \LDOI^+_d\}= \PSD_d \times \bigtimes_{1\leq i<j \leq d}\PSD_2^{(i,j)},$$
We leave the details of this case to the reader.
\end{proof}
\begin{remark}
The extremal rays of the cones $\LDUI_d^+, \CLDUI_d^+,\LDOI_d^+$ consist only of unit rank matrices, see Corollary \ref{cor:rank-X}. These are precisely the (positive multiples of) rank one projections in the respective vector spaces $\LDUI_d, \CLDUI_d,\LDOI_d$. In general, linear sections of cones can have additional extremal rays. 
\end{remark}

\bigskip 

We now shift the discussion to the case of PCP and TCP matrices, see Definitions \ref{def:pcp} and \ref{def:tcp}. Several elementary properties of these matrices can be found in \cite[Sections 3,4]{johnston2019pairwise} and \cite[Appendix B]{nechita2021graphical} respectively. In particular, it has been shown that PCP (resp.~TCP) matrices form convex cones:
\begin{align*}
    \PCP_d &:= \{ (A,B) \in \MLDUI{d} \, : \, (A,B) \text{ is pairwise completely positive}\}, \\ 
    & \qquad\qquad \subseteq \{(A,B)\in \MLDUI{d} : A\in \Mreal{d} \text{ and } B\in \Msa{d} \} \\
    \TCP_d &:= \{ (A,B,C) \in \MLDOI{d} \, : \, (A,B,C) \text{ is triplewise completely positive}\}. \\
    & \qquad\qquad \subseteq \{(A,B)\in \MLDOI{d} : A\in \Mreal{d} \text{ and } B,C\in \Msa{d} \}
\end{align*}
We now show that these cones are also closed.
\begin{proposition}
The cones $\PCP_d, \TCP_d$ are closed. 
\end{proposition}
\begin{proof}
    The proof is a simple consequence of the fact that the cone of separable bipartite matrices $\SEP_d := \operatorname{conv}\{ X \otimes Y \, : X,Y \in \PSD_d\} \subset \Msa{d^2}$ is closed \cite[Propositions 6.3 and 6.8]{watrous2018theory}. This is because from Theorem \ref{theorem:LDOI-sep}, we have:
    \begin{align*}
        \PCP_d &= \left[X^{(1)}\right]^{-1}(\SEP_d \cap \LDUI_d)\\
        &= \left[X^{(2)}\right]^{-1}(\SEP_d \cap \CLDUI_d),\\
        \TCP_d &= \left[X^{(3)}\right]^{-1}(\SEP_d \cap \LDOI_d),
    \end{align*}
where the continuous maps $X^{(\cdot)}$ have been defined in Section \ref{sec:LDUI-CLDUI-LDOI}. 
\end{proof}

We now examine the stability of the $\PCP_d$ and $\TCP_d$ cones under the action of taking direct sums and principal subtriples (refer to Section~\ref{sec:linear-structure} for the relevant definitions).
\begin{proposition}\label{prop:stable_principal}
For a triple $(A,B,C)\in \MLDOI{d}$, the following implication holds for all index sets $I\subset [d]$:
$(A,B,C)\in \TCP_d \implies (A[I],B[I],C[I])\in \TCP_{|I|}$.
\end{proposition}
\begin{proof}
    Let $V,W\in \M{d,d'}$ form a TCP decomposition of $(A,B,C)$. Construct the $|I|\times d'$ matrices $V'$ and $W'$ by selecting the rows indexed by $I\subset [d]$ from $V$ and $W$, respectively. Then, it is an easy exercise to verify that $V',W'\in \M{|I|,d'}$ form a TCP decomposition of $(A[I],B[I],C[I])$.
\end{proof}

\begin{proposition}\label{prop:stable_directsum}
For triples $(A_1,B_1,C_1)\in \MLDOI{d_1}$ and $(A_2,B_2,C_2)\in \MLDOI{d_2}$, we have $(A_1\oplus A_2,B_1\oplus B_2,C_1\oplus C_2) \in \TCP_{d_1 + d_2} \iff  (A_1,B_1,C_1)\in \TCP_{d_1}$ and $(A_2,B_2,C_2)\in \TCP_{d_2}$.
\end{proposition}
\begin{proof}
    The forward implication follows from Proposition~\ref{prop:stable_principal}. For the reverse implication, we observe that if $V_1,W_1\in \M{d_1,d'}$ and $V_2,W_2\in \M{d_2,d''}$ form TCP decompositions of $(A_1,B_1,C_1)$ and $(A_2,B_2,C_2)$, respectively (see Definition~\ref{def:tcp}), then the direct sums $V_1\oplus V_2, W_1\oplus W_2 \in \M{d_1+d_2, d'+d''}$ form a TCP decomposition of $(A_1\oplus A_2,B_1\oplus B_2,C_1\oplus C_2)$.
\end{proof}

Let us now provide two sufficient conditions for membership inside the TCP cone, which come from sufficient separability criteria proven in \cite{gurvits2002largest} and \cite{nechita2018separability} respectively. 

\begin{lemma}
Consider a triple $(A,B,C)\in \MLDOI{d}$ such that $X^{(3)}_{(A,B,C)} \in \LDOI^+_d$ and, moreover, 
$$\sum_{i,j} A_{ij}^2 + \sum_{i \neq j} |B_{ij}|^2 + \sum_{i \neq j} |C_{ij}|^2 \leq \frac{1}{d^2-1} \left( \sum_{i,j} A_{ij} \right)^2,$$
where $i,j \in [d]$. Then, $(A,B,C) \in \TCP_d$, i.e.~$X^{(3)}_{(A,B,C)}$ is separable.
\end{lemma}
\begin{proof}
First, note that $X^{(3)}_{(A,B,C)} \in \LDOI^+_d$ implies that the matrix $A$ is real. The result follows from the separability criterion in \cite[Corollary 2]{gurvits2002largest}, which states that any positive semidefinite matrix $X \in \M{d} \otimes \M{d}$ such that 
    $$\Tr(X^2) \leq \frac{(\Tr X)^2}{d^2-1}$$
    is separable. Use $ \Tr X^{(3)}_{(A,B,C)} = \sum_{i,j} A_{ij}$ and $\Tr \left( X^{(3)}_{(A,B,C)} \right)^2 = \|A\|_{\operatorname{Fro}}^2 + \|\widetilde B\|_{\operatorname{Fro}}^2 + \|\widetilde C\|_{\operatorname{Fro}}^2$, where $||.||_{\operatorname{Fro}}$ denotes the \emph{Frobenius} norm on $\M{d}$.
\end{proof}
\begin{remark}
In the case where $X^{(3)}_{(A,B,C)}$ is a quantum state (i.e.~$\sum_{i,
j} A_{ij} =1$), the condition in the statement above is that $X^{(3)}_{(A,B,C)}$ belongs to the largest separable euclidean ball centered in the maximally mixed state $(\mathbb{I}_d \otimes \mathbb{I}_d) / d^2$, which has radius $1/\sqrt{d^2(d^2-1)}$.
\end{remark}

Recall from Proposition \ref{prop:conditional-expectations} that, given a matrix $A \in \M{d}$, we define the diagonal matrices 
$$A^{\mathsf{row}}_{ii} := \sum_{j=1}^d A_{ij} \qquad \text{ and } \qquad A^{\mathsf{col}}_{jj} := \sum_{i=1}^d A_{ij}.$$

\begin{lemma}
Consider a triple $(A,B,C)\in \MLDOI{d}$ such that $X^{(3)}_{(A,B,C)} \in \LDOI^+_d$ and
$$X^{(3)}_{(d+1)(A,B,C) - (A^\square,A^\square,A^\square)} \in \LDOI^+_d,$$
for either $\square = \mathsf{row}$ or $\square = \mathsf{col}$. Then, $(A,B,C) \in \TCP_d$, i.e.~$X^{(3)}_{(A,B,C)}$ is separable.
\end{lemma}
\begin{proof}
    The result follows from the separability criterion in \cite[Proposition 11]{nechita2018separability}, which states that any positive semidefinite matrix $X \in \M{d} \otimes \M{d}$ such that either
    $$ (d+1)X \geq \mathbb I_d \otimes [\Tr \otimes \operatorname{id}](X) \qquad \text{or} \qquad (d+1)X \geq [\operatorname{id} \otimes \Tr](X) \otimes \mathbb I_d $$
    is separable. Use Proposition \ref{prop:conditional-expectations} to write the partial traces of LDOI matrices in terms of $A^{\mathsf{row}, \mathsf{col}  }$. 
\end{proof}
\begin{remark}
In order for the condition in the lemma above to hold, the matrix $A$ has to be row-(resp.~column-)balanced. For example, in the $\square = \mathsf{row}$ case, one needs that every element must not be much smaller that the average of its row:
$$\forall i,j\in[d]: \qquad A_{ij} \geq \frac{d}{d+1} \left( \frac 1 d \sum_{k=1}^d A_{ik} \right).$$
\end{remark}

\bigskip

We now proceed towards characterizing the extreme rays of the $\PCP_d$ and $\TCP_d$ cones, which will precisely translate into a characterization of the extreme rays of the cone of separable LDOI matrices through the $X^{(i)}$ isomorphisms (for $i=1,2,3$) from page 5. It is pertinent to point out here that although the extreme rays of the full separable cone $\SEP_d \subset \Msa{d^2}$ are already well-known (these are the rank one projectors onto product vectors), the restriction of the domain to the intersection $\LDOI_d \cap \SEP_d \subset \LDOI_d^{sa}$ might create new extremal rays which are not extremal in the full $\SEP_d$ cone (see also Remark~\ref{remark:extremal-separable} in this regard). Theorem~\ref{theorem:extremal-PCPTCP} below resolves these technicalities by providing a complete characterization of the extremal rays in $\PCP_d$ and $\TCP_d$.

\begin{theorem} \label{theorem:extremal-PCPTCP}
The extremal rays of the $\PCP_d$ cone are obtained from Eq~\eqref{eq:def-PCP} with matrices $V,W\in \M{d,1}$. In other words, the extremal rays are of the form ( $\ket{v},\ket{w} \in \mathbb C^d \setminus \{0\}$)
\begin{equation*}\label{eq:extremal-PCP}
    A = \ketbra{v \odot \bar v}{w \odot \bar w}, \qquad B = \ketbra{v \odot w}{v \odot w}.
\end{equation*}
A similar result holds for the $\TCP_d$ cone: the extremal rays are given by ( $\ket{v},\ket{w} \in \mathbb C^d \setminus \{0\}$)
$$A = \ketbra{v \odot \bar v}{w \odot \bar w}, \qquad B = \ketbra{v \odot w}{v \odot w}, \qquad C = \ketbra{v \odot \bar w}{v \odot \bar w}.$$
\end{theorem}
\begin{proof}
    We prove the case of PCP extremal rays, leaving the similar discussion of TCP matrices to the reader. Let us first show that extremal rays $\mathbb R_+(A,B) \in \operatorname{ext} \PCP_d$ are of the form \eqref{eq:extremal-PCP}. Assume $A,B \neq 0$ are written as in \eqref{eq:def-PCP} with $V,W \in \mathcal M_{d,k}(\mathbb C)$. Let $\ket{v^{(1)}}, \ldots, \ket{v^{(k)}} \in \mathbb C^d$ (resp.~$\ket{w^{(1)}}, \ldots, \ket{w^{(k)}}$) be the columns of $V$ (resp.~$W$). A simple computation shows that 
    $$A = \sum_{t=1}^k A^{(t)} \qquad \text{ and } \qquad B = \sum_{t=1}^k B^{(t)},$$
    where $A^{(t)},B^{(t)}$ are constructed as in \eqref{eq:extremal-PCP} from $\ket{v^{(t)}},\ket{w^{(t)}}$. Since $(A,B)$ is extremal, we must have $\mathbb R_+(A,B) = \mathbb R_+(A^{(t)},B^{(t)})$ for all $t \in [k]$, proving the claim.
    
    Consider now arbitrary $\ket{v}, \ket{w}\in \mathbb C^d \setminus \{0\}$, the corresponding $A,B$ from \eqref{eq:extremal-PCP}, and let us show that $\mathbb R_+(A,B)$ is extremal. Assume that $(A,B) = \sum_{t=1}^k (A^{(t)},B^{(t)})$ for some pairs $(A^{(t)},B^{(t)})$ as in \eqref{eq:extremal-PCP}, defined via vectors $\ket{v^{(t)}},\ket{w^{(t)}} \in \mathbb C^d$, $t \in [k]$. Since
    $$B = \ketbra{v \odot w}{v \odot w} = \sum_{t=1}^k B^{(t)}$$
    and the matrix on the left hand side has unit rank (hence it is on an extremal ray in the cone of positive semidefinite matrices in $\M{d}$), there must exist scalars $\beta_t \in \mathbb C$ such that 
    $$\ket{v^{(t)} \odot w^{(t)}} = \beta_t \ket{v \odot w},$$
    which in turn is equivalent to the relation $B^{(t)} = |\beta_t|^2 B$ (hence $\sum_{t=1}^k |\beta_t|^2 = 1$). In order to conclude, we show next that the same colinearity relation holds for the $A$ matrices: $A^{(t)} = |\beta_t|^2 A$. First, note that the PCP pairs $(A^{(t)}, B^{(t)})$ satisfy the PPT condition $A^{(t)}_{ij}A^{(t)}_{ji} \geq |B^{(t)}_{ij}|^2$
    with equality, since both the LHS and the RHS above are equal to $|v^{(t)}_i|^2|v^{(t)}_j|^2|w^{(t)}_i|^2|w^{(t)}_j|^2$. Next, we compute 
    \begin{align*}
        A_{ij}A_{ji} &= |v_i|^2|v_j|^2|w_i|^2|w_j|^2\\
        &= \sum_{s,t=1}^k A^{(s)}_{ij} A^{(t)}_{ji} = \sum_{s,t=1}^k |v^{(s)}_i|^2|w^{(s)}_j|^2 \cdot |v^{(t)}_j|^2|w^{(t)}_i|^2\\
        &= \sum_{s=1}^k |v^{(s)}_i|^2|w^{(s)}_j|^2  |v^{(s)}_j|^2|w^{(s)}_i|^2 + \sum_{ s < t } \left( |v^{(s)}_i|^2|w^{(s)}_j|^2 |v^{(t)}_j|^2|w^{(t)}_i|^2  + |v^{(t)}_i|^2|w^{(t)}_j|^2 |v^{(s)}_j|^2|w^{(s)}_i|^2 \right).\\
    \end{align*}
   A simple application of the arithmetic-geometric mean inequality then yields:
    \begin{align*}
        A_{ij}A_{ji} &\geq |v_i|^2|v_j|^2|w_i|^2|w_j|^2 \left\{ \sum_{s=1}^k |\beta_s|^4 + \sum_{s < t} 2\sqrt{ |v^{(s)}_i|^2|w^{(s)}_j|^2 |v^{(t)}_j|^2|w^{(t)}_i|^2  \cdot |v^{(t)}_i|^2|w^{(t)}_j|^2 |v^{(s)}_j|^2|w^{(s)}_i|^2} \right\} \\
        &= |v_i|^2|v_j|^2|w_i|^2|w_j|^2 \left\{\sum_{s=1}^k |\beta_s|^4 + \sum_{1 \leq s < t \leq k} 2 |\beta_s|^2|\beta_t|^2 \right\} \\
        &= |v_i|^2|v_j|^2|w_i|^2|w_j|^2 \left\{ \sum_{s=1}^k |\beta_s|^2 \right\}^2 = A_{ij}A_{ji}.
        \end{align*}
Since the AM-GM inequality above is saturated, we get
$$\forall s,t \in [k], \, \forall i,j \in [d]: \qquad |v^{(s)}_i|^2|w^{(s)}_j|^2 |v^{(t)}_j|^2|w^{(t)}_i|^2 = |v^{(t)}_i|^2|w^{(t)}_j|^2 |v^{(s)}_j|^2|w^{(s)}_i|^2.$$
    In other words, we have $A^{(s)}_{ij}A^{(t)}_{ji} = A^{(t)}_{ij}A^{(s)}_{ji} \quad \forall s,t\in [k], \, \forall i,j\in [d]$. Multiplying this equation with $A^{(s)}_{ij}A^{(t)}_{ij}$ and using $A^{(t)}_{ij}A^{(t)}_{ji} = |B^{(t)}_{ij}|^2 = |\beta_t|^4 |B_{ij}|^2$, we obtain (note that we can take square roots since the entries of the matrices $A^{(\cdot)}$ are non-negative)
    $$\forall s,t \in [k], \, \forall i,j \in [d]: \qquad |\beta_t|^2 A^{(s)}_{ij} = |\beta_s|^2 A^{(t)}_{ij}.$$
    We now have
    $$\forall s\in [k]: \qquad |\beta_s|^2 A_{ij} = \sum_{t=1}^k |\beta_s|^2 A^{(t)}_{ij} = \sum_{t=1}^k |\beta_t|^2 A^{(s)}_{ij} = A^{(s)}_{ij},$$
    and the proof is complete. 
\end{proof}

\begin{remark}
The extremal rays of the cone of completely positive matrices (i.e.~PCP pairs of the form $(A,A)$) are well-known \cite[Section 2.2]{berman2003completely}. They are of the form 
$$A = \ketbra{v \odot \bar v}{v \odot \bar v},$$
for some vector $\ket{v} \in \mathbb C^d$. The situation is thus similar to the one described above. 
\end{remark}

\begin{remark}
Similar computations show that the extremal rays of the cone of TCP matrices of the form $(A,B,B)$ are those of the form 
\begin{equation}\label{eq:extremal-TCP-ABB}
    A = \ketbra{v \odot \bar v}{w \odot \bar w}, \qquad B = \ketbra{v \odot w}{v \odot w},
\end{equation}
for vectors $\ket{v},\ket{w}\in \C{d}$.
\end{remark}

\begin{remark} \label{remark:extremal-separable}
It is well known that the convex cone of separable matrices in $\Msa{d^2}$ has extreme rays spanned by tensor products of rank-one projections: 
$$ \operatorname{ext} \SEP_d = \{\mathbb R_+(|v\rangle\langle v| \otimes |w\rangle\langle w|)\}_{0\neq \ket{v},\ket{w} \in \mathbb{C}^d}.$$
Let $(A,B,C)\in \MLDOI{d}$ be obtained from the vectors $\ket{v},\ket{w} \neq 0$ as in Proposition~\ref{theorem:extremal-PCPTCP}, i.e., $(A,B,C)$ is on an extremal ray in $\TCP_d$. Then, from Theorem~\ref{theorem:extremal-PCPTCP}, it is easy to deduce that
$$   \operatorname{Proj}_{\LDOI}(|v\rangle\langle v| \otimes |w\rangle\langle w|) = X^{(3)}_{(A,B,C)}.$$
In other words, the local averaging operation with respect to the random diagonal orthogonal matrices establishes a one-one correspondence between extreme rays of the cone of separable matrices in $\Msa{d^2}$ and $\LDOI_d^{sa}$.
\end{remark}

\begin{proposition}
Let $0 \neq \ket{v},\ket{w} \in \mathbb C^d$ be two non-zero vectors. The ranks of the extremal invariant separable matrices from Theorem \ref{theorem:extremal-PCPTCP} are as follows. Writing
$$A = \ketbra{v \odot \bar v}{w \odot \bar w}, \qquad B =  \ketbra{v \odot w}{v \odot w}, \qquad C = \ketbra{v \odot \bar w}{v \odot \bar w},$$
we have
\begin{align*}
    \operatorname{rank} X^{(1)}_{(A,C)} &=  \sigma(v)\sigma(w) - \left(\frac{\sigma(v \odot w)^2 - \sigma(v \odot w)}{2}\right), \\   
    \operatorname{rank} X^{(2)}_{(A,B)} &= \sigma(v)\sigma(w) - (\sigma(v \odot w) - \mathbbm{1}_{\sigma(v \odot w) > 0}), \\   
    \operatorname{rank} X^{(3)}_{(A,B,C)} &= \sigma(v)\sigma(w) - \left(\frac{\sigma(v \odot w)^2 + \sigma(v \odot w)}{2} - \mathbbm{1}_{\sigma(v \odot w) > 0}\right) ,
\end{align*}
where $\sigma(z)$ is the size of the support of the vector $\ket z \in \mathbb C^d$: 
$$\sigma(z) := |\{i \in [d] \, : \, z_i \neq 0\}|.$$
In particular, the ranks of extremal separable invariant states can be as high as 
\begin{align*}
    \operatorname{rank} X^{(1)}_{(A,C)} &\leq  \frac{d(d+1)}{2}, \\
    \operatorname{rank} X^{(2)}_{(A,B)} &\leq d^2-d+1, \\
    \operatorname{rank} X^{(3)}_{(A,B,C)} &\leq \frac{d(d-1)}{2} + 1,
\end{align*}
with the extremal values being attained for fully supported vectors $\ket v, \ket w$.
\end{proposition}
\begin{proof}
The idea of the proof is to express the formulas from Corollary \ref{cor:rank-X} in terms of the size of the supports of the vectors $v,w,v \odot w$. Let us start with the case of $X^{(1)}_{(A,C)}$. On one hand, we have 
$$|\{i \in [d] \, : \, A_{ii} \neq 0\}| = |\{i \in [d] \, : \, v_iw_i \neq 0\}| = \sigma(v \odot w).$$
On the other hand, note that the matrix 
$$\begin{bmatrix} A_{ij} & C_{ij} \\ C_{ji} & A_{ji} \end{bmatrix} = \begin{bmatrix} |v_i|^2|w_j|^2 & v_i \overline{w_i} \overline{v_j}w_j \\ v_j \overline{w_j} \overline{v_i}w_i & |v_j|^2|w_i|^2 \end{bmatrix}$$
is singular, so its rank is either zero or one, depending on whether at least one of $v_iw_j$ or $v_jw_i$ is non-zero. We have thus
\begin{align*}
    \sum_{i<j} \operatorname{rank}\begin{bmatrix} A_{ij} & C_{ij} \\ C_{ji} & A_{ji} \end{bmatrix} &= \sum_{i<j} \mathbbm{1}_{v_iw_j \neq 0 \text{ or } v_jw_i \neq 0}\\
    &= \frac 1 2 \sum_{i \neq j} \left(\mathbbm{1}_{v_iw_j \neq 0} + \mathbbm{1}_{v_jw_i \neq 0} - \mathbbm{1}_{v_iw_j \neq 0 \text{ and } v_jw_i \neq 0}\right)\\
    &= \sum_{i \neq j} \mathbbm{1}_{v_iw_j \neq 0} - \frac 1 2 \sum_{i \neq j} \mathbbm{1}_{v_iw_jv_jw_i \neq 0}\\
    &= \sum_{i,j} \mathbbm{1}_{v_i \neq 0} \mathbbm{1}_{w_j \neq 0}
    - \sum_i \mathbbm{1}_{v_iw_i \neq 0} - \frac 1 2 \left(\sum_{i , j} \mathbbm{1}_{v_iw_i \neq 0}\mathbbm{1}_{v_jw_j \neq 0} - \sum_i \mathbbm{1}_{v_iw_i \neq 0} \right)\\
    &= \sigma(v)\sigma(w) - \sigma(v \odot w) - \frac 1 2 \left(\sigma(v \odot w)^2 - \sigma(v \odot w)\right).
\end{align*}
Putting everything together, we obtain the result from the statement. 

For the case of $ X^{(2)}_{(A,B)}$, note that we have, using again Corollary \ref{cor:rank-X},
\begin{align*}
    \operatorname{rank} X^{(2)}_{(A,B)} &= \operatorname{rank} \ketbra{v \odot w}{v \odot w} + |\{i \neq j \, : \, |v_i|^2 |w_j|^2 \neq 0\}|\\
    &= \mathbbm{1}_{\sigma(v \odot w) > 0} +[ \sigma(v)\sigma(w) - \sigma(v \odot w)],
\end{align*}
proving the claim.

Finally, the case of $X^{(3)}$ is left to the reader, as it can be easily deduced from the first two. Regarding the maximal values of the ranks, the claims can be proven by a careful analysis of the constrained optimization problems. We give below the proof in the CLDUI case, leaving the two others to the reader. Let us write $\sigma(v \odot w) = x$, $\sigma(v) = x+a$, $\sigma(w) = x+b$ for non-negative integers $x,a,b$. Given a triple $x,a,b$, there exist $\ket v, \ket w$ as above if and only if $x+a+b \leq d$, which is the only constraint we need to consider. With these new variables, we have
$$\sigma(v)\sigma(w) - \sigma(v \odot w) = x^2 + x(a+b-1) + ab,$$
which is a non-decreasing function in the integer $x \in [0,d]$. Hence, its maximum is attained at $x = d-(a+b)$. Using this value of $x$, the claimed upper bound (containing now the indicator function) reads $ab - (a+b)(d-1) \leq 0$, which can easily be checked to be true, equality being attained at $a=b=0$ corresponding to fully supported vectors $\ket v, \ket w$.
\end{proof}

To illustrate the results above, let us consider the extremal matrices $X^{(1,2,3)}$ corresponding to the choice $\ket v = \ket w = \ket{\operatorname{diag}\mathbb{I}_d}$ (the all-ones vector). We have in this case $A = B = C = \mathbb J_d$. So,
\begin{align*}
    X^{(1)}_{(A,C)} &= \mathbb I + F - P_{eq}= (\mathbb I - P_{eq}) + P_s - P_a = 2P_s - P_{eq},\\
    X^{(2)}_{(A,B)} &= \mathbb I + d P_\omega - P_{eq} = (\mathbb I - P_{eq}) + dP_\omega,\\
    X^{(3)}_{(A,B,C)} &= \mathbb I + d P_\omega + F - 2P_{eq} = 2(P_s - P_{eq}) + dP_\omega,
\end{align*}
where $F = P_s - P_a$ is the flip operator, $P_{s,a}$ are, respectively, the projections on the symmetric and the anti-symmetric subspace of $\mathbb C^d \otimes \mathbb C^d$, $P_\omega$ is the projection on the maximally entangled state, and $P_{eq}$ is the rank-$d$ projection
$$P_{eq} = \sum_{i=1}^d \ketbra{ii}{ii}.$$
Note that we have $P_\omega \leq P_{eq} \leq P_s \leq \mathbb I$ in the lattice of projections. 

\bigskip

Having described the convex structure of positive semidefinite and separable LDUI/CLDUI/LDOI matrices, we now move on to other convex cones relevant to quantum information theory. In the absence of invariance, there are five proper closed convex cones of $\Msa{d^2}$ which play crucial roles: 
\begin{align*}
    \SEP_d &= \{X \in \Msa{d^2} \, : \, X = \sum_{i=1}^r A_i \otimes B_i, \quad A_i,B_i \in \PSD_d \text{ and } r \in \mathbb{N} \},\\
    \mathsf{PPT}_d &= \{X \in \Msa{d^2} \, : \, X, X^\Gamma \in \PSD_{d^2}\},\\
    \PSD_{d^2} &= \{X \in \Msa{d^2} \, : \, \langle z| X |z \rangle \geq 0 \quad \forall z\in \mathbb C^{d^2}\},\\
    \mathsf{DEC}_d &= \{X \in \Msa{d^2} \, : \, X = A + B^\Gamma, \quad A,B \in \PSD_{d^2}\},\\
    \mathsf{BP}_d &= \{X \in \Msa{d^2} \, : \, \langle x \otimes y| X |x \otimes y \rangle \geq 0 \quad \forall x,y\in \mathbb C^d\}.    
\end{align*}
\begin{figure}[htbp]
    \centering
    \includegraphics{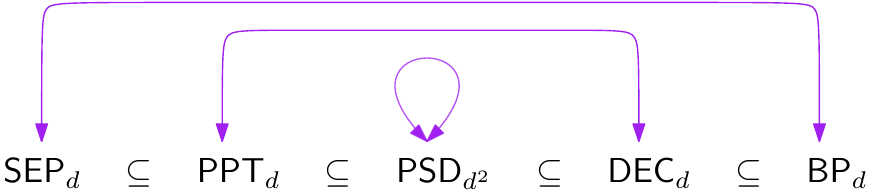}
    \caption{ Inclusion and duality (represented by arrows) relations for cones of bipartite matrices.}
    \label{fig:cones-PSD}
\end{figure}

These five cones, called respectively separable, positive partial transpose, positive semidefinite, decomposable and block-positive, satisfy the inclusion and duality relations from Figure \ref{fig:cones-PSD}. Observe that duality here is understood in the sense of convex cones. For example,
$$\SEP^*_d \coloneqq \{Y\in \Msa{d^2} \, : \, \langle X, Y \rangle = \operatorname{Tr}(X^* Y) \geq 0 \quad \forall X\in \SEP_d \} = \mathsf{BP}_d.$$
In other words, $X\in \Msa{d^2}$ is separable if and only if $\operatorname{Tr}(XY)\geq 0, \, \forall \, Y\in \mathsf{BP}_d$. Equivalently, $X$ is entangled if and only if there exits a block-positive $Y\in \mathsf{BP}_d$ such that $\operatorname{Tr}(XY)<0$. For $X\in \LDOI_d$, since the following equation holds for all $Y\in \M{d}\otimes \M{d}$: 
\begin{equation}\label{eq:proj-duality}
\operatorname{Tr}(XY) = \operatorname{Tr}[\operatorname{Proj}_{\LDOI}(X)\cdot Y] = \operatorname{Tr}[X \cdot \operatorname{Proj}_{\LDOI}(Y)],
\end{equation}
it is evident that $X$ is separable if and only if $\operatorname{Tr}(XY)\geq 0 \,\, \forall \, Y\in \mathsf{BP}_d \cap \LDOI_d$. 

We are thus led to the following definitions: 
\begin{align*}
    \LDOI_d^{\mathsf{SEP}} &= \LDOI_d \cap \SEP_d = \{ X^{(3)}_{(A,B,C)} \, : \, (A,B,C) \in \TCP_d \},\\
    \LDOI_d^{\mathsf{PPT}} &= \LDOI_d \cap \mathsf{PPT}_d\\
    &= \{ X^{(3)}_{(A,B,C)} \, : \, A \in \EWP_d, \, B,C \in \PSD_d, \, A_{ij}A_{ji} \geq \max(|B_{ij}|^2,|C_{ij}|^2) \, \forall i,j \in [d]\},\\
    \LDOI_d^+ &= \LDOI_d \cap \PSD_{d^2}\\
    &= \{ X^{(3)}_{(A,B,C)} \, : \, A \in \EWP_d, \, B \in \PSD_d, \, A_{ij}A_{ji} \geq |C_{ij}|^2 \, \forall i,j \in [d]\}= \left( \LDOI_d^+ \right)^*,
    \\
    \LDOI_d^{\mathsf{DEC}} &= \LDOI_d \cap \mathsf{DEC}_d = \left( \LDOI_d^{\mathsf{PPT}} \right)^*,\\
    \LDOI_d^{\mathsf{BP}} &= \LDOI_d \cap \mathsf{BP}_d = \left( \LDOI_d^{\mathsf{SEP}} \right)^*.
\end{align*}
\begin{figure}[H]
    \centering
    \includegraphics{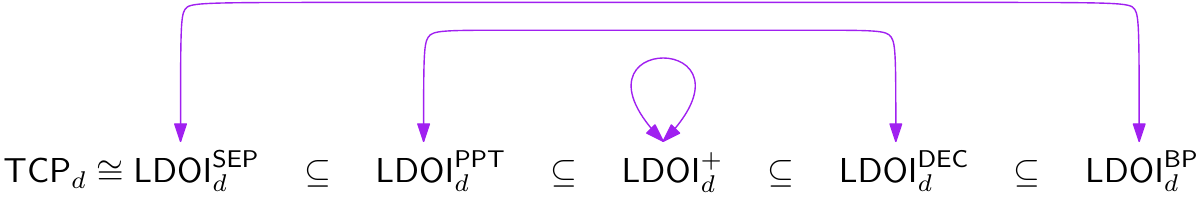}
    \caption{ Inclusion and duality (represented by arrows) relations for LDOI cones.}
    \label{fig:cones-LDOI}
\end{figure}
It is worthwhile to stress here that the above duality relations hold when the respective cones are understood as subsets of the corresponding vector space of self-adjoint LDOI matrices, that is $\LDOI_d^{sa}$. Similar inclusions and dualities hold for the LDUI / CLDUI cones, see Figure \ref{fig:cones-C-LDUI}. Note however that the positive semi-definite cones $\LDUI^+_d$ and $\CLDUI^+_d$ are \emph{not} isomorphic. 
\begin{figure}[htbp]
    \centering
    \includegraphics{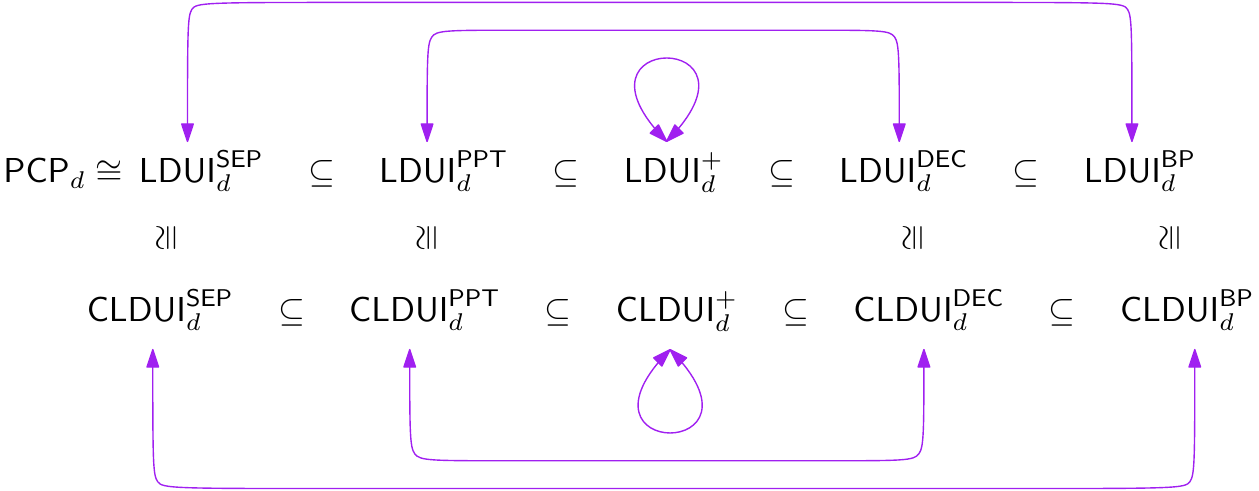}
    \caption{ Inclusion and duality (represented by arrows) relations for LDUI and CLDUI cones.}
    \label{fig:cones-C-LDUI}
\end{figure}

\section{Diagonal unitary and orthogonal covariant maps}\label{sec:DOC}
We denote the set of all linear maps $\Phi:\mathcal{M}_d(\mathbb{C})\rightarrow \mathcal{M}_d(\mathbb{C})$ by $\mathcal{T}_d(\mathbb{C})$. The Choi-Jamio{\l}kowski isomorphism identifies each map $\Phi \in \mathcal{T}_d(\mathbb{C})$ with a bipartite matrix $J(\Phi) \in \mathcal{M}_d(\mathbb{C})\otimes \mathcal{M}_d(\mathbb{C})$ (also called the Choi matrix of $\Phi$). In this section, we will use this isomorphism to study special families of covariant maps in $\mathcal{T}_d(\mathbb{C})$ by linking them with the families of local diagonal unitary/orthogonal invariant bipartite matrices from Section \ref{sec:LDUI-CLDUI-LDOI}. Before we begin, it is only fair to familiarize the readers with the basic theory of linear maps between matrix algebras. For a more detailed analysis, we refer the reader to \cite[Chapter 2]{watrous2018theory}, \cite[Chapters 2, 3]{bhatia2015positive}.

Consider a map $\Phi \in \mathcal{T}_d(\mathbb{C})$. $\Phi$ is called \emph{unital} if $\Phi (\mathbb{I}_d) = \mathbb{I}_d$, where $\mathbb{I}_d$ is the identity matrix in $\mathcal{M}_d(\mathbb{C})$. $\Phi$ is called \emph{trace preserving} if $\operatorname{Tr}(\Phi(Z))=\operatorname{Tr}(Z)$ for all $Z\in \mathcal{M}_d(\mathbb{C})$.  $\Phi$ is called \emph{hermiticity preserving} if it maps self-adjoint matrices to self-adjoint matrices in $\mathcal{M}_d(\mathbb{C})$. $\Phi$ is called \emph{positive} if $\Phi(Z)\in \PSD_d$ whenever $Z\in \PSD_d$. $\Phi$ is called \emph{completely positive} if the map $\operatorname{id} \otimes \, \Phi : \mathcal{M}_n(\mathbb{C}) \otimes \mathcal{M}_d(\mathbb{C}) \rightarrow \mathcal{M}_n(\mathbb{C}) \otimes \mathcal{M}_d(\mathbb{C})$ is positive for all $n\in \mathbb{N}$ (here $\operatorname{id}$ is the identity map in $\mathcal{T}_n(\mathbb{C})$). $\Phi$ is called \emph{completely copositive} if $\Phi \circ \top$ in $\mathcal{T}_d(\mathbb{C})$ is completely positive, where $\mathsf{T}$ is the transpose map in $\mathcal{T}_d(\mathbb{C})$. $\Phi$ is called \emph{decomposable} if it can be expressed as a sum of a completely positive and a completely copositive map and \emph{non-decomposable} otherwise. $\Phi$ is called PPT if it is both completely positive and completely copositive. $\Phi$ is called \emph{entanglement breaking} if $(\operatorname{id} \otimes \Phi) (X)$ is separable for all positive semi-definite $X\in \mathcal{M}_n(\mathbb{C}) \otimes \mathcal{M}_d(\mathbb{C})$. The \emph{dual} map $\Phi^*$ is defined as the unique adjoint of $\Phi$ with respect to the Hilbert-Schmidt inner product on $\M{d}$. With all the definitions in place, we now state the  Choi-Jamio{\l}kowski isomorphism in its full glory.

\begin{lemma}[Choi-Jamio{\l}kowski Isomorphism]  \cite{dePillis1967linear,Jamiokowski1972iso, Choi1975iso} \label{lemma:CJiso}
Define the linear bijection $J:\mathcal{T}_d(\mathbb{C}) \rightarrow \mathcal{M}_d(\mathbb{C}) \otimes \mathcal{M}_d(\mathbb{C})$ as $J(\Phi) = \sum_{i,j=1}^d \Phi(\vert i \rangle \langle j \vert) \otimes \vert i \rangle \langle j \vert$. Then, $\Phi\in \T{d}$ is
\begin{enumerate}
    \item hermiticity preserving if and only if $J(\Phi)$ is self-adjoint,
    \item positive if and only if $J(\Phi)$ is block positive, i.e., $\langle x \otimes y \vert J(\Phi) \vert x \otimes y \rangle \geq 0 \, \, \forall \ket{x}, \ket{y} \in \mathbb{C}^d$,
    \item completely positive if and only if $J(\Phi)$ is positive semi-definite,
    \item completely copositive if and only if $J(\Phi)^\Gamma$ is positive semi-definite,
    \item decomposable if and only if $\operatorname{Tr}(J(\Phi)X) \geq 0$ for all PPT matrices $X\in \mathcal{M}_d(\mathbb{C})\otimes \mathcal{M}_d(\mathbb{C})$,
    \item entanglement breaking if and only if $J(\Phi)$ is separable.
\end{enumerate}
\end{lemma}
In Lemma~\ref{lemma:CJiso}, part~$(5)$ appears in \cite{stormer1982}, while part~$(6)$ is due to \cite{horodecki2003entanglement}; see also \cite{girard2021convex} for a unified approach to the above presented results. 

The action of a map $\Phi$ and its adjoint $\Phi^*$ on $\mathcal{M}_d(\mathbb{C})$ can be retrieved from the Choi matrix $J(\Phi)$, as is depicted through the following equations:
\begin{align}
    \Phi(Z) = (\operatorname{id} \otimes \operatorname{Tr})[(\mathbb{I}_d \otimes Z^\top)J(\Phi)] &= \includegraphics[align=c]{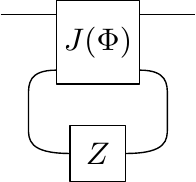} \label{eq:phi-action} \\[0.2cm]
    \Phi^*(Z) = (\operatorname{Tr} \otimes \operatorname{id})[J(\Phi)^\Gamma (Z\otimes \mathbb{I}_d)] &= \includegraphics[align=c]{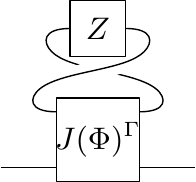}
\end{align}

This enables us to transform the unital and trace preserving property of $\Phi$ into partial trace conditions on its Choi matrix $J(\Phi)$, which forms the subject of the next Lemma. 

\begin{lemma}\label{lemma:unital-Tr-map}
For $\Phi \in \mathcal{T}_d(\mathbb{C})$, the following statements are equivalent:
\begin{enumerate}
    \item $\Phi$ is unital (resp.~trace preserving),
    \item $\Phi^*$ is trace preserving (resp.~unital),
    \item $(\operatorname{id} \otimes \operatorname{Tr})
[J(\Phi)]=\mathbb{I}_d$ (resp.~$\, (\operatorname{Tr} \otimes \operatorname{id})[J(\Phi)]=\mathbb{I}_d$).
\end{enumerate}
\end{lemma}

In Quantum Mechanics, physically allowed operations (called \emph{quantum channels}) on quantum states are completely positive and trace preserving linear maps between the spaces of bounded operators on separable Hilbert spaces: $\Lambda : \mathcal{B}(\mathcal{H}) \rightarrow \mathcal{B}(\mathcal{H'})$, see \cite{Holevo2012channels, holevo2019channels}. While positivity and trace preserving property is expected to ensure that $\Lambda$ maps quantum states in $\mathcal{B}(\mathcal{H})$ to quantum states in $\mathcal{B}(\mathcal{H'})$, complete positivity stems from the physical restriction that a local quantum operation on an arbitrary multiparty system must also be positive. In a finite dimensional setting: $\mathcal{H}\simeq \mathbb{C}^d, \mathcal{H'}\simeq \mathbb{C}^{d'}$, quantum channels are precisely those linear maps $\Phi: \M{d} \rightarrow \M{d'}$ which are completely positive and trace preserving. Entanglement breaking maps represent noisy physical operations, so much so that a local partial action of such a map on a bipartite physical system destroys all entanglement --- no matter how strong --- present in the input state.

Positive but not completely positive maps, while not physically realizable, are important nevertheless, due to their crucial role in detecting entanglement of bipartite matrices. Using the duality relations in Figure~\ref{fig:cones-PSD} and the Choi-Jamio{\l}kowski isomorphism, it can be shown that a positive semi-definite matrix $X\in \M{d}\otimes \M{d}$ is separable if and only if 
\begin{equation} \label{eq:pos-sep}
    [\Phi \otimes \operatorname{id}](X) \in \PSD_{d^2}
\end{equation}
for all positive maps $\Phi\in \T{d}$. In other words, for every entangled $X\in \M{d}\otimes \M{d}$, there is a positive map $\Phi\in \T{d}$ such that $[\Phi\otimes \operatorname{id}](X) \notin \PSD_{d^2}$, which is said to ``detect'' the entanglement in $X$. Moreover, if $X$ is PPT entangled, then such a $\Phi$ must be non-decomposable. Obviously, Eq.~\eqref{eq:pos-sep} will hold for all positive semi-definite $X$ and completely positive $\Phi$. Hence, the important class of maps --- from the perspective of Entanglement Theory --- is the class of positive but not completely positive (also non-decomposable, if one wishes to study PPT entanglement) maps in $\T{d}$. See \cite[Section VI.B.2]{horodecki2009quantum} or \cite[Section 4]{Chruciski2014positive} and references therein for a much more thorough analysis of the role of positive maps in Entanglement Theory. 

This marks the end of our brief digression on the theory of linear maps between matrix algebras. We are now fully prepared to study different families of covariant maps in $\mathcal{T}_d(\mathbb{C})$. 

\begin{definition}\label{def:DUC-CDUC-DOC}
A linear map $\Phi \in \mathcal{T}_d(\mathbb{C})$ is said to be
\begin{itemize}
    \item \emph{Diagonal Unitary Covariant (DUC)} if 
    $$ \forall Z \in \M{d} \text{ and } U \in \mathcal{DU}_d: \qquad \Phi(UZU^*) = U^*\Phi(Z)U,$$
    \item \emph{Conjugate Diagonal Unitary Covariant (CDUC)} if 
    $$ \forall Z \in \M{d} \text{ and } U \in \mathcal{DU}_d: \qquad \Phi(UZU^*) = U\Phi(Z)U^*,$$
    \item \emph{Diagonal Orthogonal Covariant (DOC)} if 
    $$ \forall Z \in \M{d} \text{ and } O \in \mathcal{DO}_d: \qquad \Phi(OZO) = O\Phi(Z)O.$$
\end{itemize}

\end{definition}

The DUC and CDUC maps were introduced in \cite{liu2015unitary,lopes2015generic}, where they were dubbed \emph{mean unitary conjugation channels}; we use here a different name to mirror the case of invariant bipartite matrices. We will denote the sets of DUC, CDUC and DOC maps in $\T{d}$ by $\DUC_d$, $\CDUC_d$ and $\DOC_d$ respectively. Using the Choi-Jamio{\l}kowski isomorphism, we can immediately start to construct links between the diagonal unitary/orthogonal covariant maps in $\T{d}$ and the local diagonal unitary/orthogonal invariant matrices in $\M{d}\otimes \M{d}$. The following result is a pivotal step in this direction. 

\begin{theorem} \label{theorem:DUC/CDUC/DOC-LDUI/CLDUI/LDOI}
Consider a linear map $\Phi$ in $\mathcal{T}_d(\mathbb{C})$. Then, the following equivalences hold:
\begin{itemize}
    \item $\Phi\in \DUC_d$ if and only if the Choi matrix $J(\Phi)\in \LDUI_d$,
    \item $\Phi\in \CDUC_d$ if and only if the Choi matrix $J(\Phi)\in \CLDUI_d$,
    \item $\Phi\in \DOC_d$ if and only if the Choi matrix $J(\Phi)\in \LDOI_d$.
\end{itemize}
\end{theorem}

\begin{proof}
    Consider an arbitrary $Z\in \mathcal{M}_d(\mathbb{C})$ and $U\in \mathcal{DU}_d$. It is clear that $\Phi\in \DUC_d$ if and only if $\Phi(Z)=\Phi(U^*UZU^*U)=U\Phi(UZU^*)U^*$. Graphically, this condition is equivalent to the diagram given in Figure~\ref{fig:DUC}. It is evident then that $\Phi\in \DUC_d$ if and only if $(U\otimes U) J(\Phi) (U^* \otimes U^*) = J(\Phi)$, i.e., if and only if $J(\Phi) \in \LDUI_d$; this identity can be seen by sliding the $U$ and the $U^*$ boxes in the diagram, and then ignoring the transpositions, since the matrices are diagonal. Similarly, the results for the CDUC and DOC maps can be easily shown to be true, see Figures~\ref{fig:CDUC} and \ref{fig:DOC} below.
    
    \begin{figure}[H]
        \centering
        \includegraphics{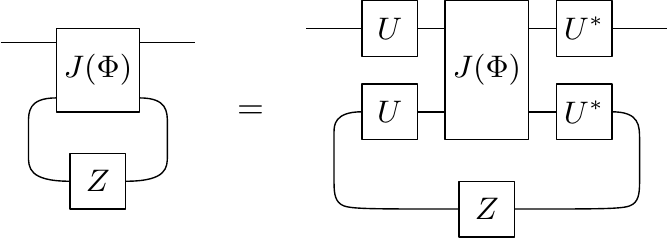}
        \caption{ The DUC condition for $\Phi\in \T{d}$ expressed in terms of $J(\Phi)$.}
        \label{fig:DUC}
    \end{figure}
    \begin{figure}[H]
        \centering
        \includegraphics{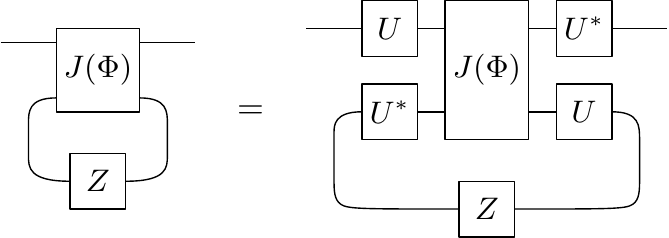}
        \caption{ The CDUC condition for $\Phi\in \mathcal{T}_d(\mathbb{C})$ expressed in terms of $J(\Phi)$.}
        \label{fig:CDUC}
    \end{figure}
    \begin{figure}[H]
        \centering
        \includegraphics{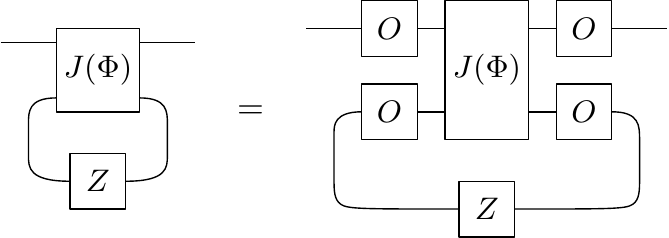}
        \caption{The DOC condition for $\Phi\in \mathcal{T}_d(\mathbb{C})$ expressed in terms of $J(\Phi)$. }
        \label{fig:DOC}
    \end{figure}
\end{proof}

With the help of Theorem~\ref{theorem:DUC/CDUC/DOC-LDUI/CLDUI/LDOI}, the task of extending the isomorphisms from Proposition~\ref{prop:LDOI-ABC} to the vector spaces of diagonal unitary/orthogonal covariant maps in $\T{d}$ becomes effortless:
\begin{align}
     \Phi^{(1)} \coloneqq J^{-1}\circ X^{(1)} :\MLDUI{d} &\rightarrow \DUC_d \nonumber \\
     (A,C) &\mapsto \Phi^{(1)}_{(A,C)}, \label{eq:DUCiso}  \\
     \Phi^{(2)} \coloneqq J^{-1}\circ X^{(2)} :\MLDUI{d} &\rightarrow \CDUC_d \nonumber \\
     (A,B) &\mapsto \Phi^{(2)}_{(A,B)},  \label{eq:CDUCiso}\\
     \Phi^{(3)} \coloneqq J^{-1}\circ X^{(3)} :\MLDOI{d} &\rightarrow \DOC_d \nonumber \\
     (A,B,C) &\mapsto \Phi^{(3)}_{(A,B,C)}. \label{eq:DOCiso}
\end{align}
We collect the explicit actions of $\Phi^{(1)}_{(A,B)}, \Phi^{(2)}_{(A,B)}$ and $\Phi^{(3)}_{(A,B,C)}$ on $\M{d}$ in the following equations:
\begin{equation}
    \forall Z\in \M{d}: \qquad \Phi^{(1)}_{(A,C)}(Z) = \includegraphics[align=c]{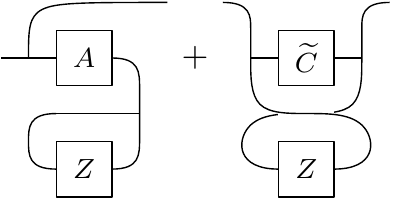} = \operatorname{diag}(A\ket{\operatorname{diag}Z}) + \widetilde{C}\odot Z^\top, \label{eq:DUC-action}
\end{equation}

\begin{align}
    \forall Z\in \M{d}: \qquad \Phi^{(2)}_{(A,B)}(Z) &= \includegraphics[align=c]{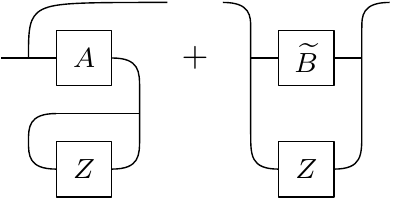} = \operatorname{diag}(A\ket{\operatorname{diag}Z}) + \widetilde{B}\odot Z, \label{eq:CDUC-action} \\
    \forall Z\in \M{d}: \quad\,\, \Phi^{(3)}_{(A,B,C)}(Z) &= \includegraphics[align=c]{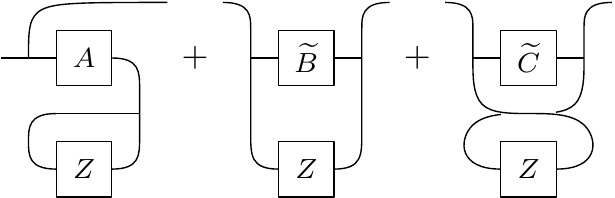} \nonumber \\[0.2cm] &=\operatorname{diag}(A\ket{\operatorname{diag}Z}) + \widetilde{B}\odot Z + \widetilde{C}\odot Z^\top. \label{eq:DOC-action}
\end{align}
Recall that $\widetilde{B} = B-\operatorname{diag}B$ and $\widetilde{C}=C-\operatorname{diag}C$. Let us take a moment to discuss the action of DOC maps on $\M{d}$ in some detail. In the quantum setting, where $\Phi^{(3)}_{(A,B,C)}$ is a quantum channel and $Z=\rho$ is a quantum state, the first term in Eq.~\eqref{eq:DOC-action} can be interpreted as a ``classical'' quantum operation. It takes in the classical probability distribution $\ket{\operatorname{diag}\rho}\in \mathbb{R}^d_+$ defined by the diagonal entries of the state $\rho$ and returns $A\ket{\operatorname{diag}\rho}$, which is again guaranteed to be a probability distribution since $\Phi^{(3)}_{(A,B,C)}$ is completely positive and trace preserving. The last two terms in the sum combine the action of the well known class of Schur multiplier maps (see Example~\ref{eg:maps-schur}) with the transposition map in $\T{d}$.

\begin{remark}\label{remark:choi(DOC)}
From the above discussed isomorphisms, it should be clear that $\DUC_d, \CDUC_d\subset \DOC_d$ are vector subspaces such that $\forall (A,B,C) \in \MLDOI{d}$ ($X^{(i)}$ are defined in Section \ref{sec:LDUI-CLDUI-LDOI}):
\begin{equation}
    J(\Phi^{(1)}_{(A,B)}) = X^{(1)}_{(A,B)}, \qquad J(\Phi^{(2)}_{(A,B)}) = X^{(2)}_{(A,B)}, \qquad J(\Phi^{(3)}_{(A,B,C)}) = X^{(3)}_{(A,B,C)}.
\end{equation}
\end{remark}

We now begin to study the properties of positivity and decomposability for maps in $\DOC_d$.
\begin{theorem}\label{theorem:DOC-pos-decomp}
For a linear map $\Phi^{(3)}_{(A,B,C)}\in \DOC_d$, the following equivalences hold:
\begin{itemize}
    \item $\Phi^{(3)}_{(A,B,C)}$ is positive $\iff$ for all extremal \emph{TCP} triples $(D,E,F)\in \MLDOI{d}$,  $$\operatorname{Tr}(AD^\top + \widetilde{B}\widetilde{E} + \widetilde{C}\widetilde{F}) \geq 0,$$ 
    \item $\Phi^{(3)}_{(A,B,C)}$ is decomposable $\iff$ for all $(D,E,F)\in \MLDOI{d}$ such that $X^{(3)}_{(D,E,F)}$ is \emph{PPT}, $$\operatorname{Tr}(AD^\top + \widetilde{B}\widetilde{E} + \widetilde{C}\widetilde{F}) \geq 0.$$
\end{itemize}
\end{theorem}

\begin{proof}
    We obtain the characterization of positive maps in $\DOC_d$, leaving a near-identical discussion on decomposability to the reader. To begin with, we use the form of LDOI matrices from the discussion following Proposition~\ref{prop:LDOI-ABC} to compute the expression $\operatorname{Tr}[X^{(3)}_{(A,B,C)}X^{(3)}_{(D,E,F)}]$ for arbitrary $X^{(3)}_{(A,B,C)}, X^{(3)}_{(D,E,F)}\in \LDOI_d$:
    \begin{align*}
        \operatorname{Tr}[X^{(3)}_{(A,B,C)}X^{(3)}_{(D,E,F)}] &= \includegraphics[align=c, scale=1.2]{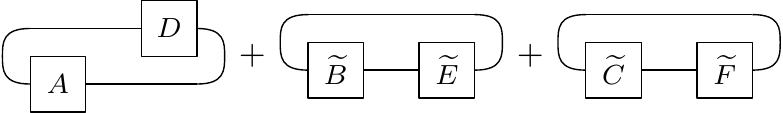} \\
        &= \operatorname{Tr}(AD^\top + \widetilde{B}\widetilde{E} + \widetilde{C}\widetilde{F}).
    \end{align*}
    Now, from Lemma~\ref{lemma:CJiso}, we know that $\Phi^{(3)}_{(A,B,C)}$ is positive if and only if $$\forall \ket{\zeta}, \ket{\eta}\in \mathbb{C}^d: \qquad \operatorname{Tr}[X^{(3)}_{(A,B,C)} \ketbra{\zeta\otimes\eta}] \geq 0.$$
    But, 
    \begin{align*}
        \operatorname{Tr}[X^{(3)}_{(A,B,C)} \ketbra{\zeta\otimes\eta}] &= \operatorname{Tr}[\operatorname{Proj}_{\LDOI}(X^{(3)}_{(A,B,C)}) \ketbra{\zeta\otimes\eta}] \\
        &= \operatorname{Tr}[X^{(3)}_{(A,B,C)} \operatorname{Proj}_{\LDOI}(\ketbra{\zeta\otimes\eta})]
    \end{align*}
    and Remark~\ref{remark:extremal-separable} tells us that all the extremal TCP rays are of the form $\operatorname{Proj}_{\LDOI}(\ketbra{\zeta\otimes\eta})$, for some $\ket{\zeta}, \ket{\eta}\in \mathbb{C}^d \setminus \{0\}$. This completes the proof.
\end{proof}

Although elegant, Theorem~\ref{theorem:DOC-pos-decomp} will seldom be of practical use, as the stated conditions are too hard to check in practice. Drawing motivation from \cite{Chruscinski2018choi}, we try to remedy this situation in the next couple of results. We first derive some easily verifiable constraints on matrix triples $(A,B,C)$ which are necessary to ensure that the corresponding maps in $\DOC_d$ are positive, see Proposition~\ref{prop:necessary-DOCpos}. Then, we present a set of sufficient conditions on triples $(A,B,C)$ which guarantee that the associated maps in $\DOC_d$ are both positive and decomposable, see Proposition~\ref{prop:sufficient-DOCpos}. 

\begin{proposition} \label{prop:necessary-DOCpos}
Let $(A,B,C)\in \MLDOI{d}$ be such that $\Phi^{(3)}_{(A,B,C)} \in \DOC_d$ is positive. Then,
\begin{itemize}
    \item $A$ is entrywise non-negative and $B,C$ are self-adjoint,
    \item $\sqrt{A_{ii}A_{jj}} + \sqrt{A_{ij}A_{ji}} - |B_{ij}| - |C_{ij}| \geq 0$ for all $i\neq j$.
\end{itemize}
\end{proposition}

\begin{proof}
Since $\Phi^{(3)}_{(A,B,C)}$ is positive, it is hermiticity preserving as well. Lemma~\ref{lemma:DOC-prop} then tells us that $A$ is real and $B,C$ are self-adjoint. Now, consider an (arbitrary) extremal TCP ray $(D,E,F)\in \MLDOI{d}$ of the form $D = |\zeta\odot\bar{\zeta}\rangle\langle \eta\odot\bar{\eta}|$, $E = |\zeta\odot\eta\rangle\langle \zeta\odot\eta| $ and $ F = |\zeta\odot\bar{\eta}\rangle\langle \zeta\odot\bar{\eta}| $, for some non-zero vectors $\ket{\zeta}, \ket{\eta}\in \mathbb{C}^d$, see Theorem~\ref{theorem:extremal-PCPTCP}. By invoking Theorem~\ref{theorem:DOC-pos-decomp}, we can write: 
\begin{equation*}
    f(\zeta, \eta) = \operatorname{Tr}[AD^\top + \widetilde{B}\widetilde{E} + \widetilde{C}\widetilde{F}] = \sum_{i,j=1}^d A_{ij}|\zeta_i |^2 |\eta_j |^2 + \sum_{1\leq i\neq j \leq d} \left(B_{ij}\overbar{\zeta_i \eta_i}\zeta_j \eta_j + C_{ij} \overbar{\zeta_i}\eta_i \zeta_j \overbar{\eta_j}\right) \geq 0.
\end{equation*}  
If we fix $k,l \in [d]$ and choose $\ket{\zeta}, \ket{\eta} \in \mathbb{C}^d$ such that $\zeta_k = \eta_l = 1$ and $\zeta_i = \eta_j = 0$ if $(i,j)\neq (k,l)$, then $f(\zeta, \eta) \geq 0 \iff A_{kl}\geq 0$ for all $k,l\in [d]$, i.e., $A$ is entrywise non-negative. Now, let us try to further expand the above expression:
\begin{align*}
    f(\zeta,\eta) &= \sum_{i,j=1}^d A_{ij}|\zeta_i |^2 |\eta_j |^2 + \sum_{1\leq i\neq j \leq d} \left(B_{ij}\overbar{\zeta_i \eta_i}\zeta_j \eta_j + C_{ij} \overbar{\zeta_i}\eta_i \zeta_j \overbar{\eta_j}\right) \\
    &= \sum_{i=1}^d A_{ii}|\zeta_i |^2 |\eta_i |^2 + \sum_{1\leq i<j\leq d} \left(A_{ij}|\zeta_i |^2 |\eta_j |^2 + A_{ji}|\zeta_j |^2 |\eta_i |^2 \right) \\
     &\qquad + \sum_{1\leq i<j\leq d} \left(B_{ij}\overbar{\zeta_i \eta_i}\zeta_j \eta_j + B_{ji}\overbar{\zeta_j \eta_j}\zeta_i \eta_i + C_{ij} \overbar{\zeta_i}\eta_i \zeta_j \overbar{\eta_j} + C_{ji} \overbar{\zeta_j}\eta_j \zeta_i \overbar{\eta_i}   \right)  \\
     &= \sum_{i=1}^d A_{ii}|\zeta_i |^2 |\eta_i |^2 + \sum_{1\leq i<j\leq d} \left( \sqrt{A_{ij}}|\zeta_i | |\eta_j | - \sqrt{A_{ji}}|\zeta_j | |\eta_i | \right)^2  \\
     &\qquad + \sum_{1\leq i<j\leq d} \left\{    2\sqrt{A_{ij}A_{ji}} |\zeta_i\eta_i\zeta_j\eta_j | + 2\operatorname{Re}(B_{ij}\overbar{\zeta_i\eta_i}\zeta_j\eta_j) +  2\operatorname{Re}(C_{ij}\overbar{\zeta_i}\eta_i\zeta_j\overbar{\eta_j}) \right\} \geq 0 .
\end{align*}
If we impose the constraint that $k< l$ and choose $\ket{\zeta}, \ket{\eta} \in \mathbb{C}^d$ such that $\zeta_i = \eta_i = 0$ for all $i \notin \{k,l\}$, it is clear that
\begin{align}
    f(\zeta, \eta) &= A_{kk}|\zeta_k |^2 |\eta_k |^2 + A_{ll}|\zeta_l |^2 |\eta_l |^2 + \left( \sqrt{A_{kl}}|\zeta_k | |\eta_l | - \sqrt{A_{lk}}|\zeta_l | |\eta_k | \right)^2  \nonumber \\     
    &\quad + 2\sqrt{A_{kl}A_{lk}} |\zeta_k\eta_k\zeta_l\eta_l | + 2\operatorname{Re}(B_{kl}\overbar{\zeta_k\eta_k}\zeta_l\eta_l) +  2\operatorname{Re}(C_{kl}\overbar{\zeta_k}\eta_k\zeta_l\overbar{\eta_l}) \nonumber \\
    &= \left( \sqrt{A_{kk}}|\zeta_k | |\eta_k | - \sqrt{A_{ll}}|\zeta_l | |\eta_l | \right)^2 + \left( \sqrt{A_{kl}}|\zeta_k | |\eta_l | - \sqrt{A_{lk}}|\zeta_l | |\eta_k | \right)^2 \nonumber \\
    &\quad + 2 \left( \sqrt{A_{kk}A_{ll}} + \sqrt{A_{kl}A_{lk}} \right) |\zeta_k\eta_k\zeta_l\eta_l | + 2\operatorname{Re}(B_{kl}\overbar{\zeta_k\eta_k}\zeta_l\eta_l) +  2\operatorname{Re}(C_{kl}\overbar{\zeta_k}\eta_k\zeta_l\overbar{\eta_l}) \geq 0. \nonumber
\end{align}
It is now possible to choose $\{ \zeta_i, \eta_i \}_{i=k,l}$ in such a way that $f(\zeta, \eta) \geq 0$ implies the following inequality, which holds for all $x_k,y_k, x_l, y_l \geq 0$ (recall that $f(\zeta,\eta)\geq 0$ holds for all $\ket{\zeta},\ket{\eta}\in \C{d}$):
\begin{align}
     \left( \sqrt{A_{kk}} x_k y_k  - \sqrt{A_{ll}} x_l y_l  \right)^2 &+ \left( \sqrt{A_{kl}} x_k y_l  - \sqrt{A_{lk}} x_l y_k  \right)^2 \nonumber \\
     &+2 \left\{ \sqrt{A_{kk}A_{ll}} + \sqrt{A_{kl}A_{lk}}  - |B_{kl}| - |C_{kl}| \right\} x_k y_k x_l y_l \geq 0. \nonumber
\end{align}
Assuming that $A_{kk},A_{ll},A_{kl}$ and $A_{lk}$ are non-zero, we can set $x_l = \displaystyle\left(\frac{A_{kl}A_{kk}}{A_{lk}A_{ll}} \right)^{1/4}$, $y_k = \displaystyle\left(\frac{A_{kl}A_{ll}}{A_{lk}A_{kk}} \right)^{1/4}$ and $x_k = y_l = 1$ to obtain:
\begin{equation}
    2\left\{ \sqrt{A_{kk}A_{ll}} + \sqrt{A_{kl}A_{lk}}  - |B_{kl}| - |C_{kl}| \right\} \sqrt{\frac{A_{kl}}{A_{lk}}} \geq 0, \nonumber
\end{equation}
which is equivalent to the desired inequality as stated in the Lemma:
\begin{equation}
    \sqrt{A_{kk}A_{ll}} + \sqrt{A_{kl}A_{lk}}  - |B_{kl}| - |C_{kl}| \geq 0. \nonumber
\end{equation}
If one or more entries of $A$ are zero, we can set them equal to an arbitrarily small non-zero value to ensure that the above inequalities hold, and the final result will then follow by taking limits.  
\end{proof}

It turns out that only a slight modification in the necessary condition on $(A,B,C)\in \MLDOI{d}$ in Proposition~\ref{prop:necessary-DOCpos} suffices to guarantee positivity as well as decomposability of $\Phi^{(3)}_{(A,B,C)}\in \DOC_d$.

\begin{proposition} \label{prop:sufficient-DOCpos}
Let $(A,B,C)\in \MLDOI{d}$ be such that
\begin{itemize}
    \item $A$ is entrywise non-negative and $B,C$ are self-adjoint,
    \item $\displaystyle\frac{\sqrt{A_{ii}A_{jj}}}{d-1} + \sqrt{A_{ij}A_{ji}} - |B_{ij}| - |C_{ij}| \geq 0$ for all $i\neq j$.
\end{itemize}
Then $\Phi^{(3)}_{(A,B,C)}\in \DOC_d$ is decomposable and thus positive.
\end{proposition}

\begin{proof}
    Since positivity trivially follows from decomposability, it is sufficient to show that $\Phi^{(3)}_{(A,B,C)}$ is decomposable. From Theorem~\ref{theorem:DOC-pos-decomp}, this is equivalent to showing that $\operatorname{Tr}[AD^\top + \widetilde{B}\widetilde{E}+ \widetilde{C}\widetilde{F}] \geq 0$ for all PPT matrices $X^{(3)}_{(D,E,F)}\in \LDOI_d$. Recall from Lemma~\ref{lemma:LDOI-psd-ppt} that $$X^{(3)}_{(D,E,F)} \text{ is PPT } \iff D\in\EWP_d, E,F\in\PSD_d \text{ and } D_{ij}D_{ji} \geq |E_{ij}|^2, |F_{ij}|^2 \,\, \forall i,j\in [d].$$
    Moreover, since $\operatorname{diag}D=\operatorname{diag}E = \operatorname{diag}F$ and $E,F\in \PSD_d$, the $i,j-$minors of $E,F$ are non-negative: $D_{ii}D_{jj} \geq |E_{ij}|^2, |F_{ij}|^2 \,\, \forall i,j\in [d]$. We now prove that the desired trace expression is non-negative:
    \begin{align}
        \operatorname{Tr}[AD^\top + \widetilde{B}\widetilde{E}+ \widetilde{C}\widetilde{F}] &= \sum_{i,j=1}^d A_{ij}D_{ij} + \sum_{1\leq i<j \leq d} \left\{ 2\operatorname{Re}(B_{ij}E_{ji}) + 2\operatorname{Re}(C_{ij}F_{ji}) \right\} \nonumber \\
        &= \sum_{i=1}^d A_{ii}D_{ii} + \sum_{1\leq i<j\leq d} \left\{ A_{ij}D_{ij} + A_{ji}D_{ji} + 2\operatorname{Re}(B_{ij}E_{ji}) + 2\operatorname{Re}(C_{ij}F_{ji})  \right\} \nonumber \\
        &= \sum_{1\leq i<j \leq d} \left\{ \frac{A_{ii}D_{ii} + A_{jj}D_{jj}}{d-1} + A_{ij}D_{ij} + A_{ji}D_{ji} + 2\operatorname{Re}(B_{ij}E_{ji}) + 2\operatorname{Re}(C_{ij}F_{ji}) \right\} \nonumber \\
        &\geq 2\sum_{1\leq i<j\leq d} \left\{ \frac{\sqrt{A_{ii}D_{ii}A_{jj}D_{jj}}}{d-1} + \sqrt{A_{ij}A_{ji}D_{ij}D_{ji}} -|B_{ij}E_{ij}| - |C_{ij}F_{ij}|  \right\} \nonumber \\ 
        &\geq 2\sum_{1\leq i<j \leq d} \operatorname{max}\{|E_{ij}|,|F_{ij}|\} \left\{ \frac{\sqrt{A_{ii}A_{jj}}}{d-1} + \sqrt{A_{ij}A_{ji}} - |B_{ij}| - |C_{ij}|  \right\} \geq 0, \nonumber
    \end{align}
where the ante-penultimate inequality follows from the arithmetic-geometric mean inequality, the penultimate inequality follows from the fact that $X^{(3)}_{(D,E,F)}\in \LDOI_d$ is PPT, and the final inequality follows from the hypothesis of the Proposition. This completes the proof.
\end{proof}

It is clear that the conditions in Propositions~\ref{prop:necessary-DOCpos} and \ref{prop:sufficient-DOCpos} are equivalent for $d=2$. This leads us to the following complete characterization of positivity for maps in $\DOC_2$, which generalizes similar results in \cite{Li1997diagonal, Kye1995diagonal} --- these were obtained for the restricted class of positive maps which preserve diagonals (see Example~\ref{eg:maps-diag-preserve}).
\begin{corollary}
For $(A,B,C)\in \MLDOI{2}$, the map $\Phi^{(3)}_{(A,B,C)}\in \DOC_2$ is positive if and only if 
\begin{itemize}
    \item $A$ is entrywise non-negative and $B,C$ are self-adjoint,
    \item $\sqrt{A_{11}A_{22}} + \sqrt{A_{12}A_{21}} - |B_{12}| - |C_{12}| \geq 0$.
\end{itemize}

\end{corollary}

Equipped with the isomorphisms from Eqs.~\eqref{eq:DUCiso}, \eqref{eq:CDUCiso} and \eqref{eq:DOCiso}, we now begin to cast several important properties of linear maps in $\DUC_d$, $\CDUC_d$ and $\DOC_d$ into appropriate constraints on the associated matrix pairs/triples. This forms the content of the next three Lemmas, which are straightforward consequences of the results from Theorem~\ref{theorem:LDOI-sep}, Lemma~\ref{lemma:LDOI-psd-ppt}, Lemma~\ref{lemma:CJiso} and Theorem~\ref{theorem:DUC/CDUC/DOC-LDUI/CLDUI/LDOI}. We leave the proofs of these to the incisive sense of the reader.

\begin{lemma} \label{lemma:DUC-prop}
Consider $(A,B)\in \MLDUI{d}$. Then, the associated map $\Phi^{(1)}_{(A,B)}\in \DUC_d$ is 
\begin{enumerate}
    \item hermiticity preserving $\iff A\in \Mreal{d}$ and $B\in \Msa{d}$,
    \item completely positive $\iff A \in \EWP_d$, $B\in \Msa{d}$ and $A_{ij}A_{ji} \geq \vert B_{ij} \vert^2 \,\, \forall i,j\in [d]$,
    \item completely copositive $\iff A \in \EWP_d$ and $B \in \PSD_d$,
    \item entanglement breaking $\iff (A,B)$ is \emph{PCP}.
\end{enumerate}
\end{lemma}

\begin{lemma} \label{lemma:CDUC-prop}
Consider $(A,B)\in \MLDUI{d}$. Then, the associated map $\Phi^{(2)}_{(A,B)}\in \CDUC_d$ is 
\begin{enumerate}
    \item hermiticity preserving $\iff A\in \Mreal{d}$ and $B\in \Msa{d}$,
    \item completely positive $\iff A \in \EWP_d$ and $B \in \PSD_d$,
    \item completely copositive $\iff A \in \EWP_d$, $B\in \Msa{d}$, and $A_{ij}A_{ji} \geq \vert B_{ij} \vert^2 \,\, \forall i,j\in [d]$,
    \item entanglement breaking $\iff (A,B)$ is \emph{PCP}.
\end{enumerate}
\end{lemma}

\begin{lemma} \label{lemma:DOC-prop}
Consider $(A,B,C)\in \MLDOI{d}$. Then, the associated map $\Phi^{(3)}_{(A,B,C)}\in \DOC_d$ is 
\begin{enumerate}
    \item hermiticity preserving $\iff A\in \Mreal{d}$ and $B,C\in \Msa{d}$,
    \item completely positive $\iff A \in \EWP_d$, $B \in \PSD_d$, $C\in \mathcal{M}^{sa}_d(\mathbb{C})$, and $A_{ij}A_{ji} \geq \vert C_{ij} \vert^2 \,\, \forall i,j\in~[d]$,
    \item completely copositive $\iff A \in \EWP_d$, $B\in \Msa{d}$, $C \in \PSD_d$, and $A_{ij}A_{ji} \geq \vert B_{ij} \vert^2 \,\, \forall i,j\in~[d]$,
    \item entanglement breaking $\iff (A,B,C)$ is \emph{TCP}.
\end{enumerate}
\end{lemma}

The next result presents conditions on matrix pairs/triples which ensure that the corresponding DUC/CDUC/DOC maps are unital and trace preserving. This will be of immediate application in describing the class of quantum DOC channels, which forms the content of Lemma~\ref{lemma:DOCquantum} below.

\begin{lemma}\label{lemma:DUC/CDUC/DOC-unital/tr}
Consider $(A,B,C) \in \MLDOI{d}$. Then, the maps $\Phi^{(1)}_{(A,B)}, \Phi^{(2)}_{(A,B)}$ and $\Phi^{(3)}_{(A,B,C)}$ in $\DUC_d$, $\CDUC_d$ and $\DOC_d$ respectively, are
\begin{enumerate}
    \item unital if and only if $\sum_{j}A_{ij}=1 \,\, \forall i\in [d]$,
    \item trace preserving if and only if $\sum_{i}A_{ij}=1 \,\, \forall j\in [d]$.
\end{enumerate}
\end{lemma}

\begin{proof}
    We tackle only the CDUC case here. Other cases can be proven similarly. To this end, first note that Lemma \ref{lemma:unital-Tr-map} expresses the unitality and trace preserving property of a map $\Phi\in \T{d}$ in terms of partial trace conditions on its Choi matrix $J(\Phi)$. The desired result is then a consequence of the fact that for $\Phi^{(2)}_{(A,B)} \in \CDUC_d$, these conditions are equivalent to the two diagrams given in Figure \ref{fig:DUC-1-unital-tr} below (the last terms in both the diagrams cancel because $ \operatorname{diag}\widetilde{B}=0$).
    \begin{figure}[htbp]
        \centering
        \includegraphics{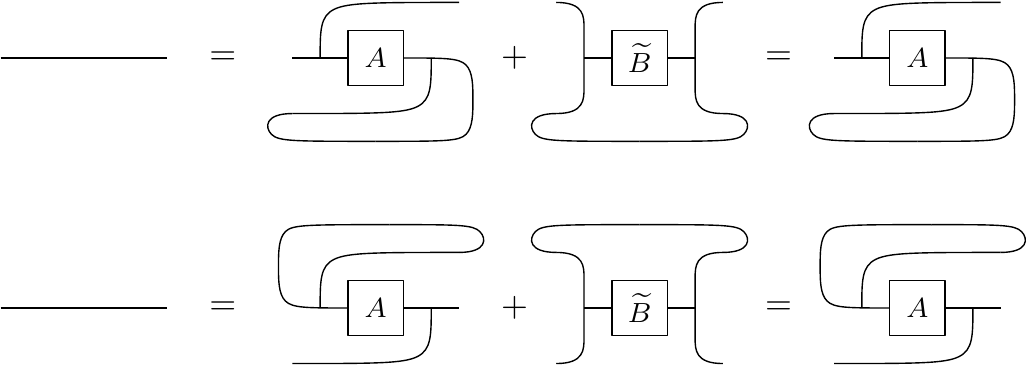}
        \caption{ Unitality (top panel) and trace preserving property (bottom panel) of a CDUC map $\Phi$ expressed in terms of the associated matrix pair $(A,B)$. \qedhere}
        \label{fig:DUC-1-unital-tr}
    \end{figure}
\end{proof}

Recall that a matrix $A$ in $\mathcal{M}_d(\mathbb{C})$ is called \emph{row} (resp.~\emph{column}) \emph{stochastic} if it is entrywise non-negative and the sum of entries in each row (resp.~column) equals one. By combining the results from Lemmas~\ref{lemma:DOC-prop} and \ref{lemma:DUC/CDUC/DOC-unital/tr}, it is straightforward to present a characterization of the class of quantum diagonal orthogonal covariant channels in $\T{d}$.

\begin{lemma} \label{lemma:DOCquantum}
A linear map $\Phi^{(3)}_{(A,B,C)}\in \DOC_d$ is a quantum channel $\iff A$ is column stochastic, $B\in \PSD_d$ and $C\in \Msa{d}$ such that $A_{ij}A_{ji} \geq \vert C_{ij} \vert^2 \,\, \forall i,j\in [d]$.
\end{lemma}

We now discuss the symmetries of the Choi matrices described in Proposition \ref{prop:leg-permutations} in terms of the corresponding maps. 

\begin{proposition}\label{prop:leg-permutations-map}
For any linear map $\Phi^{(3)}_{(A,B,C)}\in \DOC_d$, we have
\begin{align*}
    \left(\Phi^{(3)}_{(A,B,C)}\right)^{*} &= \Phi^{(3)}_{(A^\top,B^\top,C)}, \\
    \Phi^{(3)}_{(A,B,C)} \circ \top &= \Phi^{(3)}_{(A,C,B)}, \\
    \top \circ \Phi^{(3)}_{(A,B,C)}  &= \Phi^{(3)}_{(A,C^{\top},B^{\top})}, 
\end{align*}
where the $*$ denotes the dual map (i.e.~the adjoint of the map with respect to the Hilbert-Schmidt inner product in $\mathcal M_d(\mathbb C)$), and $\top$ is the transposition map. In particular, the composition of a \emph{CDUC} (resp.~\emph{DUC, DOC}) map with the transposition yields a \emph{DUC} (resp.~\emph{CDUC, DOC}) map, while adjoint preserves the three classes \emph{CDUC, DUC, DOC}. 
\end{proposition}
\begin{proof}
    The stated equalities are simple consequences of Proposition \ref{prop:leg-permutations}.
\end{proof}

\section{Important classes of DOC maps} \label{sec:DOCexamples}

This section contains a number of examples of (classes of) DOC maps. The examples of LDOI bipartite matrices from Section \ref{sec:LDOIexamples} can be seen, through the Choi-Jamio{\l}kowski isomorphism, as examples of DOC maps. We list some of them, together with some important classes of maps discussed in the literature, below. One of the main achievements of the current work is realizing that all these linear maps fall under the same umbrella, and hence can be studied within a unified framework. A neatly summarized list of all the examples is presented in Table \ref{tbl:DOC} towards the end of the section.

\begin{example}[\emph{Identity and transposition}]\ \\[0.1cm] \label{eg:maps-id-transp}
The identity map $\operatorname{id}\in \T{d}$ corresponds to a CDUC map with $A=\mathbb{I}_d$ and $B=\mathbb{J}_d$, whereas the transposition map $\top\in \T{d}$ corresponds to a DUC map with $A=\mathbb{I}_d$ and $C=\mathbb{J}_d$. While the identity map is clearly completely positive, transposition, on the other hand, is the most common example of a positive but not completely positive map. These maps are special examples of the more general class of diagonal-preserving maps, which is discussed in Example~\ref{eg:maps-diag-preserve}.
\end{example}

\begin{example}[\emph{Classical maps}]\ \\[0.1cm] \label{eg:maps-classical} 
Equations~\eqref{eq:DUC-action} and \eqref{eq:CDUC-action} entail that a map $\Phi\in \DOC_d$ is both DUC and CDUC if and only if both the associated matrices $B$ and $C$ are diagonal. These maps are then parameterized by a single matrix $A\in \M{d}$, and have the following action:
\begin{equation}
    \Phi_A (Z) = \operatorname{diag}(A\ket{\operatorname{diag}(Z)}).
\end{equation}
From the above equation, one understands that these maps completely discard the off-diagonal entries of its input, and act only on the diagonal part through the matrix $A$. In the quantum setting, these maps only change the classical probability distribution $\ket{\operatorname{diag}\rho}$ associated with the input state $\rho$, which earns them their \emph{classical} nature. Notable examples from this class include the \emph{completely depolarizing} and the \emph{completely dephasing} maps, for which $A=\mathbb{J}_d$ (see Example~\ref{eg:maps-dep+schur+transp}) and $A=\mathbb{I}_d$ (see Example~\ref{eg:maps-diag-preserve}). Positivity, complete positivity and entanglement breaking property for maps in this class are all equivalent to the condition that $A\in \EWP_d$, as is clear from the discussion of the corresponding diagonal Choi matrices in Example~\ref{eg:states-diag}.
\end{example}

\begin{example}[\emph{Schur multipliers}] \cite{harris2018schur}\cite[Chapters 3 and 8]{paulsen2002schur}\cite[Section 4.1.3]{watrous2018theory}\ \\[0.1cm] \label{eg:maps-schur}
Given $S \in \mathcal{M}_d(\mathbb{C})$, the Schur multiplier map $\Phi_S \in \mathcal{T}_d(\mathbb{C})$ is defined as 
\begin{equation}
    \Phi_S(X) = S\odot X, \qquad \forall X \in \M{d}.
\end{equation}
Clearly, $\Phi_S = \Phi^{(2)}_{(A,B)} \in \CDUC_d \subset \DOC_d$ for $B = S$ and $A=\operatorname{diag}(S)$. Using Lemma~\ref{lemma:CDUC-prop}, we can easily see that $\Phi_S$ is completely positive if and only if $S\in \PSD_d$. The same lemma tells us that complete copositivity forces $S$ to be diagonal, implying that the map $\Phi_S$ is PPT if and only if $S$ is diagonal and entrywise non-negative, in which case, it is entanglement breaking as well. 

Quantum channels within the class of Schur multipliers are also known as \emph{generalized dephasing} or \emph{Hadamard} channels. From Lemma~\ref{lemma:DOCquantum}, it is clear that $\Phi_S$ is a quantum channel if and only if $S$ is a correlation matrix (Definition~\ref{def:corr}), see also \cite[Section 4.1.3]{watrous2018theory}. These channels are used to describe a decoherence type noise in quantum systems, since their action on a quantum state preserves the diagonal entries (see Example~\ref{eg:maps-diag-preserve}) and reduces the magnitude of the off-diagonal entries. The extreme case occurs when $S=\mathbb{I}_d$, in which case the associated channel entirely discards the off-diagonal part of the input, thus resulting in a \emph{completely dephasing} action.
\end{example}

\begin{example}[\emph{Unitary and conjugate unitary covariant maps}]\ \\[0.1cm] \label{eg:maps-UC-CUC}
We call a map $\Phi\in \T{d}$ Unitary Covariant (UC) and Conjugate Unitary Covariant (CUC) if for all $Z\in \M{d}$ and unitary matrices $U\in \M{d}$, respectively, $\Phi(UZU^*) = \overbar{U}\Phi(Z)U^\top$ and $\Phi(UZU^*) = U\Phi(Z)U^*$. From Definition~\ref{def:DUC-CDUC-DOC}, it is obvious that UC maps are DUC, and CUC maps are CDUC. Mimicking the proof of Lemma~\ref{theorem:DUC/CDUC/DOC-LDUI/CLDUI/LDOI}, we infer that $\Phi$ is UC (resp.~CUC) if and only if its Choi matrix $J(\Phi)$ is a Werner (resp.~isotropic) matrix, see Example~\ref{eg:states-wer-iso}. From the form of these Choi matrices, we can express the action of UC and CUC maps on $\M{d}$ as follows:
\begin{equation}
    \Phi^{\mathsf{uc}}_{a,b}(Z) = a\operatorname{Tr}(Z)\mathbb{I}_d + bZ^\top, \qquad \Phi^{\mathsf{cuc}}_{a,b}(Z) = a\operatorname{Tr}(Z)\mathbb{I}_d + bZ,
\end{equation}
where $a,b$ are complex numbers. Now, following Example~\ref{eg:states-wer-iso}, we define the matrices $A=b\,\mathbb{I}_d + a\mathbb{J}_d$ and $B=a\mathbb{I}_d + b\mathbb{J}_d$. It then becomes evident that $\Phi^{\mathsf{uc}}_{a,b} = \Phi^{(1)}_{(A,B)} \in \DUC_d$ and $\Phi^{\mathsf{cuc}}_{a,b} = \Phi^{(2)}_{(A,B)} \in \CDUC_d$. The result of Proposition~\ref{prop:wer-iso-sep} then translates into the fact that these maps are PPT if and only if they are entanglement breaking if and only if $a\geq 0$ and $-a/d \leq b \leq a$.

It is perhaps worthwhile to mention that quantum channels (Lemma~\ref{lemma:DOCquantum}) within the CUC and UC classes of linear maps are important models of quantum noise and are known as depolarizing \cite{King2003depol} and transpose depolarizing \cite{Datta2006trans-depol} channels, respectively. 
\end{example}

\begin{example}[\emph{Choi-type maps}]\ \\[0.1cm] \label{eg:maps-choi}
Consider the Choi map $\Phi_{\mathsf{ch}} \in \mathcal{T}_3(\mathbb{C})$ defined as:
\begin{equation}
    \Phi_{\mathsf{ch}}(Z) = \left(   \begin{array}{ccc}
        Z_{11}+Z_{33} & -Z_{12} & -Z_{13}  \\
        -Z_{21} & Z_{11}+Z_{22} & -Z_{23} \\
        -Z_{31} & -Z_{32} & Z_{22}+Z_{33}
    \end{array}  \right).
\end{equation}
This was the first example of a positive non-decomposable map between matrix algebras, presented by Choi in the '70s \cite{choi1972positive,choi1975,choi1982positive}. Since then, many generalizations of this map have been proposed. In \cite{cho1992choi}, the authors introduced the family $\{ \Phi_{(a,b,c)}^I \in \mathcal{T}_3(\mathbb{C}) : a,b,c \geq 0 \}$ and studied constraints on the triple $(a,b,c)\in \mathbb{R}^3$ which guarantee that the corresponding map is positive/decomposable:
\begin{equation}
    \Phi_{(a,b,c)}^I(Z) =  \left(   \begin{array}{ccc}
        aZ_{11}+bZ_{22}+cZ_{33} & -Z_{12} & -Z_{13}  \\
        -Z_{21} & cZ_{11}+aZ_{22}+bZ_{33} & -Z_{23} \\
        -Z_{31} & -Z_{32} & bZ_{11}+cZ_{22}+aZ_{33}
    \end{array}  \right).
\end{equation}
A slightly different variant $\Phi^{II}_{(a,c_1,c_2,c_3)}\in \mathcal{T}_3(\mathbb{C})$ of the above maps was introduced in \cite{kye1992choi}:
\begin{equation}
    \Phi^{II}_{(a,c_1,c_2,c_3)}(Z) = \left(   \begin{array}{ccc}
        aZ_{11}+c_1Z_{33} & -Z_{12} & -Z_{13}  \\
        -Z_{21} & c_2Z_{11}+aZ_{22} & -Z_{23} \\
        -Z_{31} & -Z_{32} & c_3Z_{22}+aZ_{33}
    \end{array}  \right).
\end{equation}
where $a,c_1,c_2,c_3 \geq 0$. In higher dimensions, the family $\{ \tau_{d,k}\in \mathcal{T}_d(\mathbb{C}) : d\in \mathbb{N}, k=1,2,\ldots ,d-1\}$ has received considerable attention \cite{tanahashi1988choi, osaka1991choi, yamagami1993choi, Ha1998choi}, where the maps are defined in terms of a cyclic permutation matrix $S \in \mathcal{M}_d(\mathbb{C})$:
\begin{equation}
    \tau_{d,k}(Z) = (d-k)\operatorname{diag}(Z) + \sum_{j=1}^k \operatorname{diag}(S^j Z S^{*j}) - Z \label{eq:choi_d,k}.
\end{equation}
Finally, the most general $d-$dimensional Choi maps (parameterized by an entrywise non-negative matrix $A\in \mathcal{M}_d(\mathbb{C})$) have been analyzed in \cite{Ha2003choi, Chruscinski2007choi, Chruscinski2018choi}: \begin{align}
    \Phi_A(Z) = \operatorname{diag}\left(A\ket{\operatorname{diag}Z} \right) - \widetilde{Z}, \label{eq:choi_A}
\end{align}
where $\widetilde{Z} = Z-\operatorname{diag}Z$. It should be evident from Eq.~\eqref{eq:choi_A} that specific choices of the matrix $A$ can be used to retrieve the action of all the previously discussed generalizations of the Choi map. For instance, by choosing $A = (d-k-1)\mathbb{I}_d + \sum_{j=1}^k S^j$, it is easy to check that $\tau_{d,k} = \Phi_A$. Moreover, by choosing $B = \operatorname{diag}(A) -\widetilde{\mathbb{J}_d}$, we can write $\Phi_A = \Phi^{(2)}_{(A,B)}$ for all $A\in \mathcal{M}_d(\mathbb{C})$. This shows that all the generalized Choi maps lie in $\CDUC_d \subset \DOC_d$. The familiar properties of complete positivity, unitality, etc.~of these maps can be studied through the use of Lemmas~\ref{lemma:CDUC-prop} and \ref{lemma:DUC/CDUC/DOC-unital/tr}. 
\end{example}

\begin{example}[\emph{Mixture of completely depolarizing, Schur multiplier and transposition maps}]\ \\[0.1cm] \label{eg:maps-dep+schur+transp}
The completely depolarizing map in $\T{d}$ is defined as $\Phi_{\mathsf{dep}}(Z) = \operatorname{Tr}(Z)\mathbb{I}_d$. For $A=\mathbb{J}_d$ and $B=\mathbb{I}_d$, it is clear that $\Phi_{\mathsf{dep}} = \Phi^{(1)}_{(A,B)} = \Phi^{(2)}_{(A,B)} \in \DUC_d \cap \CDUC_d$, see Example~\ref{eg:maps-classical}. The general class of DOC maps with $A=\mathbb{J}_d$ have the form 
\begin{equation}
    \Phi^{(3)}_{(\mathbb{J}_d, B,C)}(Z) = \Phi_{\mathsf{dep}}(Z) + \widetilde{B}\odot Z + \widetilde{C}\odot Z^\top,
\end{equation}
for matrices $B,C\in \M{d}$ with $\operatorname{diag}B = \operatorname{diag}C = \mathbb{I}_d$, and correspond precisely to the LDOI Choi matrices from Example~\ref{eg:states-A=J}. The same example informs us that these maps are PPT if and only if $B$ and $C$ are correlation matrices, see Definition~\ref{def:corr}. If we restrict ourselves to the DUC/CDUC maps in this class, which are of the form
\begin{align}
    \Phi^{(1)}_{(\mathbb{J}_d, B)} (Z) &= \Phi_{\mathsf{dep}}(Z) + \widetilde{B}\odot Z^\top,  \\
    \Phi^{(2)}_{(\mathbb{J}_d, B)} (Z) &= \Phi_{\mathsf{dep}}(Z) + \widetilde{B}\odot Z,
\end{align}
then Proposition~\ref{prop:A=J-LDUI/CLDUI-sep} can be immediately applied on the corresponding Choi matrices to deduce the following sequence of equivalences for $i=1,2$:
\begin{align}
    \Phi^{(i)}_{(\mathbb{J}_d, B)} \text{ is completely positive} \iff \Phi^{(i)}_{(\mathbb{J}_d, B)} \text{ is PPT} &\iff \Phi^{(i)}_{(\mathbb{J}_d, B)} \text{ is entanglement breaking} \nonumber \\
    &\iff B\in \mathsf{Corr}_d.
\end{align}
\end{example}

\begin{example}[\emph{Diagonal-preserving maps}]\cite{Kye1995diagonal, Li1997diagonal}\ \\[0.1cm] \label{eg:maps-diag-preserve}
In \cite{Kye1995diagonal, Li1997diagonal}, the authors studied the class of linear maps in $\T{d}$ which fix diagonals. The positive maps in this class were shown to be of the form
\begin{equation}
    \Phi_{\tilde X,\tilde Y}(Z) = \mathbb{I}_d\odot Z + \tilde X\odot Z + \tilde Y\odot Z^\top ,
\end{equation}
where $\tilde X,\tilde Y\in \Msa{d}$ have zero diagonals. A distinguished element of this class is the \emph{completely dephasing map} $Z \mapsto \operatorname{diag}(Z)$, which corresponds to the choice $\tilde X=\tilde Y=0$ and is an element of $\DUC_d \cap \CDUC_d$, see Example~\ref{eg:maps-schur} as well. More generally, for $A=\mathbb{I}_d$, $B=\tilde X+\mathbb{I}_d$ and $C=\tilde Y+\mathbb{I}_d$, it is clear that $\Phi_{\tilde X,\tilde Y} = \Phi^{(3)}_{(A,B,C)}\in \DOC_d$.
We utilize Lemma~\ref{lemma:DOC-prop} to infer that $\Phi_{\tilde X,\tilde Y}$ is completely positive if and only if $\tilde Y=0$ and $B=\tilde X+\mathbb{I}_d$ is a correlation matrix, i.e.~$\Phi_{\tilde X,\tilde Y}$ is a Schur multiplier, see Example~\ref{eg:maps-schur}. Positivity of the maps $\Phi_{\tilde X,\tilde Y}\in \T{d}$ was shown to be equivalent to decomposability if and only if $d\leq 3$. This is clearly not true for positive maps in $\DOC_d$, as the celebrated Choi map in $\CDUC_3 \subset \DOC_3$ is positive and non-decomposable.
\end{example}

\begin{example}\label{eg:maps-lambda}
In \cite{miller2015lambda}, the map $\Lambda_3 \in \mathcal{T}_3(\mathbb{C})$ (defined in Eq.~\eqref{eq:maps-lambda-3}) was shown to be positive and non-decomposable. This was later generalized in \cite{rutkowski2015lambda} to the positive non-decomposable map $\Lambda_d \in \mathcal{T}_d(\mathbb{C})$ for arbitrary $d\in \mathbb{N}$, see Eq.~\eqref{eq:maps-lambda-d}. These maps were introduced in an effort to understand the structure of stable subspaces of extremal \emph{bistochastic} (positive maps which are unital and trace-preserving) maps between matrix algebras, see \cite{miller2015lambda}.
\begin{equation}
    \Lambda_3(Z) = \left(   \begin{array}{ccc}
        \frac{1}{2}(Z_{11}+Z_{22}) & 0 & \frac{1}{\sqrt{2}}Z_{13}  \\
        0 & \frac{1}{2}(Z_{11}+Z_{22})  & \frac{1}{\sqrt{2}}Z_{32} \\
        \frac{1}{\sqrt{2}}Z_{31} & \frac{1}{\sqrt{2}}Z_{23} & Z_{33}
    \end{array}  \right) \label{eq:maps-lambda-3} \\
\end{equation}
\begin{equation}
    \Lambda_d(Z) = \frac{1}{d-1}\left(   \begin{array}{ccccc}
        \displaystyle\sum_{i=1}^{d-1}Z_{ii} & & & 0  & \sqrt{d-1}Z_{1,d} \\
         & \ddots  &  & \vdots  & \vdots \\
         &  &  \displaystyle\sum_{i=1}^{d-1}Z_{ii} & 0  & \sqrt{d-1}Z_{d-2,d} \\
        0 & \ldots & 0  & \displaystyle\sum_{i=1}^{d-1}Z_{ii} & \sqrt{d-1}Z_{d,d-1} \\
        \sqrt{d-1}Z_{d,1} & \ldots & \sqrt{d-1}Z_{d,d-2} & \sqrt{d-1}Z_{d-1,d} & (d-1)Z_{dd}  \label{eq:maps-lambda-d}
    \end{array}  \right)
\end{equation}
By defining the matrices $A,B$ and $C \in \mathcal{M}_d(\mathbb{C})$ entrywise as follows
\begin{equation} \label{eq:maps-lambda-AC}
    A_{ij} = \begin{cases}
    \frac{1}{d-1},& \text{if } 1\leq i,j \leq d-1 \\
    1,  & \text{if } i=j=d \\
    0,  & \text{otherwise} ,
\end{cases} \quad C_{ij} = \begin{cases}
    A_{ij},& \text{if } i=j \\
    \frac{1}{\sqrt{d-1}},  & \text{if } (i,j)=(d,d-1) \text{ or } (i,j)=(d-1,d) \\
    0,  & \text{otherwise} ,
\end{cases}
\end{equation}
\begin{equation} \label{eq:maps-lambda-B}
    B_{ij} = \begin{cases}
    A_{ij},& \text{if } i=j \\
    \frac{1}{\sqrt{d-1}}, & \text{if } (i=d \,\text{ and }\, j\leq d-2) \text{ or } (j=d \,\text{ and }\, i\leq d-2) \\
    0,  & \text{otherwise} ,
\end{cases} 
\end{equation}
we observe that $\Lambda_d = \Phi^{(3)}_{(A,B,C)} \in \DOC_d$ for all $d\in \mathbb{N}$.
\end{example} 

We summarize the ensemble of cases discussed in this section in Table \ref{tbl:DOC}.

\begin{table}[htb]
\begin{tabular}{@{}|l|l|l|l|l|l|@{}}
\toprule
         \cellcolor[HTML]{DCD0F4}\textbf{Ex.}            & \cellcolor[HTML]{DCD0F4}\textbf{Name}                                                                  & \cellcolor[HTML]{DCD0F4}\textbf{Defining Characteristic}                                                                                                             & \cellcolor[HTML]{DCD0F4}\begin{tabular}[c]{@{}l@{}}\textbf{Ambient} \\ \textbf{Space}\end{tabular} & \cellcolor[HTML]{DCD0F4} \begin{tabular}[c]{@{}l@{}} \textbf{Associated} \\ \textbf{$(A,B,C)$} \end{tabular}                                                                                                 & \cellcolor[HTML]{DCD0F4}\textbf{References} \\ \midrule
\ref{eg:maps-id-transp}                         & Identity                                                                                      & $\Phi (Z) = Z$                                                                                                                                              & $\CDUC_d$                                                                        & $A=\mathbb{I}_d, B=\mathbb{J}_d$                                                        & ---                                \\ \midrule
\cellcolor[HTML]{DCD0F4}\ref{eg:maps-id-transp}                          & \cellcolor[HTML]{DCD0F4}Transposition                                                         & \cellcolor[HTML]{DCD0F4}$\Phi(Z) = Z^\top$                                                                                                                  & \cellcolor[HTML]{DCD0F4}$\DUC_d$                                                 & \cellcolor[HTML]{DCD0F4}$A=\mathbb{I}_d, C=\mathbb{J}_d$                                     &     \cellcolor[HTML]{DCD0F4}---       \\ \midrule
\ref{eg:maps-classical}                        & Classical                                                                                     & $\Phi (Z) = \operatorname{diag}(A\ket{\operatorname{diag}(Z)})$                                                                                             & \begin{tabular}[c]{@{}l@{}}$\DUC_d\,\ \cap$\\ $\CDUC_d$\end{tabular}                & \begin{tabular}[c]{@{}l@{}}$A\in \M{d}$\\ $B=\operatorname{diag}A$ \end{tabular}                                    & ---                                    \\ \midrule          \cellcolor[HTML]{DCD0F4}\ref{eg:maps-schur}               & \cellcolor[HTML]{DCD0F4} \begin{tabular}[c]{@{}l@{}} \hspace{-0.1cm}Schur \\ \hspace{-0.1cm}Multipliers \end{tabular}                                                     & \cellcolor[HTML]{DCD0F4}$\Phi(Z) = B\odot Z$                                                                                                                & \cellcolor[HTML]{DCD0F4}$\CDUC_d$                                                & \cellcolor[HTML]{DCD0F4}\begin{tabular}[c]{@{}l@{}}$A=\operatorname{diag}B$\\ $B\in \M{d}$\end{tabular}                                  & \cellcolor[HTML]{DCD0F4}\begin{tabular}[c]{@{}l@{}} \cite{paulsen2002schur}, \\ \cite{harris2018schur} \end{tabular}         \\ \midrule
\ref{eg:maps-UC-CUC}                         & \begin{tabular}[c]{@{}l@{}} Unitary \\ Covariant \end{tabular}                                                                             & \begin{tabular}[c]{@{}l@{}}$\Phi(UZU^*) = \overbar{U}\Phi(Z)U^\top$ for\\ all unitary matrices $U\in \M{d}$\end{tabular}                                                 & $\DUC_d$                                                                         & \begin{tabular}[c]{@{}l@{}}$A=b\,\mathbb{I}_d + a\mathbb{J}_d$\\ $C=a\mathbb{I}_d + b\mathbb{J}_d$\end{tabular}                          & ---                                \\ \midrule
\cellcolor[HTML]{DCD0F4}\ref{eg:maps-UC-CUC}                          & \cellcolor[HTML]{DCD0F4}\begin{tabular}[c]{@{}l@{}l@{}}Conjugate \\ Unitary\\ Covariant\end{tabular} & \cellcolor[HTML]{DCD0F4}\begin{tabular}[c]{@{}l@{}}$\Phi(UZU^*) = U\Phi(Z)U^*$ for\\ all unitary matrices $U\in \M{d}$\end{tabular}                         & \cellcolor[HTML]{DCD0F4}$\CDUC_d$                                                & \cellcolor[HTML]{DCD0F4}\begin{tabular}[c]{@{}l@{}}$A=b\,\mathbb{I}_d + a\mathbb{J}_d$ \\ $B=a\mathbb{I}_d + b\mathbb{J}_d$\end{tabular} & \cellcolor[HTML]{DCD0F4}---        \\  \midrule
\ref{eg:maps-choi}                          & Choi-type                                                             & $\Phi(Z) = \operatorname{diag}(A\ket{\operatorname{diag}Z}) -\widetilde{Z}$                                                         & $\CDUC_d$                                                & \begin{tabular}[c]{@{}l@{}}$A\in \M{d}$ \\ $B=\operatorname{diag}A - \widetilde{\mathbb{J}_d}$\end{tabular}      &   
\begin{tabular}[c]{@{}l@{}} \cite{choi1975}, \\   \cite{Chruscinski2018choi} \end{tabular}     \\ \midrule
\cellcolor[HTML]{DCD0F4}\ref{eg:maps-dep+schur+transp}                          & \cellcolor[HTML]{DCD0F4}DOC $A=\mathbb J$                                                                                       & \cellcolor[HTML]{DCD0F4} \begin{tabular}[c]{@{}l@{}l@{}} Mixture of completely \\ depolarizing, Schur multiplier, \\ and transposition maps.  \end{tabular}                        & \cellcolor[HTML]{DCD0F4} $\DOC_d$                                                                         &\cellcolor[HTML]{DCD0F4}\begin{tabular}[c]{@{}l@{}l@{}l@{}}$A=\mathbb{J}_d$ \\ $B,C\in \M{d}$ \\ $\operatorname{diag}B = \mathbb{I}_d$ \\ $\operatorname{diag}C = \mathbb{I}_d$ \end{tabular}                                                    &            \cellcolor[HTML]{DCD0F4}---                         \\ \midrule
\ref{eg:maps-diag-preserve}                          & \begin{tabular}[c]{@{}l@{}} Diagonal \\ Preserving \end{tabular}                                                                           & \begin{tabular}[c]{@{}l@{}}Positive maps with \\ $\operatorname{diag}\Phi(Z)=\operatorname{diag}Z$\end{tabular}                                             & $\DOC_d$                                                                         & \begin{tabular}[c]{@{}l@{}@{}}$A=\mathbb{I}_d$\\ $B,C\in \Msa{d}$\\ $\operatorname{diag}B=\mathbb{I}_d$ \\ $\operatorname{diag}C=\mathbb{I}_d$\end{tabular}     &    \begin{tabular}[c]{@{}l@{}} \cite{Kye1995diagonal}, \\ \cite{Li1997diagonal}  \end{tabular}                                \\ \midrule
\cellcolor[HTML]{DCD0F4}\ref{eg:maps-lambda} & \cellcolor[HTML]{DCD0F4}Example~\ref{eg:maps-lambda}                                & \cellcolor[HTML]{DCD0F4}\begin{tabular}[c]{@{}l@{}}Related to the characterization \\ of stable subspaces of \\ extremal bistochastic maps.\end{tabular} & \cellcolor[HTML]{DCD0F4}$\DOC_d$                                                 & \cellcolor[HTML]{DCD0F4}Eqs.~\eqref{eq:maps-lambda-AC}, \eqref{eq:maps-lambda-B}                                                                           & \cellcolor[HTML]{DCD0F4}\begin{tabular}[c]{@{}l@{}}\cite{miller2015lambda}, \\ \cite{rutkowski2015lambda}  \end{tabular}                       \\ \bottomrule
\end{tabular}
\medskip
\caption{Important classes of maps in $\DOC_d$. For more details, see the appropriate examples.}
\label{tbl:DOC}
\end{table}

\section{Kraus and Stinespring characterizations of DOC maps}\label{sec:kraus}

This section aims to study maps in $\DOC_d$ in terms of their so-called Kraus and Stinespring representations. We begin with a brief review of these representations for arbitrary linear maps in $\T{d}$, and then proceed to give a general uniqueness result which links different minimal Kraus and Stinespring representations of a given map. This result will then be used to provide necessary and sufficient conditions on the minimal Kraus/Stinespring representation of a given map, in order for it to be DOC.  

Given a linear map $\Phi\in \T{d}$, it admits a representation of the form
\begin{equation}\label{eq:krausrep}
    \Phi (X) = \sum_{i=1}^n P_i X Q^*_i,
\end{equation}
known as a \emph{Kraus} representation, where $\{P_i \}_{i=1}^n, \{ Q_i \}_{i=1}^n \subset \M{d}$ are known as Kraus operators associated with the stated representation of $\Phi$. For a given map $\Phi$, the minimal number $r$ of operators needed for such a representation to exist is known as the \emph{Choi rank} of the map, which can be easily shown to be equal to the rank of its Choi matrix $J(\Phi)$; in the case of invariant matrices, see Corollary \ref{cor:rank-X}. A representation of $\Phi$ which uses the minimal number of Kraus operators $n=r=\operatorname{rank}(J(\Phi))$ is said to be \emph{minimal}.

Given such a representation, we can define tensors $P,Q\in  \C{n}\otimes \M{d}$ (these can be thought of as mappings $P,Q:\C{d}\rightarrow \C{n}\otimes\C{d}$) as follows: 
\begin{equation}\label{eq:stinespringAB}
    P = \sum_{i=1}^n \ket{i} \otimes P_i = \includegraphics[align=c,scale=1.1]{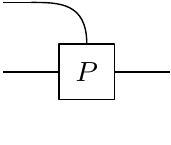}, \qquad Q = \sum_{i=1}^n \ket{i} \otimes Q_i = \includegraphics[align=c,scale=1.1]{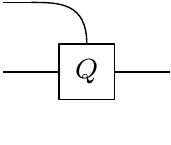},
\end{equation}
where $\C{n}$ acts as an auxiliary space, such that the action of the map can be expressed as follows:
\begin{equation} \label{eq:stinespringrep}
    \Phi(X) = \operatorname{Tr}_{\C{n}}(PXQ^*) = \includegraphics[align=c,scale=1.1]{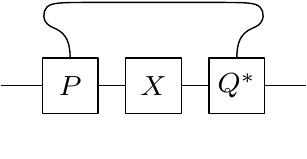} \cdot
\end{equation}
This is known as a \emph{Stinespring} representation of the given map. It is not too difficult to retrace one's steps in order to interchange between the Stinespring and Kraus representations. Graphically, the strings popping out vertically from the boxes correspond to the auxiliary space $\C{n}$. 

If the map $\Phi\in \T{d}$ is completely positive, then the Kraus operators $\{P_i \}_{i=1}^n, \{ Q_i \}_{i=1}^n \subset \M{d}$ are forced to be equal, i.e., $P_i = Q_i \,\, \forall i\in [n]$, and we get a representation of the form
\begin{equation}
    \Phi(X) = \sum_{i=1}^n P_i X P^*_i = \operatorname{Tr}_{\C{n}}(PXP^*).
\end{equation}

With the relevant background in place, we now link different minimal Kraus representations of an arbitrary map $\Phi\in \T{d}$ in the following Lemma.

\begin{lemma} \label{ref:lemma-krausunique}
Consider a linear map $\Phi\in \T{d}$ with $r = \operatorname{rank}(J(\Phi))$, such that it admits the following minimal Kraus and Stinespring representations, with operators $\{P_i,Q_i,R_i,S_i \}_{i=1}^r \subset \M{d}$ and $P,Q,R,S\in \C{r}\otimes \M{d}$ (or $P,Q,R,S: \C{d}\rightarrow \C{r}\otimes \C{d}$):
\begin{equation*}
    \Phi (X) = \sum_{i=1}^r P_i X Q^*_i = \sum_{i=1}^r R_i X S^*_i ,
\end{equation*}
\begin{equation*}
    \Phi(X) = \operatorname{Tr}_{\C{r}}(PXQ^*) = \operatorname{Tr}_{\C{r}}(RXS^*).
\end{equation*}
Then, there exists an invertible matrix $Z\in \M{r}$ such that the following equivalent relations hold:
\begin{align*}
    P_i = \sum_{j=1}^r Z_{ij}R_j \qquad&\text{and}\qquad Q_i = \sum_{j=1}^r (Z^*)^{-1}_{ij} S_j, \\
    P = [Z^\top \otimes \mathbb{I}_d]R \qquad&\text{and}\qquad Q = [\overbar{Z^{-1}}\otimes \mathbb{I}_d]S ,
\end{align*}
where $\mathbb{I}_d\in \M{d}$ is the identity matrix. In case $\Phi\in \T{d}$ is completely positive so that $P_i = Q_i$ and $R_i = S_i \,\, \forall i\in [r]$, the invertible matrix $Z\in \M{r}$ above is also unitary.
\end{lemma}

\begin{proof}
    From the given representations of $\Phi$, it is easy to see that the Choi matrix $J(\Phi)$ has the following rank one decompositions:
    \begin{align}
        J(\Phi) &= \sum_{i=1}^r \ketbra{\operatorname{vec}P_i}{\operatorname{vec}Q_i} = \sum_{i=1}^r \ketbra{\operatorname{vec}R_i}{\operatorname{vec}S_i} \label{eq:fullrankchoi} \\
        &= \includegraphics[align=c,scale=1.1]{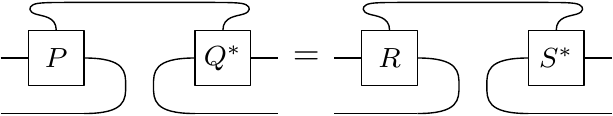}, \nonumber
 \end{align}
 where $\operatorname{vec}:\M{d}\rightarrow \C{d}\otimes \C{d}$ is defined graphically as $$\includegraphics[align=c, scale=1.1]{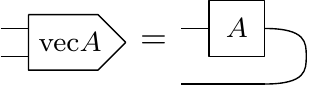} = \sum_{i,j=1}^d A_{ij} \ket{ij}.$$ Now, since $r=\operatorname{rank}(J(\Phi))$, Eq.~\eqref{eq:fullrankchoi} represents two \emph{full rank factorizations} of $J(\Phi)$. Hence, from the uniqueness of full rank factorizations (see \cite[Theorem 2]{Piziak1999fullfact}), there exists an invertible matrix $Z\in \M{r}$ such that the required relations hold: $$ P = [Z^\top \otimes \mathbb{I}_d]R, \qquad Q = [\overbar{Z^{-1}} \otimes \mathbb{I}_d]S. $$  
 Graphically, the above uniqueness result can be visualized by imagining that the wire connecting the $R$ and $S^*$ matrices in Eq.~\eqref{eq:fullrankchoi} is replaced by the identity matrix $\mathbb{I}_r = ZZ^{-1}$. Finally, if $\Phi$ is completely positive, it is trivial to see that $Z\in \M{r}$ must be unitary.
\end{proof}

Equipped with Lemma~\ref{ref:lemma-krausunique}, we now present the two main results of this section.

\begin{theorem} \label{theorem:DOCkraus}
Consider a linear map $\Phi\in \T{d}$ with $r = \operatorname{rank}(J(\Phi))$, such that 
\begin{equation*}
    \Phi (X) = \sum_{i=1}^r P_i X Q^*_i = \operatorname{Tr}_{\C{r}}(PXQ^*),
\end{equation*}
where $\{P_i \}_{i=1}^r , \{Q_i \}_{i=1}^r \subset \M{d}$ and $P,Q:\C{d}\rightarrow \C{r}\otimes \C{d}$ are the respective Kraus and Stinespring operators. Then, $\Phi\in \DOC_d$ if and only if for every diagonal orthogonal matrix $O\in \mathcal{DO}_d$, there exists an invertible matrix $Z_{o}\in \M{r}$ such that the following equations hold $\,\,\forall i\in [r]$:
\begin{align*}
OP_iO &= \sum_{j=1}^r [Z_o]_{ij}P_j \in \operatorname{span}\{ P_1, P_2, \ldots, P_r \}, \\
OQ_iO &= \sum_{j=1}^r [Z_o^*]^{-1}_{ij} Q_j\in \operatorname{span}\{ Q_1, Q_2, \ldots, Q_r \}.
\end{align*}
\end{theorem}

\begin{proof}
    Let us first assume that $\Phi\in \T{d}$ is DOC. Then, it is clear from Definition~\ref{def:DUC-CDUC-DOC} that for every $O\in \mathcal{DO}_d$, we have $\Phi(OXO) = O\Phi(X)O$, i.e., $$ \sum_{i=1}^r (P_i O)X (Q_i O)^* = \sum_{i=1}^r (OP_i) X (OQ_i)^* .$$ A swift application of Lemma~\ref{ref:lemma-krausunique} then yields us the desired invertible matrix $Z_o\in \M{r}$. Conversely assume that such a $Z_o\in \M{r}$ exists for every diagonal $O\in \mathcal{DO}_d$. Then, a straightforward computation shows that the map $\Phi$ is DOC:
    \begin{align*}
        \Phi(OXO) &= \sum_{i=1}^r P_i OXO Q^*_i \\
        &= \sum_{i=1}^r \left\{ \sum_{j=1}^r [Z_o]_{ij} OP_j \right\} X \left\{ \sum_{k=1}^r [Z_o]^{-1}_{ki} Q_k^* O \right\}  \\
        &= O \left[ \sum_{j,k=1}^r P_j X Q^*_k \left\{ \sum_{i=1}^r [Z_o]^{-1}_{ki}[Z_o]_{ij} \right\} \right] O \\
        &= O \left\{ \sum_{k=1}^r P_k X Q^*_k \right\} O = O\Phi(X)O .
    \end{align*}
    \end{proof}

\begin{theorem} \label{theorem:CP-DOCkraus}
Let $\Phi\in \T{d}$ be completely positive with $r = \operatorname{rank}(J(\Phi))$, such that 
\begin{equation*}
    \Phi (X) = \sum_{i=1}^r P_i X P^*_i = \operatorname{Tr}_{\C{r}}(PXP^*),
\end{equation*}
where $\{P_i \}_{i=1}^r \subset \M{d}$ and $P:\C{d}\rightarrow \C{r}\otimes \C{d}$ are the respective Kraus and Stinespring operators. Then, $\Phi\in \DOC_d$ if and only if for every diagonal orthogonal matrix $O\in \mathcal{DO}_d$, there exists a unitary matrix $U_{o}\in \M{r}$ such that the following equation holds $\, \forall i\in [r]$:
\begin{align*}
OP_iO = \sum_{j=1}^r [U_o]_{ij}P_j \in \operatorname{span}\{ P_1, P_2, \ldots, P_r \} .
\end{align*}
\end{theorem}
\begin{proof}
    Identical to that of Theorem~\ref{theorem:DOCkraus}.
\end{proof}

\begin{remark}
The analogues of Theorem~\ref{theorem:DOCkraus} and \ref{theorem:CP-DOCkraus} for maps in $\DUC_d$ and $\CDUC_d$ are identical in structure, with the only difference being that the conjugations with diagonal orthogonal matrices $OP_i O, OQ_i O$ get replaced by conjugations with diagonal unitaries $U\in \mathcal{DU}_d$ instead: $UP_i U, U^* Q_i U^*$ and $U^*P_i U, U^*Q_i U$ (for $\DUC_d$ and $\CDUC_d$ respectively).
\end{remark}

\section{DOC maps and triplewise complete positivity}\label{sec:tcp}

In this section, we provide an alternate characterization of the family of diagonal orthogonal covariant maps in terms of invariant subspaces, which will be used to derive necessary and sufficient conditions for triplewise complete positivity of matrix triples in $\MLDOI{d}$ (or equivalently, for the separability of matrices in $\LDOI_d$). We will employ this characterization to provide an example of a non-TCP triple $(A,B,C)\in \MLDOI{d}$ such that both $(A,B)$ and $(A,C)$ are PCP. Recall that if $B=C$, this is not possible, since a triple of the form $(A,B,B)$ is TCP if and only if the pair $(A,B)$ is PCP, as was shown in Proposition~\ref{prop:AB-ABB}. In this process, we explicitly compute the partial action of a map $\Phi\in \DOC_d$ on a matrix $X\in \LDOI_d$, which is then connected to the operation of map composition in $\DOC_d$. Without further delay, we delve straight into the promised alternate characterization of the set $\DOC_d$.

\begin{proposition} \label{prop:DOC-LDOIinv}
Consider a linear map $\Phi\in \T{d}$. Then $\Phi\in \DOC_d$ if and only if the vector subspace $\LDOI_d \subset \M{d}\otimes \M{d}$ stays invariant under the linear map $\Phi\otimes \operatorname{id}: \M{d}\otimes \M{d} \rightarrow \M{d}\otimes \M{d}$ (or $\operatorname{id}\otimes \Phi: \M{d}\otimes \M{d} \rightarrow \M{d}\otimes \M{d}$).
\end{proposition}

\begin{proof}
For an arbitrary $O\in \mathcal{DO}_d, \, \Phi\in \T{d}$ and $X\in \LDOI_d$, it is evident that $(\Phi\otimes \operatorname{id})(X)\in \LDOI_d$ if and only if the equality in the diagram given below holds, which is clearly equivalent to the condition that $J(\Phi)\in \LDOI_d$ or $\Phi\in \DOC_d$. The case with the map $\operatorname{id}\otimes \Phi$ can be proven similarly.
    \begin{figure}[htbp]
        \centering 
        \includegraphics{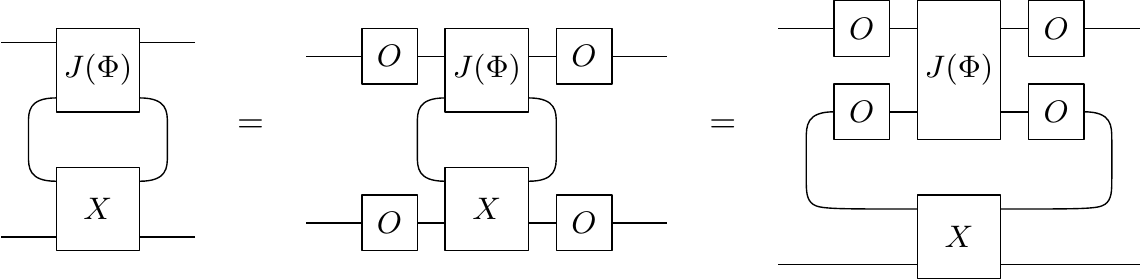}
        \caption{Invariance of the vector subspace $\LDOI_d \subset \M{d}\otimes \M{d}$ under the partial action of a map $\Phi\in \DOC_d$.}
    \end{figure}
\end{proof}

\begin{remark}
The analogue of Proposition~\ref{prop:DOC-LDOIinv} for maps in $\DUC_d$ (resp. $\CDUC_d$) can be derived similarly, with the only difference being a change in the invariant subspace from $\LDOI_d$ to $\LDUI_d$ (resp. $\CLDUI_d$).
\end{remark}

\begin{lemma} \label{lemma:DOC-LDOI-composition}
Consider the bilinear composition $\circ$ defined on the space $\MLDOI{d}$ as follows
\begin{alignat}{2}
\circ : \qquad \MLDOI{d} \,\,\, &\times \,\,\, \MLDOI{d} &&\rightarrow \MLDOI{d} \nonumber \\
\{ (A_1,B_1,C_1) &, (A_2,B_2,C_2) \} &&\mapsto (\mathfrak{A},\mathfrak{B},\mathfrak{C}), \nonumber
\end{alignat}
where 
\begin{align*}
    \mathfrak{A} = A_1 A_2, \qquad \mathfrak{B} &= B_1 \odot B_2 + C_1\odot C_2^\top + \operatorname{diag}(A_1 A_2 - 2A_1\odot A_2), \\
    and \qquad \mathfrak{C} &= B_1\odot C_2 + C_1\odot B_2^\top + \operatorname{diag}(A_1 A_2 - 2A_1\odot A_2).
\end{align*}
Then, for $(A_1,B_1,C_1), (A_2,B_2,C_2)\in \MLDOI{d}$, the following holds true:
\begin{equation}
    [\Phi^{(3)}_{(A_1,B_1,C_1)} \otimes \operatorname{id}](X^{(3)}_{(A_2,B_2,C_2)}) = X^{(3)}_{(\mathfrak{A},\mathfrak{B},\mathfrak{C})}. \nonumber
\end{equation}
\end{lemma}

\begin{proof}
    We wish to explicitly compute the following action:
\begin{align}
    \Phi^{(3)}_{(A_1,B_1,C_1)} \otimes \operatorname{id} : \qquad \LDOI_d &\rightarrow \LDOI_d \nonumber \\
    X^{(3)}_{(A_2,B_2,C_2)} &\mapsto X^{(3)}_{(\mathfrak{A},\mathfrak{B},\mathfrak{C})}.
\end{align}
Proceeding diagrammatically, it is clear that 
\begin{equation}
    [\Phi^{(3)}_{(A_1,B_1,C_1)} \otimes \operatorname{id}](X^{(3)}_{(A_2,B_2,C_2)}) = \includegraphics[align=c, scale=1.1]{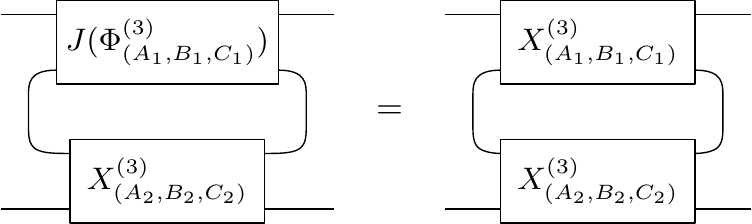}, \label{eq:partial-action}
\end{equation}
where the equalities follow from Eq.~\eqref{eq:phi-action} and Remark~\ref{remark:choi(DOC)}. By exploiting the isomorphism from Proposition~\ref{prop:LDOI-ABC}, we can express the above diagram as in Figure~\ref{fig:partial-action-1}.
\begin{figure}[H]
    \centering
    \includegraphics[scale=1.1]{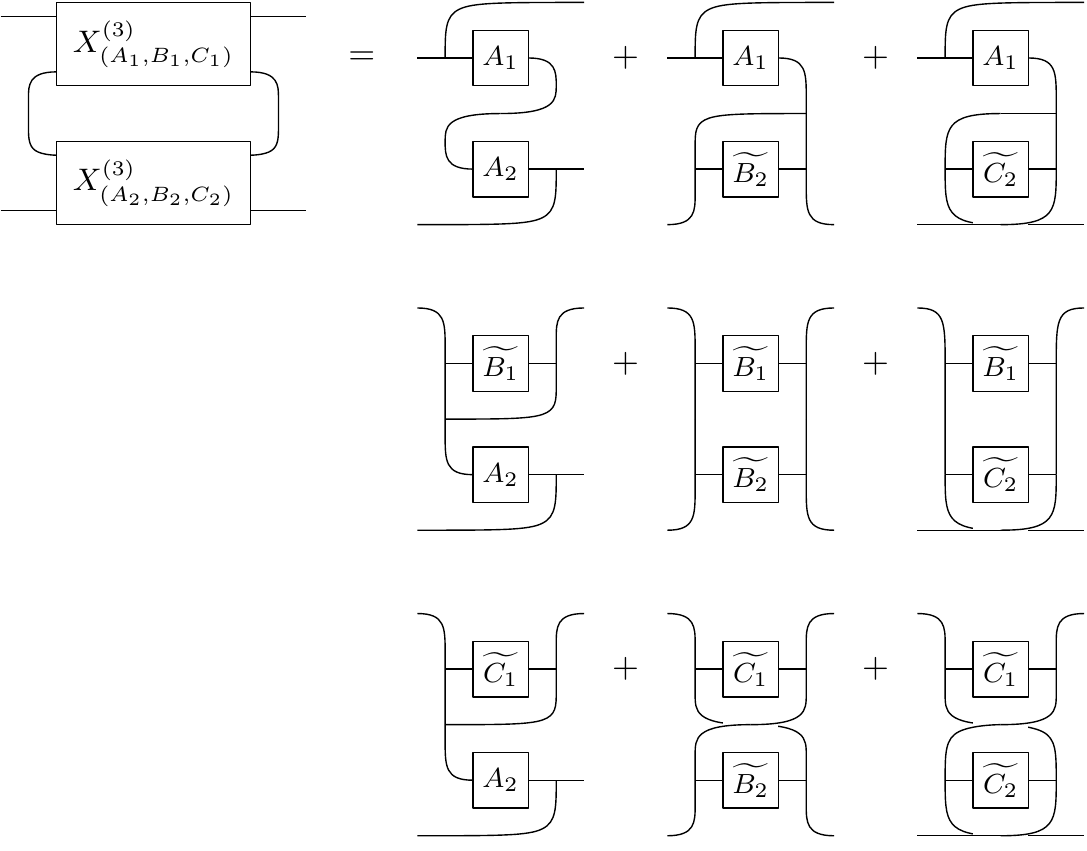}
    \caption{Expansion of Eq.~\eqref{eq:partial-action} using the form of the $X^{(3)}$ matrices from Proposition~\ref{prop:LDOI-ABC}.}
    \label{fig:partial-action-1}
\end{figure}

Notice that, as before, $\widetilde{A_i}, \widetilde{B_i}$ and $\widetilde{C_i}$ are matrices with the same off-diagonal entries as $A_i, B_i$ and $C_i$ respectively, but with $\operatorname{diag}(\widetilde{A_i})=\operatorname{diag}(\widetilde{B_i})= \operatorname{diag}(\widetilde{C_i})= 0$, for $i=1,2$. This leads us to the final expression in Figure~\ref{fig:partial-action-2}, where $\mathfrak{A} = A_1 A_2$, $\mathfrak{B} = \operatorname{diag}(A_1 A_2) + \widetilde{B_1}\odot \widetilde{B_2} + \widetilde{C_1}\odot \widetilde{C_2}^\top$ and $\mathfrak{C} = \operatorname{diag}(A_1 A_2) + \widetilde{B_1}\odot \widetilde{C_2} + \widetilde{C_1}\odot~\widetilde{B_2}^\top$. The proof is now complete.

\begin{figure}[H]
    \centering
    \includegraphics[scale=1.1]{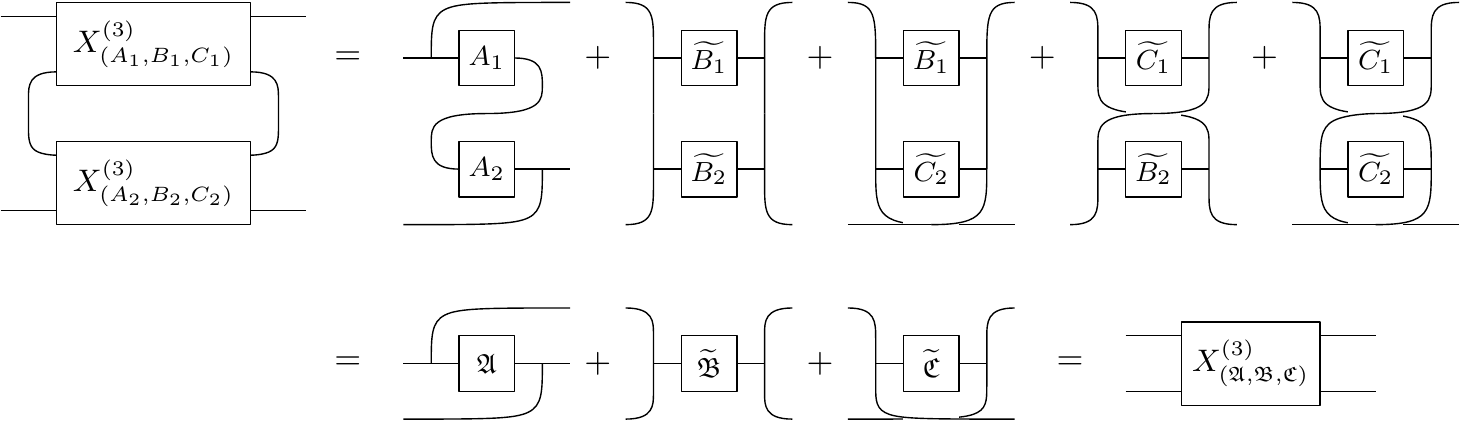}
    \caption{Cancelling diagrams from Figure~\ref{fig:partial-action-1} using $\operatorname{diag}(\widetilde{A_i})=\operatorname{diag}(\widetilde{B_i})= \operatorname{diag}(\widetilde{C_i})= 0$.}
    \label{fig:partial-action-2}
\end{figure}
\end{proof}

With the composition rule from Lemma~\ref{lemma:DOC-LDOI-composition} in hand, let us now consider two particular instances of it, for matrix pairs $(A_1, B_1), (A_2, B_2) \in \MLDUI{d}$:
\begin{equation} \label{eq:DUC-composition}
    (A_1,\operatorname{diag}(A_1), B_1)\circ (A_2,\operatorname{diag}(A_2), B_2) = (A_1 A_2, B_1 \odot B_2^\top + \operatorname{diag}(A_1 A_2 - B_1\odot B_2), \operatorname{diag}(A_1 A_2))
\end{equation}
\begin{equation} \label{eq:CDUC-composition}
    (A_1, B_1, \operatorname{diag}(A_1))\circ (A_2, B_2, \operatorname{diag}(A_2)) = (A_1 A_2, B_1 \odot B_2 + \operatorname{diag}(A_1 A_2 - B_1\odot B_2), \operatorname{diag}(A_1 A_2))
\end{equation}

The following definition formulates these new rules in a more succinct fashion.

\begin{definition}\label{def:DUC-CDUC-composition}
    On $\MLDUI{d}$, define bilinear compositions $\circ_1$ and $\circ_2$ as follows:
\begin{alignat}{2}
\circ_1 : \qquad \MLDUI{d} \,\,\, &\times \,\,\, \MLDUI{d} &&\rightarrow \MLDUI{d} \nonumber \\
\{ (A_1,B_1) &, (A_2,B_2) \} &&\mapsto (\mathfrak{A},\mathfrak{B}) = (A_1 A_2, B_1 \odot B_2^\top + \operatorname{diag}(A_1 A_2 - B_1\odot B_2)), \nonumber
\end{alignat}
\begin{alignat}{2}
\circ_2 : \qquad \MLDUI{d} \,\,\, &\times \,\,\, \MLDUI{d} &&\rightarrow \MLDUI{d} \nonumber \\
\{ (A_1,B_1) &, (A_2,B_2) \} &&\mapsto (\mathfrak{A},\mathfrak{B}) = (A_1 A_2, B_1\odot B_2 + \operatorname{diag}(A_1 A_2 - B_1\odot B_2)). \nonumber
\end{alignat}
\end{definition}

\begin{remark} \label{remark:DUC-CDUC-composition}
It is obvious from the above definition that 
$$(A_1, B_1)\circ_1 (A_2, B_2) = (A_1, B_1)\circ_2 (A_2, B_2^\top) \qquad \forall (A_1, B_1), (A_2, B_2) \in \MLDUI{d}. $$
\end{remark}

Next, we state and prove an important proposition, which connects all the composition rules on matrix pairs/triples introduced so far to the operations of map composition in $\DUC_d, \CDUC_d$ and $\DOC_d$. But first, we need familiarity with the notion of \emph{stability} under composition.
\begin{definition}
    A set $K\subseteq \T{d}$ is said to be \emph{stable} under composition if $\Phi_1\circ \Phi_2 \in K \, \forall \, \Phi_1, \Phi_2 \in K$.
\end{definition}

\begin{lemma} \label{lemma:DUC-CDUC-DOC-composition}
The linear subspaces $\CDUC_d, \DOC_d \subset \T{d}$ are stable under composition, but $\DUC_d \subset \T{d}$ is not. Moreover, for triples $(A_1,B_1,C_1), (A_2,B_2,C_2)\in \MLDOI{d}$, the following composition rules hold (see Eqs.~\eqref{eq:DUC-action},\eqref{eq:CDUC-action},\eqref{eq:DOC-action}): 
\begin{itemize}
    \item $\Phi^{(i)}_{(A_1,B_1)} \circ \Phi^{(j)}_{(A_2,B_2)} = \begin{cases}
    \, \Phi^{(2)}_{(\mathfrak{A},\mathfrak{B})},  &\text{where } (\mathfrak{A},\mathfrak{B}) = (A_1,B_1)\circ_i (A_2,B_2) \quad \text{if } i=j \\
    \, \Phi^{(1)}_{(\mathfrak{A},\mathfrak{B})},  &\text{where } (\mathfrak{A},\mathfrak{B}) = (A_1,B_1)\circ_i (A_2,B_2) \quad \text{if } i\neq j 
    \end{cases} $
    \item $\Phi^{(3)}_{(A_1,B_1,C_1)} \circ \Phi^{(3)}_{(A_2,B_2,C_2)} = \Phi^{(3)}_{(\mathfrak{A},\mathfrak{B},\mathfrak{C})}$, where $(\mathfrak{A},\mathfrak{B},\mathfrak{C}) = (A_1,B_1,C_1)\circ (A_2,B_2,C_2)$.
\end{itemize}
\end{lemma}

\begin{proof}
    The stability results follow directly from Definition~\ref{def:DUC-CDUC-DOC}. It is also trivial to check that if $\Phi_1\in \DUC_d$ and $\Phi_2\in \CDUC_d$ (or vice-versa), then $\Phi_1\circ \Phi_2 \in \DUC_d$, since $\forall \, U\in \mathcal{DU}_d$ and $Z\in \M{d}$, the following equation holds: $[\Phi_1\circ \Phi_2](UZU^*) = \Phi_1 [U \Phi_2(Z) U^*] = U^* [\Phi_1\circ \Phi_2 (Z)] U$.
    
    Now, since $\DOC_d$ is stable under composition, we know that $\Phi_1 \circ \Phi_2 \in \DOC_d \, \forall \, \Phi_1, \Phi_2\in \DOC_d$. Hence, let $\Phi^{(3)}_{(A_1,B_1,C_1)} \circ\, \Phi^{(3)}_{(A_2,B_2,C_2)} =\Phi^{(3)}_{(\widetilde{\mathfrak{A}}, \widetilde{\mathfrak{B}}, \widetilde{\mathfrak{C}})} \in \DOC_d$. Now, Lemma~\ref{lemma:DOC-LDOI-composition} tells us that 
    \begin{align}
        &[\Phi^{(3)}_{(A_1,B_1,C_1)} \otimes \operatorname{id}](X^{(3)}_{(A_2,B_2,C_2)}) = X^{(3)}_{(\mathfrak{A},\mathfrak{B},\mathfrak{C})} \nonumber \\
        \implies &[\Phi^{(3)}_{(A_1,B_1,C_1)} \otimes \operatorname{id}] \circ [\Phi^{(3)}_{(A_2,B_2,C_2)} \otimes \operatorname{id}] (\Omega) = X^{(3)}_{(\mathfrak{A},\mathfrak{B},\mathfrak{C})} \nonumber \\
        \implies &[\Phi^{(3)}_{(A_1,B_1,C_1)} \circ \Phi^{(3)}_{(A_2,B_2,C_2)} \otimes \operatorname{id}](\Omega) = X^{(3)}_{(\mathfrak{A},\mathfrak{B},\mathfrak{C})} \nonumber \\
        \implies &[\Phi^{(3)}_{(\widetilde{\mathfrak{A}}, \widetilde{\mathfrak{B}}, \widetilde{\mathfrak{C}})} \otimes \operatorname{id}](\Omega) = X^{(3)}_{(\mathfrak{A},\mathfrak{B},\mathfrak{C})} \nonumber \\
        \implies &(\widetilde{\mathfrak{A}}, \widetilde{\mathfrak{B}}, \widetilde{\mathfrak{C}}) = (\mathfrak{A},\mathfrak{B},\mathfrak{C}),
    \end{align}
    where $(\mathfrak{A},\mathfrak{B},\mathfrak{C}) = (A_1,B_1,C_1)\circ (A_2,B_2,C_2)$. Notice that the Choi-Jamio{\l}kowski isomorphism for DOC maps was implemented in obtaining the first and last implications above: $$\forall \, (A,B,C)\in \MLDOI{d}: \qquad [\Phi^{(3)}_{(A,B,C)} \otimes \operatorname{id}](\Omega) = X^{(3)}_{(A,B,C)},$$ where $\Omega = |\psi \rangle\langle \psi| \in \M{d}$ and $\ket{\psi} = \sum_{i=1}^d \ket{ii}$ is the maximally entangled vector in $\mathbb{C}^d \otimes \mathbb{C}^d$.
    
    For the remaining results, we first infer from Remark~\ref{remark:(C)LDUIsubspaceLDOI} that  
    $$\forall \, (A,B)\in \MLDUI{d}: \qquad \Phi^{(1)}_{(A,B)} = \Phi^{(3)}_{(A,\operatorname{diag}A,B)}, \qquad \Phi^{(2)}_{(A,B)} = \Phi^{(3)}_{(A,B,\operatorname{diag}A)}.$$ Then, an amalgamation of the recently proved result and Eqs.~\eqref{eq:DUC-composition}, \eqref{eq:CDUC-composition} immediately yields the desired composition rules.
\end{proof}

The composition rule from Lemma~\ref{lemma:DOC-LDOI-composition} allows us to construct necessary and sufficient conditions on a triple $(A,B,C)\in \MLDOI{d}$ which guarantee that it is triplewise completely positive -- these are presented in Theorem~\ref{theorem:DOC-LDOI-sep} below. The reader is advised to keep the discussion from Section~\ref{sec:convex-structure} in mind before proceeding further.

\begin{theorem}\label{theorem:DOC-LDOI-sep}
Consider a matrix $X\in \LDOI_d$. Then, the following equivalent statements hold:
\begin{itemize}
    \item $X$ is separable if and only if $(\Phi\otimes \operatorname{id})(X)\in \LDOI_d$ is positive semi-definite for all positive maps $\Phi\in \DOC_d$.
    \item $(A,B,C)\in \MLDOI{d}$ with $X=X^{(3)}_{(A,B,C)}$ is triplewise completely positive if and only if $(\mathfrak{A},\mathfrak{B},\mathfrak{C}) = (D,E,F)\circ (A,B,C)$ corresponds to a positive semi-definite $X^{(3)}_{(\mathfrak{A},\mathfrak{B},\mathfrak{C})}\in \LDOI_d$ for all $(D,E,F) \in \MLDOI{d}$ such that $\Phi^{(3)}_{(D,E,F)}\in \DOC_d$ is positive.
\end{itemize}
\end{theorem}

\begin{proof}
    For the fist part, we observe that the following isomorphisms can be established from the discussion in Section~\ref{sec:convex-structure} and the Choi-Jamio{\l}kowski isomorphism (Lemma~\ref{lemma:CJiso}(2)):
    \begin{align}
        \LDOI_d^{\mathsf{BP}} &\simeq \{ (A,B,C)\in \MLDOI{d} \, \big| \, X^{(3)}_{(A,B,C)}\in \LDOI_d \text{ is block positive} \} \nonumber \\ 
        &\simeq \{ \Phi\in \DOC_d \, \big| \, \Phi \text{ is positive} \}.
    \end{align}
    Now, assume that $(\Phi \otimes \operatorname{id})(X)$ is positive semi-definite for all positive $\Phi\in \DOC_d$. Let $\ket{\psi}\coloneqq \sum_{i=1}^d \ket{ii}$ be the maximally entangled vector in $\mathbb{C}^d \otimes \mathbb{C}^d$ and let $\Omega = |\psi \rangle\langle \psi | \in \M{d}$.
    So, we have
    \begin{alignat}{2}
        &\operatorname{Tr}[(\Phi\otimes \operatorname{id})(X) \Omega] \geq 0 \qquad &&\forall \text{ positive }\Phi\in \DOC_d \\
        \implies &\operatorname{Tr}[X (\Phi^* \otimes \operatorname{id})(\Omega)] \geq 0 \qquad &&\forall \text{ positive }\Phi\in \DOC_d \\
        \implies &\operatorname{Tr}[X J(\Phi^*)] \geq 0 \qquad &&\forall \text{ positive }\Phi\in \DOC_d \\
        \implies &\operatorname{Tr}[XY] \geq 0 \qquad &&\forall \, Y\in \LDOI_d^{\mathsf{BP}},
    \end{alignat}
    which shows that $X$ is separable. The other direction of the proof is trivial.
    
    The second part follows directly from the first, since we know that $X^{(3)}_{(A,B,C)}\in \LDOI_d$ is separable if and only if $(A,B,C)\in \MLDOI{d}$ is TCP (see Theorem~\ref{theorem:LDOI-sep}), and for all $(A,B,C), (D,E,F)$ in $\MLDOI{d}$, we have $[\Phi^{(3)}_{D,E,F}\otimes \operatorname{id}](X^{(3)}_{(A,B,C)}) = X^{(3)}_{(\mathfrak{A},\mathfrak{B},\mathfrak{C})}$, where 
    $ (\mathfrak{A},\mathfrak{B},\mathfrak{C}) = (D,E,F)\circ (A,B,C)$ (see Lemma~\ref{lemma:DOC-LDOI-composition}).
\end{proof}

The simple necessary conditions that follow from triplewise completele positivity of $(A,B,C)\in \MLDOI{d}$ (see Lemma~\ref{lemma:tcp-properties}) can be easily derived with the help of the previous Theorem.

\begin{corollary} \label{corollary-tcp-psd-ppt}
If $(A,B,C)\in \MLDOI{d}$ is \emph{TCP}, then $A\in \EWP_d$, $B,C\in \PSD_d$ and $A_{ij}A_{ji} \geq \vert B_{ij} \vert^2, \, A_{ij}A_{ji} \geq \vert C_{ij} \vert^2 \,\, \forall i,j \in [d]$. Equivalently, $X^{(3)}_{(A,B,C)} \in \LDOI_3 $ is \emph{PPT}.
\end{corollary}
\begin{proof}
    Choose the (obviously positive) identity and transposition maps in $\DOC_d$ from Example~\ref{eg:maps-id-transp} and apply Theorem~\ref{theorem:DOC-LDOI-sep} to obtain the desired result. 
\end{proof}

Although the conditions in Corollary \ref{corollary-tcp-psd-ppt} are necessary for triplewise complete positivity, the following example elucidates that they are not sufficient, which is equivalent to the fact that PPT entangled matrices exist in $\LDOI_d$.
\begin{example} 
We consider the one-parameter family of matrices in $\M{3}$ from \cite{stormer1982}:
\begin{equation}
A(\mu) = \left(   \begin{array}{ccc}
        2\mu & 4\mu^2 & 1  \\
        1 & 2\mu & 4\mu^2 \\
        4\mu^2 & 1 & 2\mu
    \end{array}  \right), \qquad B(\mu) = \left(   \begin{array}{rrr}
        2\mu & 2\mu & 2\mu  \\
        2\mu & 2\mu & 2\mu \\
        2\mu & 2\mu & 2\mu
    \end{array}  \right).
\end{equation}
It can be easily checked that $\{(A(\mu), B(\mu), \operatorname{diag}A(\mu) ) : \mu\geq 1\}\subset \MLDOI{3}$ satisfies the conditions of Corollary~\ref{corollary-tcp-psd-ppt}. Now, consider the family of Choi-type maps $\Phi^{(2)}_{(X(\mu),Y(\mu))} = \Phi_{(1,\mu,0)}^I \in \mathcal{T}_3(\mathbb{C})$ from Example~\ref{eg:maps-choi}, where
\begin{equation}
    X(\mu) = \left(   \begin{array}{rrr}
        1 & \mu & 0  \\
        0 & 1 & \mu \\
        \mu & 0 & 1
    \end{array}  \right), \qquad Y(\mu) = \left(   \begin{array}{rrr}
        1 & -1 & -1  \\
        -1 & 1 & -1 \\
        -1 & -1 & 1
    \end{array}  \right).
\end{equation}
It was shown in \cite{cho1992choi, choi1980pos} that these maps are positive for $\mu \geq 1$. If we let 
$$(\mathfrak{A}(\mu),\mathfrak{B}(\mu),\mathfrak{C}(\mu)) = (X(\mu), Y(\mu), \operatorname{diag}X(\mu))\circ(A(\mu), B(\mu), \operatorname{diag}A(\mu)),$$ then it is fairly easy to see that
\begin{equation}
    \mathfrak{B}(\mu) = \left(  \begin{array}{rrr}
        3\mu & -2\mu & -2\mu  \\
        -2\mu & 3\mu & -2\mu \\
        -2\mu & -2\mu & 3\mu
    \end{array}  \right)
\end{equation}
is not positive semi-definite for $\mu\geq 1$. Hence, $\{( A(\mu), B(\mu), \operatorname{diag}A(\mu) ) : \mu\geq 1\}$ is a non-TCP family of matrix triples in $\MLDOI{3}$, see Theorem~\ref{theorem:DOC-LDOI-sep}. Speaking in terms of bipartite matrices,
\begin{equation}
    [\Phi^{(2)}_{(X(\mu),Y(\mu))} \otimes \operatorname{id}](X^{(2)}_{(A(\mu), B(\mu))}) \notin \PSD_{d^2}
\end{equation}
and hence $X^{(2)}_{(A(\mu), B(\mu))} \in \CLDUI_d\subset \LDOI_d$ is PPT entangled, with the positive non-decomposable map $\Phi^{(2)}_{(X(\mu), Y(\mu))}\in \CDUC_d \subset \DOC_d$ detecting it for all $\mu \geq 1$.
\end{example}

While discussing PT invariant LDOI matrices in Example~\ref{eg:states-PTinv}, we saw that a matrix triple of the form $(A,B,B) \in \MLDOI{d}$ is TCP if and only if the pair $(A,B)\in \MLDUI{d}$ is PCP. We now show that this does not generalize to the case of arbitrary triples $(A,B,C)$ with $B\neq C$: $\exists (A,B,C) \in \MLDOI{d}$ such that both $(A,B)$ and $(A,C)$ are PCP but $(A,B,C)$ is still not TCP.

\begin{example}
Consider the matrix triple $(A,B,C) \in \MLDOI{3}$:
\begin{equation}
    A = \left(  \begin{array}{ccc}
        1 & 0 & 1  \\
        0 & 1 & 1 \\
        1 & 1 & 1
    \end{array}  \right), \quad B = \left(  \begin{array}{rrr}
        1 & 0 & -1  \\
        0 & 1 & 0 \\
        -1 & 0 & 1
    \end{array}  \right), \quad C = \left(  \begin{array}{rrr}
        1 & 0 & 0  \\
        0 & 1 & -1 \\
        0 & -1 & 1
    \end{array}  \right).
\end{equation}
It is straightforward to verify that $(A,B,C)$ satisfies the conditions of Corollary \ref{corollary-tcp-psd-ppt}. Moreover, since $B$ and $C$ are diagonally dominant, Lemma~\ref{lemma:PCP-comparison} shows that $(A,B)$ and $(A,C)$ are PCP. Now, consider the triple $(X,Y,Z) \in \MLDOI{3}$ from Example ~\ref{eg:maps-lambda} associated with the positive map $\Phi^{(3)}_{(X,Y,Z)} = \Lambda_3 \in \DOC_3$ and let $(\mathfrak{A},\mathfrak{B},\mathfrak{C}) = (X, Y, Z)\circ(A, B, C)$. Then,
\begin{equation}
    \mathfrak{B} = \frac{1}{2}\left(  \begin{array}{rrr}
        1 & 0 & -\sqrt{2}  \\
        0 & 1 & -\sqrt{2} \\
        -\sqrt{2} & -\sqrt{2} & 2
    \end{array}  \right)
\end{equation}
is not positive semi-definite, thus proving that $(A,B,C)$ is not TCP, see Theorem \ref{theorem:DOC-LDOI-sep}. It is also not too difficult to verify that while the triple $(A,B,C)$ does meet the conditions stated in part (4) of Lemma~\ref{lemma:tcp-properties}, it violates the condition in part (5) of the same Lemma. In the realm of bipartite matrices, the above discussion translates to the fact that $X^{(i)}_{(A,B)}, X^{(i)}_{(A,C)}\in \LDUI_d$ (resp.~$\CLDUI_d$) are separable for $i=1$ (resp.~$i=2$) but $X^{(3)}_{(A,B,C)}\in \LDOI_d$ is PPT entangled.
\end{example}

\section{Conclusions and future directions}
We have presented an elaborate study of the family of local diagonal unitary and orthogonal invariant bipartite matrices ($\LDUI_d$, $\CLDUI_d$, and $\LDOI_d$), along with the accompanying class of diagonal unitary and orthogonal covariant maps between matrix algebras ($\DUC_d$, $\CDUC_d$, and $\DOC_d$). By easing the analysis of several important properties of objects in these classes, the isomorphisms with the family of matrix pairs and triples with equal diagonals ($\MLDUI{d}$ and $\MLDOI{d}$) play an instrumental role in our endeavors. In particular, we show that the cone of separable LDOI matrices admits an equivalent description in terms of the cone of triplewise completely positive matrices, which generalizes the well-studied cone of completely positive matrices. We entirely determine the extreme rays of these cones, along with the cone of positive semi-definite LDOI matrices. We also spend considerable time on describing the linear structure of the vector space $\LDOI_d$. For linear maps in $\DOC_d$, several equivalent characterizations are presented based on their Choi-Jamio{\l}kowski, Kraus, and Stinespring representations. The familiar properties of positivity, decomposability, complete positivity, PPT, and the like, are dealt with in detail. Our investigations into the invariant subspaces of these maps reveal key connections between the cones of positive DOC maps and separable LDOI matrices. Finally, comprehensive lists of important examples --- both of LDOI matrices and DOC maps --- are exhibited and discussed at length; these cover the existent literature, as well as many new important examples.

Ever since Choi discovered the first example of a positive non-decomposable map in the '70s, there has been immense interest in studying its generalizations, especially after the relatively recent Entanglement Theory associations were unraveled. We have seen that all Choi-type maps are a particular example of a much broader class defined by a unique covariance property. One of the merits of our work is to provide a unifying framework for the study of these maps, leveraging tools from linear and multi-linear algebra, and convex geometry, to obtain powerful characterizations of the relevant properties of these maps. 

Several research prospects stem from our work. The membership problem in the cone of triplewise completely positive matrices (or equivalently, in the cone of separable LDOI matrices) is the most glaring one. Simple and easily verifiable sufficient conditions to guarantee that a matrix triple is or is not TCP are desirable. A significant attempt in this direction is made in \cite{Singh2020entanglement}, where crucial graph-theoretic techniques are implemented to explore a new variety of entanglement in both LDOI and arbitrary bipartite matrices. Other entanglement-theoretic properties of positive semi-definite LDOI matrices (like entanglement of formation, distillation, cost, and concurrence, to name a few) deserve further scrutiny. The cones of positive/decomposable linear maps between matrix algebras have evaded simple characterizations for quite some time now, which translates into similar difficulties while dealing with the intersection of these cones with the $\DOC_d$ subspace. Characterization of these cones' convex structure has the potential to provide new insights into the theory of entanglement.

PPT square conjecture \cite{PPTsq} posits that the composition of a PPT map in $\T{d}$ with itself is entanglement breaking. In \cite{singh2020ppt2}, we prove that this conjecture holds for DUC and CDUC maps, thus establishing its validity for a very large class and generalizing many known results from the literature. The analysis in \cite{singh2020ppt2} is based on the tools developed in the current paper, such as the composition formulas from Lemmas \ref{lemma:DOC-LDOI-composition} and \ref{lemma:DUC-CDUC-DOC-composition}, as well as the separability results for LDOI matrices from Theorem~\ref{theorem:LDOI-sep} and Lemma \ref{lemma:PCP-comparison}. 

\bigskip

\noindent\textit{Acknowledgements.} We thank M\={a}ris Ozols and an anonymous referee for the careful reading of our paper and for the many comments and remarks which helped improve the quality of the presentation. 

\bibliographystyle{plainurl}
\bibliography{references}

\bigskip

\hrule
\bigskip

\end{document}